\documentclass{amsart}

\usepackage{amsfonts, amssymb, amsmath, amsthm}                               
\usepackage{slashed}                                                          
\usepackage{graphicx}                                                         
\usepackage{color}                                                            
\usepackage{enumerate}                                                        
\usepackage{tikz}                                                             
\usetikzlibrary{decorations.pathmorphing, arrows}

\newtheorem{theorem}{Theorem} 	      	      	                              
\newtheorem{corollary}[theorem]{Corollary}     	      	      	      	      
\newtheorem{lemma}[theorem]{Lemma}     	       	      	      	      	      
\newtheorem{proposition}[theorem]{Proposition} 	      	      	      	      
\newtheorem{definition}[theorem]{Definition} 	      	      	                
\newtheorem*{remark}{Remark}                                                  

\numberwithin{equation}{section}                                              
\numberwithin{theorem}{section}                                               

\newcommand{\ul}[1]{\underline{#1}}                                           
\newcommand{\mf}[1]{\mathfrak{#1}}                                            
\newcommand{\mc}[1]{\mathcal{#1}}                                             

\newcommand{\N}{\mathbb{N}}                                                   
\newcommand{\R}{\mathbb{R}}                                                   
\newcommand{\Sph}{\mathbb{S}}                                                 

\newcommand{\grad}{\nabla^\sharp}                                             
\newcommand{\nasla}{\slashed{\nabla}}                                         

\newcommand{\tfr}{y}                                                          
\newcommand{\AdS}{{ \rm AdS }}                                                

\allowdisplaybreaks
\begin{document}

\title[Unique continuation]{Unique continuation from infinity in \\
asymptotically Anti-de Sitter spacetimes}

\author{Gustav Holzegel}
\address{Department of Mathematics\\
South Kensington Campus\\
Imperial College\\
London SW7 2AZ\\ United Kingdom}
\email{gholzege@imperial.ac.uk}

\author{Arick Shao}
\address{Department of Mathematics\\
South Kensington Campus\\
Imperial College\\
London SW7 2AZ\\ United Kingdom}
\email{c.shao@imperial.ac.uk}

\begin{abstract}
We consider the unique continuation properties of asymptotically Anti-de Sitter spacetimes by studying Klein-Gordon-type equations $\Box_g \phi + \sigma \phi = \mathcal{G} ( \phi, \partial \phi )$, $\sigma \in \mathbb{R}$, on a large class of such spacetimes.
Our main result establishes that if $\phi$ vanishes to sufficiently high order (depending on $\sigma$) on a sufficiently long time interval along the conformal boundary $\mathcal{I}$, then the solution necessarily vanishes in a neighborhood of $\mathcal{I}$.
In particular, in the $\sigma$-range where Dirichlet and Neumann conditions are possible on $\mathcal{I}$ for the forward problem, we prove uniqueness if \emph{both} these conditions are imposed.
The length of the time interval can be related to the refocusing time of null geodesics on these backgrounds and is expected to be sharp.
Some global applications as well a uniqueness result for gravitational perturbations are also discussed.
The proof is based on novel Carleman estimates established in this setting.
\end{abstract}

\maketitle

\section{Introduction} \label{sec.intro}

Asymptotically anti-de Sitter (aAdS) spacetimes are solutions $( \mathcal{M}, g )$ of the vacuum Einstein equations with negative cosmological constant $\Lambda < 0$,
\begin{equation} \label{ee}
\operatorname{Ric} [g] = \Lambda g \text{,}
\end{equation}
which asymptotically in space behave like the maximally symmetric solution of \eqref{ee}, anti-de Sitter (AdS) space. Such spacetimes play a dominant role in theoretical physics, mostly due to their putative relation to strongly coupled fields theories \cite{Malda}, and, more recently, phenomena in condensed matter theory \cite{Hartnoll}.

On the mathematical side, the rigorous study of the classical dynamical aspects of these spaces has only been initiated very recently.
Somewhat surprisingly, perhaps, the causal geometry of aAdS spacetimes (in particular, the existence of a timelike conformal boundary at infinity) already renders nontrivial the simplest conceivable hyperbolic problem in this context: that of establishing well-posedness for the massive linear wave equation on such a fixed aAdS background $( \mc{M}, g )$,
\begin{align} \label{weintro}
\Box_g \phi - \frac{\Lambda}{n} \sigma \phi = 0 \text{.}
\end{align}
Through the works of \cite{Bachelot, BF, Hol:wp, Vasy}, and most importantly \cite{Warnick:2012fi}, we now have a definite understanding of the boundary initial value problem associated with \eqref{weintro}, reviewed briefly below.
This theory has been extended to non-linear equations \cite{EncisoK}.
For \eqref{ee} itself we have the classical result of Friedrich \cite{Friedrich}; see also \cite{Enciso}.

While it is reassuring that basic hyperbolic equations on aAdS spacetimes admit well-posed formulations, some physicists prefer instead to entertain the following loosely formulated question:
\begin{itemize}
\item[] \emph{In what way is the trace on the boundary of a solution of \eqref{ee}  ``in correspondence" with the solution in the interior?} 
\end{itemize}

The purpose of this paper is to develop the analytical techniques leading to both a precise formulation and a satisfactory answer to the above question.
We begin here by treating wave equations on \emph{fixed} aAdS backgrounds before turning to the full non-linear Einstein equations (\ref{ee}) in \cite{ASfollowup}.

We emphasize that while the background metric will be fixed here, our results will hold for a large class of asymptotically AdS spacetimes\footnote{See Definitions \ref{def.aads_manifold} and \ref{def.aads}.} (of any dimension), and for a general class of tensorial\footnote{See Section \ref{sec.aads_tensor}.} and even non-linear wave equations \eqref{we}; see the discussion in Section \ref{sec.intro_gen}.
However, for the present introductory discussion, the reader may restrict attention to linear waves on $(3+1)$-dimensional ``pure" AdS spacetime $( \mathbb{R}^4, g_{\rm AdS} )$, whose metric in global coordinates $( t, r, \theta, \phi )$ takes the form
\begin{equation} \label{ads}
g_{\rm AdS} = - \left( 1 - \frac{\Lambda}{3} r^2 \right) dt^2 + \left( 1 - \frac{\Lambda}{3} r^2 \right)^{-1} dr^2 + r^2 ( d \theta^2 + \sin^2 \theta d \phi^2 ) \text{,}
\end{equation}
where for convenience, we will set $\Lambda = -3$ below.

\subsection{Unique Continuation and Pseudoconvexity}

It is well-known that the AdS metric \eqref{ads} is conformally equivalent to half of the Einstein static cylinder $( \R \times \Sph^3_+, g_E )$, where $g_E = \Omega^2 g_{\rm AdS} = -dt^2 + d\Sigma^2$, with 
\[ d \Sigma^2 = d \chi^2 + \sin^2 \chi (d \theta^2 + \sin^2 \theta d \varphi^2 ) \]
the standard round metric on a hemisphere $\Sph^3_+$ of $\Sph^3$.
In the case of conformal mass, $\sigma = 2$ in \eqref{weintro}, the rescaled field $\psi = \Omega^{-1} \phi$ satisfies the wave equation $\Box_{g_E} \psi - \psi = 0$ on $( \mathbb{R} \times \Sph^3_+, g_E )$, allowing one to transform the problem to a regular initial boundary value problem on a bounded domain.\footnote{For $\sigma \neq 2$, one can also transform the problem to a bounded domain at the cost of introducing divergent (towards the boundary) potentials in the transformed wave equation. These divergent terms are, however, absent for $\sigma = 2$.}

In this formulation, it is quite clear from the analogy with a timelike hypersurface in Minkowski spacetime (examples going back to Hadamard \cite{Hadamard}) that the problem arising from specifying Cauchy data for \eqref{weintro} on the timelike boundary of AdS will generally not be well-posed.
In particular, the solution, if it exists (it does, for instance, for analytic data), may not depend continuously on the data prescribed.

In spite of this, one can still hope for (local) \emph{unique continuation}, which can be formulated in a general manner as follows:
\begin{itemize}
\item[] \emph{Does Cauchy data on a boundary hypersurface determine the solution---if it exists---of a PDE uniquely in a neighbourhood of (one side of) the boundary.}
\end{itemize}
In the setting of linear equations, this can be equivalently stated as:
\begin{itemize}
\item[] \emph{Does zero Cauchy data on a boundary hypersurface imply that the solution of a PDE must vanish in a neighbourhood of (one side of) the boundary.}
\end{itemize}
With regards to the above discussion, the PDE under consideration is \eqref{weintro}, perhaps with additional lower-order terms present, and the ``Cauchy data" is imposed on AdS infinity.
The precise definition of ``Cauchy data" at conformal infinity is naturally suggested by the forward well-posedness theory, which, as explained below, introduces a notion of Dirichlet and Neumann data at infinity.
Ideally, we wish to show that any solution having \emph{both} vanishing Dirichlet and Neumann data at infinity must necessarily vanish in the interior as well.

\subsubsection{Classical Methods of Analysis}

There exist at least two well-known techniques to prove such uniqueness results.
The first is via Holmgren's theorem, which implies uniqueness in the class of distributions whenever the PDE is both linear and analytic and the boundary hypersurface is noncharacteristic.
In particular, this would be directly applicable to \eqref{weintro} in the case of pure AdS spacetime (\ref{ads}) and the conformal mass $\sigma = 2$ (in the finite setting $( \mathbb{R} \times \Sph^3_+, g_E )$).

For our purposes, however, this approach would be unsatisfactory for a number of reasons.
The first is the requirement that the PDE be analytic; this method breaks down entirely if one adds a non-analytic potential to \eqref{weintro} or if $g_{\rm AdS}$ is replaced by a non-analytic aAdS metric.\footnote{Indeed, Holmgren's theorem fails for linear PDE with only smooth coefficients. See the discussion in \cite[Sect. 13.6]{hor:lpdo2} and the references within.}
Secondly, since Holmgren's theorem is applicable only to linear equations (see \cite{meti:counter_holmg}), this cannot provide a viable path toward attacking our main goal: the highly nonlinear vacuum Einstein equations \eqref{ee}.

The second technique, which applies to PDEs with only sufficiently smooth coefficients, is based on a general class of weighted $L^2$-estimates known as \emph{Carleman estimates}, pioneered by Carleman in \cite{carl:uc_strong}. For geometric wave equations of the form 
\begin{align} \label{welinear}
\Box_g \phi = a^\alpha \nabla_\alpha \phi + V \phi \text{,}
\end{align}
where the lower-order coefficients $a^\alpha$ and $V$ are sufficiently smooth but not necessarily analytic (the metric $g$ is also allowed to be non-analytic), the main assumption required for deriving unique continuation via this method is that of (strong) \emph{pseudoconvexity}, a notion introduced by H\"ormander.

Intuitively, when an oriented hypersurface $S$ in a spacetime $( \mc{M}, g )$ is pseudoconvex (with respect to $\Box_g$, and in the ``positive" direction), then any null geodesic which is tangent to $S$ at a point remains strictly to the ``negative" side of $S$ nearby.
The precise geometric characterization of pseudoconvexity that we will use throughout this paper can be found in Definition \ref{def.pseudoconvex}.

The classical unique continuation result for wave equations is roughly stated below.
For details, the reader is referred to \cite{hor:lpdo4}.

\begin{proposition} \label{thm.uc_classical}
Let $( \mc{M}, g )$ and $S$ be as above, and suppose $\phi$ is a (sufficiently) smooth solution of \eqref{welinear} on $\mc{M}$.
Suppose $S$ is pseudoconvex (with respect to $\Box_g$, and in the ``positive" direction), and suppose $\phi$ has vanishing Cauchy data on some open subset $V \subseteq S$.
Then, there is an $\mc{M}$-neighborhood $U$ of $V$ such that $\phi$ vanishes on the portion of $U$ on the ``positive" side of $S$.
\end{proposition}

In other words, solutions of (\ref{welinear}) can be locally uniquely continued from the ``negative" side of pseudoconvex $S$ to its ``positive" side.

The following two examples are instructive in this context.
The timelike cylinder
\[ C_R = \{ (t, x, y, z ) \in \mathbb{R}^{1+3} \mid x^2 + y^2 + z^2 = R \} \text{,} \qquad R > 0 \]
in Minkowski spacetime is pseudoconvex in the inward direction, as null geodesics tangent to $C_R$ at a point remain outside $C_R$ elsewhere.
As a result, solutions can indeed be (locally) uniquely continued from any portion of $C_R$ into its interior.
On the other hand, a timelike hyperplane $H$ in Minkowski spacetime is not pseudoconvex, as there exist null geodesics remaining entirely in $H$ for all times.
Results of Alinhac and Baouendi, \cite{alin_baou:non_unique}, imply that one does not have such a local unique continuation result for wave equations with smooth coefficients.\footnote{On the other hand, Holmgren's theorem still guarantees unique continuation from $H$ for a wave equation with analytic coefficients.}

\begin{remark}
There also exist uniqueness theorems which combine elements of both the Holmgren and H\"ormander theories to obtain improved results for equations that are analytic in only some of the variables; see \cite{hor:uc_interp, robb_zuil:uc_interp, tata:uc_interp}.
\end{remark}

\subsubsection{Degenerate Pseudoconvexity}

The situation becomes less clear when pseudoconvexity degenerates.
To be more precise, the preceding description of pseudoconvexity in terms of null geodesics suggests the following:

\begin{definition}
A hypersurface $S$ in a spacetime $( \mc{M}, g )$ is \emph{zero}, or \emph{degenerate}, \emph{pseudoconvex} (with respect to $\Box_g$) iff it is ruled by null geodesics.
\end{definition}

In the zero pseudoconvex setting, various scenarios can occur.
A more careful analysis of the geometry \emph{near} the hypersurface $S$ is required in order to determine whether a (degenerate) Carleman estimate, and hence a unique continuation result, still holds for $S$ and, if so, what the nature of the result is.

For example, a timelike hyperplane $H$ in Minkowski spacetime is zero pseudoconvex, and the preceding \emph{local} unique continuation result fails for $H$.
On the other hand, there do exist \emph{global} unique continuation results for $H$. 
For instance, it was shown in \cite{ken_ruiz_sog:sobolev_unique} that if a solution $\phi$ of a wave equation on $\R^{n+1}$ vanishes on one side of $H$ and is globally regular, then it must in fact vanish everywhere.

Another natural appearance of zero pseudoconvex hypersurfaces is in unique continuation problems ``from infinity".
In particular, both strong and zero pseudoconvexity are conformally invariant notions, thus one can define pseudoconvexity of a hypersurface at infinity in terms of its properties in a conformally compactified picture.
In this sense, we can think of (future and past) null infinity $\mc{I}^\pm$ in Minkowski spacetime as yet another example of zero pseudoconvexity.\footnote{The same is also true for null infinity in other asymptotically flat spacetimes.}

In \cite{alex_schl_shao:uc_inf}, it was shown that if a wave vanishes to infinite order \emph{on more than half of $\mc{I}^\pm$} (with these portions being connected to spacelike infinity), then it must also vanish in the interior near this portion of $\mc{I}^\pm$.
While this is not a fully global result like the preceding example, there is still a non-local aspect: one needs vanishing on a sufficiently large portion of infinity for unique continuation to hold.
In addition, in some cases, the infinite-order vanishing assumption can be further relaxed if one makes additional global regularity assumptions; see \cite{alex_shao:uc_global}.

Finally, \cite{alex_schl_shao:uc_inf} also showed that for Schwarzschild, Kerr, and other positive mass spacetimes, one in fact has a purely local result: one can uniquely continue from an arbitrarily small portion of $\mc{I}^\pm$ near spacelike infinity.

The preceding examples show that there are multiple possibilities in zero pseudoconvex settings.
Whether one can prove global, semi-global, or local results depends crucially on the geometry near the hypersurface under consideration.

The proofs of uniqueness results in zero pseudoconvex settings also become more difficult.
For instance, while the general strategy still revolves around Carleman-type estimates, the degenerating pseudoconvexity often produces additional weights that must be matched or balanced with other terms.
This produces additional complications that one does not see in the classical finite problems assuming strong pseudoconvexity; see \cite{alex_schl_shao:uc_inf} for examples of these issues.

\subsection{The Asymptotic Geometry of AdS} \label{sec:geoads}

Turning to the AdS spacetime \eqref{ads}, one can see in the conformal picture on $(\mathbb{R} \times \Sph^3_+, g_E)$ that there are null geodesics everywhere tangent to conformal infinity, i.e., along the equator $\chi = \pi/2$.
Thus, we can think of AdS infinity as being zero pseudoconvex.
Can one nevertheless expect unique continuation?
If so, what kind of uniqueness result can we derive?

The first observation one can make is that the level sets of $r$ are strongly pseudoconvex near infinity, with this pseudoconvexity degenerating as one reaches infinity.
From this observation, one can derive (degenerate) Carleman estimates on the shaded region in Figure \ref{fig.carleman_r}, which implies a unique continuation result from all of AdS infinity, \emph{provided one assumes in addition that the wave must decay sufficiently quickly as $|t| \nearrow \infty$}.\footnote{This extra decay condition arises from the fact that the level sets of $r$ and AdS infinity do not form a bounded region, and one must hence treat boundary terms tending toward $|t| \nearrow \infty$.}
However, one can ask whether this can be improved:
\begin{itemize}
\item Can one do away with the additional decay assumption at $|t| \nearrow \infty$.

\item Can one also prove local unique continuation from arbitrarily small neighborhoods of a single point in AdS infinity, or perhaps a weaker result from some sufficiently large open subset of infinity?
\end{itemize}

A crucial step in answering these questions affirmatively would be to localize our arguments.
This can be done by finding foliations of (strongly) pseudoconvex hypersurfaces near infinity which also terminate at infinity at finite times.
In particular, we ask whether these level sets of $r$ remain pseudoconvex if one were to ``bend them back toward infinity" to some degree.

\begin{figure}
\centering
\begin{tikzpicture}[scale=1.75]
\draw[shade, white] (0, 1) -- (1, 1) -- (1, -1) -- (0, -1) -- cycle;
\draw[style=dashed, <->] (1, 1.2) -- (1, -1.2);
\draw[<->] (0, 1.2) -- (0, -1.2);
\draw[style=dashed, color=red] (0, 1) -- (1, 1);
\draw[style=dashed, color=red] (0, -1) -- (1, -1);
\node[right] at (1, 0) {$\scriptscriptstyle r = \infty$};
\node[left] at (0, 0) {$\scriptscriptstyle r = c$};
\node[above, color=red] at (0.5, 1) {$\scriptscriptstyle t \nearrow +\infty$};
\node[below, color=red] at (0.5, -1) {$\scriptscriptstyle t \searrow -\infty$};
\end{tikzpicture}
\caption{Level sets of $r$ (solid lines) near conformal infinity (dashed line). Red lines denote open ends of this foliation going to $|t| \nearrow \infty$ in the limit.}
\label{fig.carleman_r}
\end{figure}
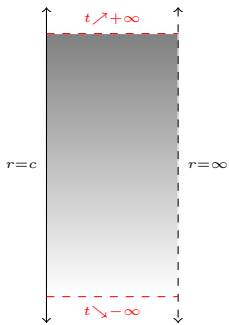

\subsubsection{Pseudoconvexity Near Infinity}

To find such hypersurfaces, we first recall that in AdS spacetime, there exists a family of null geodesics emanating from the conformal boundary at $t = 0$ that return to the boundary at a later finite time.
In addition, all these geodesics actually refocus at infinity at the \emph{same} characteristic AdS time $T_{AdS} = \pi$, as indicated in the left illustration on Figure \ref{fig.adscarly}.
Moreover, these geodesics generate a foliation of zero pseudoconvex (timelike) hypersurfaces near the portion $0 < t < \pi$ of AdS infinity.

This suggests that one could perhaps construct a pseudoconvex foliation if one were to ``elongate" the above hypersurfaces, i.e., initiating from $t = - \varepsilon$ and terminating at $t = \pi + \varepsilon$; see the right illustration in Figure \ref{fig.adscarly}.
In fact, we will show in Section \ref{sec.aads_defn} that such a foliation can be explicitly constructed.
This then further suggests that one could expect a unique continuation result if one assumes vanishing at infinity on a sufficiently long (but finite) time interval, such as $- \varepsilon < t < \pi + \varepsilon$.

\begin{figure}
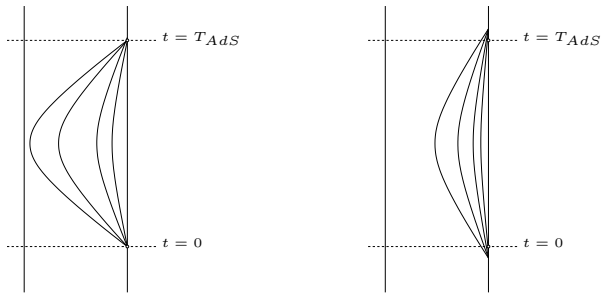
\caption{Left: Null geodesics from AdS infinity which refocus back at infinity at a later time.
Right: A family of strongly pseudoconvex hypersurfaces near infinity.}
\label{fig.adscarly}
\end{figure}

The above pictures also provide some heuristic justification to the conjecture that this ``sufficiently long time interval" assumption is in general necessary.
Suppose we assume that a wave vanishes at infinity (in whichever well-defined sense) on a ``short" time interval $\varepsilon < t < \pi - \varepsilon$.
Then, using the method of Gaussian beams \cite{Ralston, Sbierski:2013mva}, one could, roughly, construct waves which are ``locally concentrated" near one of the null geodesics in the left illustration in Figure \ref{fig.adscarly}.
However, these geodesics can be chosen to lie arbitrarily close to the boundary.

\subsubsection{Vanishing Assumptions} \label{sec:vanishing}

The above geometric insights suggest that it may be possible to prove (degenerate) Carleman estimates near AdS infinity, which would then lead to corresponding unique continuation results from infinity.
However, the observations so far do not yet suggest what order of vanishing imposed on $\phi$ at the conformal boundary will lead to these Carleman estimates and hence uniqueness. 

To gain some intuition as to what are reasonable vanishing assumptions to guarantee uniqueness, we consider solutions of the wave equation \eqref{weintro} (with $\Lambda = -3$ and $n = 3$) on AdS spacetime that are purely radial.
Defining the inverse radius $\rho := r^{-1}$, \eqref{weintro} reduces to a linear second-order ODE in $\rho$, which, via the Frobenius method, yields general solutions of the form\footnote{When $\beta_+ - \beta_-$ differs by an integer, $\phi_-$ may have an extra term of the form $c \cdot \phi_+ \ln \rho$.}
\[ \phi_\pm = \rho^{ \beta_\pm } \sum_{ k = 0 }^\infty a^\pm_k \rho^k \text{,} \qquad \beta_\pm := \frac{3}{2} \pm \sqrt{ \frac{9}{4} - \sigma } \text{} \]
near $\rho=0$.
Clearly, any unique continuation result from infinity must eliminate both branches $\phi_\pm$ of such solutions, hence we must at least assume
\begin{equation} \label{eq.vanishing_initial} 
r^{ \beta_+ } \phi = \rho^{ - \beta_+ } \phi \rightarrow 0 \text{,} \qquad r \nearrow \infty \text{.} 
\end{equation}

It is also easy to see that a more general separation ansatz of the form 
\[ \phi (t, r, \omega) := e^{i \lambda t} \alpha (r) Y_{l, m} (\omega) \text{,} \]
where $\lambda \in \R$, where $\omega \in \Sph^2$, and where $Y_{l, m}$ denote the standard spherical harmonics on $\Sph^2$ yields the same boundary asymptotics in powers of $r$ as before \cite{BF}.
The natural question is whether \eqref{eq.vanishing_initial} is also sufficient to imply vanishing.

\begin{remark}
Note the analogue of \eqref{eq.vanishing_initial} in Minkowski spacetime fails to imply vanishing.
Indeed, both $1$ and $r^{-1}$ solves the free wave equation on $\R^{3+1}$ away from $r = 0$.
However, by taking spatial Cartesian derivatives of $r^{-1}$, we can generate functions which vanish to any finite order at infinity.
\end{remark}

The exponents $\beta_\pm$ also play a crucial role in the rigorous forward well-posedness theory of (\ref{weintro}).
Without going into details (the theory will be briefly reviewed in Section \ref{sec.wp}), we recall that for $\frac{5}{4} < \sigma < \frac{9}{4}$, a unique solution of (\ref{weintro}) arising from data at $t=0$ can be constructed in a suitable energy space if \emph{either} Dirichlet \emph{or} Neumann boundary conditions are imposed\footnote{Inhomogeneous Dirichlet, inhomogeneous Neumann and Robin conditions are also possible but omitted from the present discussion for simplicity} at the conformal boundary:
\begin{equation} \label{eq.dir_neu} \rho^{-\beta_-} \phi \rightarrow 0 \quad \text{(Dirichlet),} \qquad \rho^{-2 + 2\beta_-} \partial_\rho ( \rho^{- \beta_-} \phi ) \rightarrow 0 \quad \text{(Neumann).} \end{equation}
It is not difficult to see that the two conditions in \eqref{eq.dir_neu} imply (\ref{eq.vanishing_initial}); see Section \ref{sec.wp}.

Thus, for $\frac{5}{4} < \sigma < \frac{9}{4}$, the conditions \eqref{eq.dir_neu} serve as natural assumptions on the boundary for unique continuation.
On the other hand, for $\sigma \leq \frac{5}{4}$, a solution arising from the forward well-posedness theory is already unique in the energy space---more specifically, having finite energy necessarily eliminates the $\rho^{\beta_-}$-branch.
Thus, in this case, assuming (\ref{eq.vanishing_initial}) and membership in the energy space would be natural.

\subsection{The main results}

We can now informally state our main results in the simplest case of the metric being (\ref{ads}).
As discussed below in Section \ref{sec.intro_gen}, these theorems in fact hold for a large class of aAdS spacetimes and any spacetime dimension.
The first theorem expresses the fact that for the mass range $\sigma \leq 2$, the condition (\ref{eq.vanishing_initial}) suffices to prove unique continuation:

\begin{theorem} \label{thm.uc_rough_optimal}
Suppose $\phi$ solves
\begin{equation} \label{eq.we_rough} \Box_{g_{AdS}} \phi + \sigma \phi = a^\alpha \nabla_\alpha \phi + V \phi \text{,} \end{equation}
where $a^\alpha$ and $V$ are smooth and decay sufficiently at infinity.
Suppose $\sigma \leq 2$, and suppose $\phi$ satisfies, \emph{on a segment $I$ of infinity of time length $T > \pi$},
\[ | \rho^{-\beta_+} \phi | + | \nabla_{ t, \rho, \Sph^2 } ( \rho^{-\beta_+ + 1} \phi ) | \rightarrow 0 \text{,} \]
where $\beta_+$ is as in \eqref{eq.vanishing_initial}.\footnote{Here, $\nabla_{t, \rho, \Sph^2}$ refers to derivatives in $t$, $\rho$, and some fixed bounded spherical coordinates. See the statement of Theorem \ref{theo:ads} for the precise vanishing conditions.}
Then, $\phi \equiv 0$ in the AdS interior near $I$.

Furthermore, the above result can be directly generalized to tensorial waves on a large class of aAdS backgrounds in all dimensions.
\end{theorem}

For $\sigma > 2$, we require stronger conditions at infinity than \eqref{eq.vanishing_initial}:

\begin{theorem} \label{thm.uc_rough_nonoptimal}
Suppose $\phi$ solves \eqref{eq.we_rough}, but with $\sigma > 2$.
Suppose $\phi$ satisfies, \emph{on a segment $I$ of infinity of time length $T > \pi$}, the vanishing condition
\[ | \rho^{-2} \phi | + | \nabla_{ t, \rho, \Sph^2 } ( \rho^{-1} \phi ) | \rightarrow 0 \text{.} \]
Then, $\phi \equiv 0$ in the AdS interior near $I$.

Furthermore, the above result can be directly generalized to tensorial waves on a large class of aAdS backgrounds in all dimensions.
\end{theorem}

For the precise and most general statements, see Theorem \ref{theo:ads}.
We also mention that for general merely bounded potentials $V$, one requires infinite-order vanishing of $\phi$ to prove a similar uniqueness result; see Theorem \ref{theo:ads5}.

As hoped, in the mass range where both Dirichlet and Neumann conditions are allowed by the forward well-posedness theory, $\frac{5}{4} < \sigma < \frac{9}{4}$, we can show unique continuation if both Dirichlet and Neumann conditions are imposed and the solution lives in the natural energy space of the forward well-posedness theory:

\begin{theorem} \label{thm.uc_rough_optimal2}
Assume $\phi$ solves (\ref{eq.we_rough}), with $\frac{5}{4} < \sigma < \frac{9}{4}$. If $\phi$ has bounded renormalized $H^2$-energy and satisfies, \emph{on a segment $I$ of infinity of time length $T > \pi$}, both Dirichlet and Neumann conditions (\ref{eq.dir_neu}), then $\phi \equiv 0$ in the interior near $I$.

Furthermore, the result can be directly generalized to tensorial waves on a large class of aAdS backgrounds in all dimensions.
\end{theorem}

We can also ask whether it is possible to globalize our local unique continuation results in particular settings.
While we will study global applications to black hole spacetimes, the study of gravitational perturbations, and the putative AdS-CFT correspondence in our companion paper \cite{ASfollowup}, we give a positive answer to the above question in the pure AdS case, whose proof we sketch in Section \ref{sec:gloads}.

\begin{corollary} \label{thm.global}
Let $\psi$ be a solution of (\ref{weintro}), with $\frac{5}{4} < \sigma < \frac{9}{4}$, on AdS satisfying the Dirichlet condition.
Then, if $\psi$ also satisfies the Neumann condition for a time length strictly larger than $\pi$, the solution is zero globally.
Moreover, the result generalizes to tensorial waves in all dimensions.
\end{corollary}

Finally, an application to gravitational perturbations and the Teukolsky equation is given in Corollary \ref{cor:adscft}.
See also Section \ref{sec:tensorwave} below.

\subsection{Discussion of the General Results} \label{sec.intro_gen}

Since the goal is to treat the nonlinear Einstein equations (\ref{ee}), it is important that our results and techniques for the wave equation are sufficiently robust.
We will therefore now discuss our Theorem \ref{theo:ads}, which provides the most general version of the results proven in this paper and includes Theorems \ref{thm.uc_rough_optimal} and \ref{thm.uc_rough_nonoptimal} as special cases.
Emphasis is laid on the features which suggest that treating the full Einstein equations is indeed possible.

\subsubsection{Nonlinear Equations and Systems}

By construction, unique continuation results obtained via Carleman estimates extend to some nonlinear equations.
In particular, Theorem \ref{theo:ads} applies to partial differential \emph{inequalities} of the form (\ref{we}), hence it applies to any \emph{nonlinear} wave equations that satisfies this inequality.

Furthermore, we remark that our results extend immediately to systems of wave equations satisfying analogous bounds.
This is simply a result of summing Carleman estimates that one obtains for each component and then applying again the standard argument detailed in the proof of Theorem \ref{theo:ads}.

\subsubsection{Arbitrary Dimensions}

Since the physics literature often focuses on higher-dimensional gravity, one is interested in uniqueness properties in AdS spacetimes of higher dimensions.
The results in Theorem \ref{theo:ads} (and hence Theorems \ref{thm.uc_rough_optimal} and \ref{thm.uc_rough_nonoptimal}) hold for any spacetime dimension\footnote{The only technical issue is that we must choose our multiplier vector field in our Carleman estimates more carefully; see Definition \ref{def.S}.}, with the only difference being that the relevant constants (e.g., the vanishing rates $\beta_\pm$) change according to the dimension.
See Theorem \ref{theo:ads} for the precise numerology.

\subsubsection{Asymptotically AdS Spacetimes} \label{sec:intro_aads}

Since solving the Einstein equations is tantamount to solving for the geometry of the spacetime itself, it is essential that our results extend to backgrounds that are not pure AdS but decay to AdS toward the conformal boundary.
Conceptually, the main hurdle is showing that the pseudoconvexity properties of the hypersurfaces discussed in Section \ref{sec:geoads} persist.

We show in Section \ref{sec.aads_pseudoconvex} that this pseudoconvexity property does in fact extend to a large class of aAdS spacetimes.
In particular, we must assume:
\begin{itemize}
\item The (conformally rescaled) metric $\mathring{g}$ induced on the boundary is static.\footnote{We will relax this staticity assumption in an upcoming article, \cite{hol_shao:uc_ads_ns}.}

\item The second-order expansion of the asymptotic boundary metric satisfies a certain positivity condition; see Definition \ref{def.pseudoconvex_ass}.
\end{itemize}
The full class of asymptotically AdS spacetimes we treat is defined precisely in Definitions \ref{def.aads} and \ref{def.pseudoconvex_ass}.
In particular, we note that we do not assume a particular topology or metric on the cross-sections of conformal infinity.

We also remark that in the case where the metric is Einstein-vacuum, the preceding positivity condition required for pseudoconvexity can be expressed geometrically: if the boundary metric $\mathring{g}$ is static, and if the metric induced on the cross-sections of infinity perpendicular to the timelike Killing vector field $\partial_t$ has positive curvature, then the positivity property holds.
There exist many nontrivial examples of dynamical aAdS Einstein spacetimes satisfying the criterion \cite{Enciso, Friedrich}.

\subsubsection{Tensorial Wave Equations} \label{sec:tensorwave}

A crucial feature of the Einstein equations (\ref{ee}) is their tensorial (as opposed to scalar) nature.
Our Theorem \ref{theo:ads} applies directly to wave equations satisfied by ``horizontal" tensor fields on aAdS spacetimes which are everywhere tangent to the level sets of $(t, r)$. We make precise sense of such tensor fields and wave equations on these fields in Section \ref{sec.aads_tensor}.

The restriction to such horizontal tensorfields is actually sufficient to study (\ref{ee}), if one adopts---as is commonly done (see, e.g., \cite{ChrKla})---a null or $(2+2)$-decomposition into a system of equations on such ``horizontal" tensor fields, which represent various components of the curvature and connection coefficients.
Thus, our tensorial results are well-adapted to study the Einstein equations in this type of formulation.

In particular, for the \emph{linearized} Einstein equations near AdS (or more generally, Kerr-AdS), it is known that certain components of the spacetime null-curvature components (corresponding precisely to the horizontal tensor fields above) satisfy decoupled wave equations, known as the \emph{Teukolsky equations}.
Theorem \ref{theo:ads} can then be applied directly to these equations, which---when combined with Corollary \ref{thm.global} above---produces a linearized version of a holographic correspondence: \emph{fixing both the (linearized) conformal class and the holographic stress energy tensor of a metric perturbation on the boundary, the perturbation is determined uniquely in the interior}, provided it exists.
See Corollary \ref{cor:adscft} for a precise formulation.
Further applications of this type will be discussed in our companion paper \cite{ASfollowup}.

Uniqueness results for the linearized Einstein equations can also be connected to the conditional rigidity results of \cite{Anderson} for the Einstein-vacuum equations.\footnote{The results in \cite{Anderson} assume that a unique continuation property holds for linearized Einstein-vacuum equations. Corollary \ref{thm.global} provides a unique continuation property for the linearized Bianchi equations (as opposed to a \emph{full} linearization of the vacuum Einstein equations). Hence some further work is required to prove the property assumed in \cite{Anderson}.}
We postpone a detailed analysis of this setting to \cite{ASfollowup}.

\subsection{Overview of the Proof} \label{sec.intro_proof}

We finally give a brief summary of the proof of our general result, Theorem \ref{theo:ads}, noting that Theorem \ref{thm.uc_rough_optimal2} (and its generalizations) can be deduced from Theorem \ref{theo:ads}, as done in Section \ref{sec.wp}.

As mentioned before, the proof of Theorem \ref{theo:ads} is based on a degenerate Carleman estimate, Theorem \ref{thm.carleman}.
In particular, the Carleman estimate controls a weighted $H^1$-norm of the wave $\phi$ on a spacetime region by a weighted $L^2$-norm $( \Box_g + \sigma ) \phi$ on the same region.\footnote{Plus another boundary term which vanishes in the limit as one approaches infinity.}
From this bound, a standard argument (see Section \ref{sec.proof}) yields the desired unique continuation result.

We adopt a geometric viewpoint and prove this estimate using a variant of the vector field method.
The details of this can be found in Section \ref{sec.carleman}.
At a more conceptual level, the process revolves around two main ingredients:

\subsubsection{Finding a Pseudoconvex Foliation}

The first task is to capture the pseudoconvexity described in Section \ref{sec:geoads}.
More specifically, we wish to construct a function $f$ whose level sets form hypersurfaces that both initiate from and terminate at infinity.
Then, by choosing the multiplier vector field in our estimates to be orthogonal to the level sets of $f$, we can obtain the positivity needed to control derivatives of $\phi$ tangential to the level sets of $f$.

Recalling the notation $\rho := r^{-1}$, then a viable candidate for $f$ is
\[ f = \frac{\rho}{ \sin ( y t ) } \text{,} \qquad y > 0 \text{.} \]
In particular, the level sets of $f$ intersect infinity at $t = 0$ and $t = y^{-1} \pi$, that is, these level sets span a time interval of length $y^{-1}  \pi$.
Computations in Section \ref{sec.aads_pseudoconvex} show that, as hoped, these level sets are pseudoconvex whenever $y < 1$.
This pseudoconvexity of course degenerates as one approaches infinity, i.e., as $f \searrow 0$.

\subsubsection{Vanishing Rates}

The other crucial ingredient in the Carleman estimate is to find a suitable weight $F$, constructed from the above $f$ and (at least) a free real parameter.
As is standard for Carleman estimates, the main technical step is to consider the conjugated wave equation applied to $\psi := e^{-F} \phi$.

This weight $e^{-F}$ in particular determines the strength of the vanishing assumption one requires for $\phi$.\footnote{More specifically, $e^{-F}$ shows up in boundary terms that one obtains via integrations by parts. Such boundary terms are required to vanish as the boundary approaches infinity.}
The additional difficulty here is that one wants a vanishing condition as close as possible to the rate indicated by \eqref{eq.vanishing_initial}.
In obtaining this rate in Theorem \ref{thm.uc_rough_optimal} (and its generalizations in Theorem \ref{theo:ads}), one must carefully exploit all the positivity present in the estimate.

As previously mentioned, because of additional decaying weights (both from the pseudoconvexity and from $F$), one must take extra care in our degenerate Carleman estimate to ensure that the weights match, and that various terms can be absorbed as required.
As in \cite{alex_schl_shao:uc_inf}, this causes added technical difficulties throughout.

\subsection{Acknowledgements} 
The authors acknowledge support through a grant of the European Research Council and thank Claude Warnick for valuable comments.

\section{Asymptotically AdS Spacetimes} \label{sec.aads}

In this section, we give a precise description of the class of asymptotically Anti-de Sitter (abbreviated aAdS) spacetimes that we will treat.
We also establish a criterion for these spacetimes, stated in terms of quantities at infinity, which guarantees that a certain class of hypersurfaces terminating at infinity are pseudoconvex.

\subsection{Construction of the Spacetimes} \label{sec.aads_defn}

We first provide the definition of our class of (static) aAdS spacetimes.
Throughout, we fix the spatial dimension $n \in \N$.

\begin{remark}
Our main assumptions will be stated in terms of a certain collection of bounded coordinate systems.
Although these assumptions can be stated more invariantly, we opt instead for formulations that are closer to how they are applied.
\end{remark}

\subsubsection{Bounded Static Infinities}

For our main results, we will restrict our attention to spacetimes whose metric at infinity is static.
As a result, this geometry at infinity will be defined by that of its $t$-level sets, which we define below:

\begin{definition} \label{def.aads_horizontal}
Let $( \mc{S}, \mathring{\gamma} )$ be an $(n - 1)$-dimensional Riemannian manifold.
We assume that $\mc{S}$ can be covered by a collection of coordinate systems $x^1, \dots, x^{n-1}$, such that with respect to each of these coordinate systems, we have the bounds:
\begin{equation} \label{eq.aads_bdd}
| \partial_{ A_1 } \dots \partial_{ A_m } \mathring{\gamma}_{B C} | \lesssim_{n, m} 1 \text{,} \qquad | \mathring{\gamma}^{B C} | \lesssim_n 1 \text{,} \qquad m \geq 0 \text{.}
\end{equation}
We refer to each of these coordinate systems as \emph{bounded}.\footnote{Although each individual estimate in \eqref{eq.aads_bdd} may depend on $n$ and $m$, these constants are, on the other hand, independent of the choice of admissible coordinate systems.}
\end{definition}

\begin{remark}
Note that the coordinate condition in Definition \ref{def.aads_horizontal} holds trivially whenever $\mc{S}$ is compact.
More generally, if $\mc{S}$ satisfies that its curvature is bounded to all orders, then $\mc{S}$ can be covered with a set of bounded normal coordinate systems.
Non-compact examples of this include $\R^{n - 1}$ and $\mathbb{H}^{n - 1}$.
\end{remark}

From $( \mc{S}, \mathring{\gamma} )$, we obtain a static AdS infinity in the expected manner:

\begin{definition} \label{def.aads_infinity}
Let $( \mc{S}, \mathring{\gamma} )$ satisfy the assumptions in Definition \ref{def.aads_horizontal}, and let $\tfr > 0$.
We then refer to the $n$-dimensional Lorentzian manifold
\begin{equation} \label{eq.aads_infinity}
( \mc{I}_\tfr, \mathring{g} ) \text{,} \qquad \mc{I}_\tfr := (0, \tfr^{-1} \pi) \times \mc{S} \text{,} \qquad \mathring{g} := - dt^2 + \mathring{\gamma}
\end{equation}
as a segment of \emph{bounded static AdS infinity}.
Here, we use the symbol $t \in C^\infty ( \mc{I}_\tfr )$ to denote the projection to the first component $(0, \tfr^{-1} \pi)$.

Moreover, given a bounded coordinate system $( x^1, \dots, x^{n-1} )$ on $\mc{S}$, we refer to the collection $( t, x^1, \dots, x^{n-1} )$ as a \emph{bounded coordinate system} on $\mc{I}_\tfr$.
\end{definition}

\begin{definition} \label{def.aads_inf_indices}
We adopt the following coordinate conventions on $\mc{S}$ and $\mc{I}_\tfr$:
\begin{itemize}
\item We let $x^A$ denote coordinate functions from the bounded coordinate systems in Definition \ref{def.aads_horizontal}.
Similarly, we use upper-case Latin indices denote coordinate components with respect to these bounded coordinate systems.

\item We let $x^a \in \{ t, x^A \}$ denote the corresponding bounded coordinates on $\mc{I}_\tfr$.
Lower-case Latin indices denote the corresponding components on $\mc{I}_\tfr$.
\end{itemize}
Furthermore, as usual, we adopt Einstein summation notation: repeated indices denote sums over all frame and coframe elements.
\end{definition}

\begin{remark}
Note in particular that from Definitions \ref{def.aads_infinity} and \ref{def.aads_inf_indices}, we have, for any bounded static AdS infinity, the identities
\begin{equation} \label{eq.aads_static} \partial_t \mathring{g}_{ab} \equiv 0 \text{,} \qquad \mathring{g}_{tt} \equiv -1 \text{,} \qquad \mathring{g}_{t A} \equiv 0 \text{,} \end{equation}
\end{remark}

\subsubsection{Admissible Spacetimes}

Now that we have defined the geometry of infinity, we construct the appropriate class of spacetimes that have such an infinity.
The first step is merely to define the manifold of the spacetime.

\begin{definition} \label{def.aads_manifold}
Let $\mc{M}$ be a smooth $(n+1)$-dimensional manifold of the form
\begin{equation} \label{eq.aads_manifold}
\mc{M} := (r_0, \infty) \times \mc{I}_\tfr \text{,} \qquad r_0, \tfr > 0 \text{.}
\end{equation}
Moreover, let $r \in C^\infty ( \mc{M} )$ denote the projections to the first component of $\mc{M}$.
\end{definition}

For convenience, we also define notations for functions representing ``error terms", whose exact forms are unimportant to the analysis.

\begin{definition} \label{def.O_scr_pre}
Given a positive function $\zeta$ on $\mc{M}$, we use the symbol $\mc{O} ( \zeta )$ to denote any function $\xi: \mc{M} \rightarrow \R$ such that the following holds,
\begin{equation} \label{eq.O_scr_pre}
| \partial_r^k \partial_{ a_1 } \dots \partial_{ a_m } \xi | \lesssim_{n, \tfr, k, m} r^{- k + m} \zeta \text{,} \qquad k, m \geq 0 \text{,}
\end{equation}
where the derivatives are with respect to bounded coordinate systems on $\mc{I}_\tfr$.\footnote{Again, the constants of the inequalities \eqref{eq.aads_pseudo_bdd} are allowed to depend on $n$, $\tfr$, and $m$, but not on the choice of bounded coordinate systems on $\mc{I}_\tfr$.}
\end{definition}

We are now prepared to define our class of aAdS spacetimes:

\begin{definition} \label{def.aads}
Let $( \mc{S}, \mathring{\gamma} )$ be an $(n-1)$-dimensional Riemannian manifold satisfying the assumptions of Definition \ref{def.aads_horizontal}; let $( \mc{I}_\tfr, \mathring{g} )$, $\tfr > 0$, denote the corresponding segment of bounded static AdS infinity (as in Definition \ref{def.aads_infinity}); and let the manifold $\mc{M} := (r_0, \infty) \times \mc{I}_\tfr$ be as in Definition \ref{def.aads_manifold}.

We say that $( \mc{M}, g )$ is an \emph{admissible aAdS segment} iff $g$ is a Lorentzian metric, which, with respect to each bounded coordinate system of $\mc{I}_\tfr$ (see Definition \ref{def.aads_infinity}), satisfies an asymptotic expansion of the form
\begin{equation} \label{eq.aads} \begin{split}
g &= [ r^{-2} - r^{-4} + \mc{O} ( r^{-5} ) ] dr^2 + \sum_a \mc{O} ( r^{-3} ) \cdot dr dx^a + r^2 \mf{g} \text{,}
\end{split} \end{equation}
where $\mf{g}$ contains only $x^a$-components and is of the form
\begin{equation} \label{eq.aads_asymp} \mf{g} := ( \mathring{g}_{a b} + r^{-2} \bar{g}_{a b} ) dx^a dx^b + \sum_{a, b} \mc{O} ( r^{-3} ) \cdot dx^a dx^b \text{,} \end{equation}
and where $\bar{g}$ is a symmetric covariant $2$-tensor field on $\mc{I}_y$, which is \emph{bounded} in the following sense: with respect to any bounded coordinate system $x^a$ on $\mc{I}_y$,
\begin{equation} \label{eq.aads_pseudo_bdd}
| \partial_{ a_1 } \dots \partial_{ a_m } \bar{g}_{bc} | \lesssim_{n, \tfr, m} 1 \text{,} \qquad m \geq 0 \text{.}
\end{equation} 
\end{definition}

\begin{remark}
The class of metrics in Definition \ref{def.aads} includes those of \cite{Hol:wp} but is strictly smaller than those considered in \cite{Warnick:2012fi}, which do not require staticity of $\mathring{g}$.
In fact, the assumption that $\mathring{g}$ is static can be replaced by a weaker assumption that $\mathring{g}$ ``does not vary too much in $t$".
This generalization will be made precise and proved in an upcoming paper, \cite{hol_shao:uc_ads_ns}.
\end{remark}

\begin{remark}
Moreover, $g_{rr} = r^{-2} - r^{-4} + \mc{O} ( r^{-5} )$ in \eqref{eq.aads} can be replaced by
\[
r^{-2} - \varsigma r^{-4} + \mc{O} ( r^{-5} ) \text{,}
\]
where $\varsigma$ is some smooth function on $\mc{I}_y$ that is bounded up to sufficiently high order, \emph{provided that we weaken our definition of ``$\mc{O} ( \zeta )$"}.\footnote{In particular, taking an $x^a$-derivative does not result in a loss of a power of $r$.}
In our upcoming paper \cite{hol_shao:uc_ads_ns}, we will work with a class of aAdS spacetimes defined precisely in this manner.
\end{remark}

For our upcoming pseudoconvexity criterion, we will need information on higher-order asymptotics of the metric at infinity.
This is captured in our aAdS segments by the tensor field $\bar{g}$ in Definition \ref{def.aads}.

\begin{definition} \label{def.aads_indices}
We will also adopt the following notations for objects on $(\mc{M}, g)$:
\begin{itemize}
\item Let $\nabla$ denote the Levi-Civita connection associated with $g$, and let $\nasla$ denote the induced connections on the level sets of $(t, r)$, i.e., the copies of $\mc{S}$.

\item Lower-case Greek indices denote spacetime coordinate components, though not necessarily with respect to any special coordinate system.
Again, repeated indices will denote sums over all frame and coframe elements.
\end{itemize}
\end{definition}

\begin{remark}
In our analysis, we will only use a finite number of derivatives of $g$.
Thus, the assumptions for admissible aAdS segments in Definition \ref{def.aads} can be weakened to requiring bounds for only a finite but large enough number of derivatives of the metric.
For simplicity, we avoid this level of generality in this paper.
\end{remark}

\subsubsection{Examples}

The prototypical example of an admissible aAdS segment is (a subset of) \emph{AdS spacetime} itself.
More specifically, consider
\begin{equation} \label{eq.ads} \begin{split}
\mc{M}_\AdS &:= (r_0, \infty) \times (0, \tfr^{-1} \pi) \times \Sph^{n-1} \text{,} \\
g_\AdS &:= ( 1 + r^2 )^{-1} dr^2 - ( 1 + r^2 ) dt^2 + r^2 \mathring{\gamma} \text{,}
\end{split} \end{equation}
where $\mathring{\gamma}$ is now the standard metric on the unit sphere $\Sph^{n-1}$.
Then, $( \mc{M}_\AdS, g_\AdS )$ indeed satisfies the postulates in Definition \ref{def.aads}, with
\begin{equation} \label{eq.ads_g_ring}
\mathring{g} = - dt^2 + \mathring{\gamma} \text{,} \qquad \bar{g} = - dt^2 \text{.}
\end{equation}

More generally, one still has an admissible aAdS segment when $( \Sph^{n - 1}, \mathring{\gamma} )$ is replaced by an arbitrary $(n - 1)$-dimensional Riemannian manifold $( \mc{S}, \mathring{\gamma} )$ satisfying the boundedness assumptions of Definition \ref{def.aads_horizontal}.
In other words, the spacetime
\begin{equation} \label{eq.ads_gen} \begin{split}
\mc{M} &:= (r_0, \infty) \times (0, \tfr^{-1} \pi) \times \mc{S} \text{,} \\
g &:= ( 1 + r^2 )^{-1} dr^2 - ( 1 + r^2 ) dt^2 + r^2 \mathring{\gamma} \text{,}
\end{split} \end{equation}
also satisfies the assumptions of Definition \ref{def.aads}.

In fact, admissible aAdS segments can be viewed as the perturbations of \eqref{eq.ads} and \eqref{eq.ads_gen} for which the error terms decay at appropriate rates toward infinity. 

\subsection{Geometric Properties} \label{sec.aads_prop}

In this subsection, we derive some basic asymptotic properties regarding the geometry of admissible aAdS spacetimes.
Throughout the remainder of the section, we assume as our background setting such an admissible aAdS segment $( \mc{M}, g )$, as specified in Definition \ref{def.aads}.

We begin by adopting a change of radial coordinate that is better suited for analyzing geometry of our spacetime at infinity.

\begin{definition}
We define the inverted radius $\rho$ by
\begin{equation} \label{eq.rho} \rho := r^{-1} \text{,} \qquad \rho \in ( 0, r_0^{-1} ) \text{.} \end{equation}
In particular, infinity corresponds to $\rho = 0$.
\end{definition}

\begin{remark} \label{def.O_scr}
Note that Definition \ref{def.O_scr_pre} can be equivalently expressed as follows: given $\zeta: \mc{M} \rightarrow \R^+$, we use the symbol $\mc{O} ( \zeta )$ to denote any function $\xi: \mc{M} \rightarrow \R$ such that
\begin{equation} \label{eq.O_scr}
| \partial_\rho^k \partial_{ a_1 } \dots \partial_{ a_m } \xi | \lesssim_{n, \tfr, k, m} \rho^{- k - m} \zeta \text{,} \qquad k, m \geq 0 \text{.}
\end{equation}
\end{remark}

\subsubsection{Metric Expansions}

We now expand the various components of $g$ and its derivatives, with respect to $\rho$-$t$-$x^A$-coordinates.

\begin{proposition} \label{thm.g}
Whenever $\rho \ll_{n, \tfr} 1$ (or in other words, $r \gg_{n, \tfr} 1$), the following asymptotic properties hold in $\rho$-$t$-$x^A$-coordinates:
\begin{itemize}
\item The components of $g$ satisfy
\begin{equation} \label{eq.g}
g_{\rho \rho} = \rho^{-2} - 1 + \mc{O} ( \rho ) \text{,} \qquad g_{\rho a} = \mc{O} ( \rho ) \text{,} \qquad g_{a b} = \rho^{-2} \mathring{g}_{ab} + \bar{g}_{ab} + \mc{O} ( \rho ) \text{.}
\end{equation}
In particular,
\begin{equation} \label{eq.g_t}
g_{tt} = - \rho^{-2} + \bar{g}_{tt} + \mc{O} ( \rho ) \text{,} \qquad g_{t A} = \bar{g}_{t A} + \mc{O} ( \rho ) \text{.}
\end{equation}

\item The dual of $g$ satisfies\footnote{By $\mathring{g}^{ab}$, we mean the components of the inverse of the matrix formed by the $\mathring{g}_{ab}$'s.}
\begin{equation} \label{eq.g_inv} \begin{split}
g^{\rho \rho} &= \rho^2 + \rho^4 + \mc{O} ( \rho^5 ) \text{,} \\
g^{\rho a} &= \mc{O} ( \rho^5 ) \text{,} \\
g^{a b} &= \rho^2 \mathring{g}^{ab} - \rho^4 \mathring{g}^{ac} \mathring{g}^{bd} \bar{g}_{cd} + \mc{O} ( \rho^5 ) \text{.}
\end{split} \end{equation}
In particular,
\begin{equation} \label{eq.g_inv_t}
g^{tt} = - \rho^2 - \rho^4 \bar{g}_{tt} + \mc{O} ( \rho^5 ) \text{,} \qquad g^{t A} = \rho^4 \mathring{g}^{A B} \bar{g}_{t B} + \mc{O} ( \rho^5 ) \text{.}
\end{equation}

\item The Christoffel symbols with respect to these coordinates satisfy
\begin{equation} \label{eq.Gamma} \begin{split}
\Gamma^\rho_{\rho \rho} = - \rho^{-1} - \rho + \mc{O} ( \rho^2 ) \text{,} &\qquad \Gamma^\rho_{\rho a} = \mc{O} ( \rho^2 ) \text{,} \\
\Gamma^\rho_{a b} = ( \rho^{-1} + \rho ) \mathring{g}_{ab} + \mc{O} ( \rho^2 ) \text{,} &\qquad \Gamma^a_{\rho \rho} = \mc{O} ( \rho^2 ) \text{,} \\
\Gamma^a_{\rho b} = - \rho^{-1} \delta^a_b + \rho \mathring{g}^{ac} \bar{g}_{cb} + \mc{O} ( \rho^2 ) \text{,} &\qquad \Gamma^a_{b c} = \mc{O} (1) \text{.}
\end{split} \end{equation}
In addition, we have improved bounds when $\Gamma^a_{b c}$ contains a $t$-component:
\begin{equation} \label{eq.Gamma_t}
\Gamma^t_{a b} = \mc{O} ( \rho^2 ) \text{,} \qquad \Gamma^a_{t b} = \mc{O} ( \rho^2 ) \text{.}
\end{equation}

\item Furthermore, the Christoffel symbols satisfy
\begin{equation} \label{eq.Gamma_deriv}
\partial_\rho \Gamma^c_{a b} = \mc{O} ( \rho ) \text{,} \qquad \partial_t \Gamma^c_{a b} = \mc{O} ( \rho ) \text{,} \qquad \partial_a \Gamma^c_{\rho b} = \mc{O} ( \rho ) \text{,} \qquad \partial_a \Gamma^c_{t b} = \mc{O} ( \rho ) \text{.}
\end{equation}
\end{itemize}
\end{proposition}

\begin{proof}
From \eqref{eq.aads} and \eqref{eq.rho}, we obtain that
\begin{equation} \label{eq.g_0} \rho^2 g = ( 1 - \rho^2 ) d \rho^2 + \mf{g} + \mc{O} ( \rho^3 ) \text{.} \end{equation}
The expansions \eqref{eq.g} follow immediately from \eqref{eq.g_0}, while \eqref{eq.g_t} follows from \eqref{eq.aads_asymp} and the last equation in \eqref{eq.g}.
Inverting the matrix defined by \eqref{eq.g_0} yields
\begin{equation} \label{eq.g_1} \rho^{-2} g^{-1} = ( 1 + \rho^2 ) \partial_\rho \otimes \partial_\rho + \mf{g}^{-1} + \mc{O} ( \rho^3 ) \text{,} \end{equation}
where we have used that $\rho$ is small.
The expansions \eqref{eq.g_inv} now follow from \eqref{eq.g_1}, while \eqref{eq.g_inv_t} is a consequence of \eqref{eq.aads_asymp} and \eqref{eq.g_t}.

To compute the Christoffel symbols, we first note that
\begin{equation} \label{eq.Gamma_0} \begin{split}
\partial_\rho g_{\rho \rho} = -2 \rho^{-3} + \mc{O} ( 1 ) \text{,} &\qquad \partial_c g_{\rho \rho} = \mc{O} ( 1 ) \text{,} \\
\partial_\rho g_{\rho a} = \mc{O} ( 1 ) \text{,} &\qquad \partial_c g_{\rho a} = \mc{O} ( 1 ) \text{,} \\
\partial_\rho g_{a b} = -2 \rho^{-3} \mathring{g}_{ab} + \mc{O} ( 1 ) \text{,} &\qquad \partial_c g_{a b} = \rho^{-2} \partial_c \mathring{g}_{ab} + \mc{O} ( 1 ) \text{.}
\end{split} \end{equation}
Equations \eqref{eq.Gamma} follow from \eqref{eq.g_inv} and \eqref{eq.Gamma_0}.
Furthermore, by \eqref{eq.aads_static} and \eqref{eq.Gamma_0},
\begin{equation} \label{eq.Gamma_1}
\partial_t g_{a b} = \mc{O} ( 1 ) \text{,} \qquad \partial_c g_{t a} = \mc{O} ( 1 ) \text{.}
\end{equation}
The identities \eqref{eq.Gamma_t} now follow from \eqref{eq.g_inv}, \eqref{eq.g_inv_t}, \eqref{eq.Gamma_0}, and \eqref{eq.Gamma_1}.

Finally, for \eqref{eq.Gamma_deriv}, we apply \eqref{eq.g} and \eqref{eq.g_inv} to write
\begin{equation} \label{eq.Gamma_deriv_0}
\Gamma^c_{a b} = \frac{1}{2} \mathring{g}^{c d} ( \partial_a \mathring{g}_{d b} + \partial_b \mathring{g}_{d a} - \partial_d \mathring{g}_{a b} ) + \mc{O} ( \rho^2 ) \text{.}
\end{equation}
The first two bounds in \eqref{eq.Gamma_deriv} follow from the fact that $\mathring{g}$ is independent of both $\rho$ and $t$ (see \eqref{eq.aads_static}).
Moreover, by \eqref{eq.aads_static} and \eqref{eq.Gamma_deriv_0}, we have
\begin{equation} \label{eq.Gamma_deriv_1}
\Gamma^c_{t b} = \mc{O} ( \rho^2 ) \text{,} \qquad \partial_a \Gamma^c_{t b} = \mc{O} ( \rho ) \text{,}
\end{equation}
proving the last bound in \eqref{eq.Gamma_deriv}.
The remaining bound for $\partial_a \Gamma^c_{\rho b}$ in \eqref{eq.Gamma_deriv} follows by simply differentiating the identity for $\Gamma^c_{\rho b}$ in \eqref{eq.Gamma}.
\end{proof}

Later, we will also require asymptotic expansions for the gradient of $\rho$:

\begin{proposition} \label{thm.rho_grad}
Suppose $\rho \ll_{n, \tfr} 1$.
Then, the ($g$-)gradient of $\rho$ satisfies
\begin{equation} \label{eq.rho_grad}
\grad \rho = [ \rho^2 + \mc{O} ( \rho^4 ) ] \partial_\rho + \mc{O} ( \rho^5 ) \cdot \partial_t + \sum_{ A = 1 }^{n - 1} \mc{O} ( \rho^5 ) \cdot \partial_{ x^A } \text{.}
\end{equation}
Moreover, the outer-pointing unit normal to a level set $\{ \rho = \rho_0 \}$ is given by
\begin{equation} \label{eq.rho_normal}
\mc{N} := | \nabla^\alpha \rho \nabla_\alpha \rho |^{- \frac{1}{2}} \grad \rho = [ \rho + \mc{O} ( \rho^3 ) ] \partial_\rho + \mc{O} ( \rho^4 ) \cdot \partial_t + \sum_{ A = 1 }^{n - 1} \mc{O} ( \rho^4 ) \cdot \partial_{ x^A } \text{.}
\end{equation}
\end{proposition}

\begin{proof}
The expansion \eqref{eq.rho_grad} follows immediately from \eqref{eq.g_inv} and the identity
\[
\grad \rho = g^{\alpha \beta} \partial_\alpha \rho \partial_\beta = g^{\rho \beta} \partial_\beta \text{.}
\]
By a similar computation, we have
\[
\nabla^\alpha \rho \nabla_\alpha \rho = g^{\rho \rho} = \rho^2 + \mc{O} ( \rho^4 ) \text{.}
\]
The remaining expansion \eqref{eq.rho_normal} now follows from \eqref{eq.rho_grad} and the above.
\end{proof}

\subsubsection{The Function $f$}

We now introduce the function $f \in C^\infty ( \mc{M} )$, defined
\begin{equation} \label{eq.f}
f := \frac{ \rho }{ \sin ( \tfr t ) } \text{.}
\end{equation}
Note that the level sets $f = \varepsilon > 0$ focus at infinity as $t \searrow 0$ and $t \nearrow \tfr^{-1} \pi$, cf.~the figure shown in the introduction.
Note also that in our domain of consideration,
\begin{equation} \label{eq.f_trivial}
0 < \rho \leq f \text{,} \qquad f^2 \cos^2 (\tfr t) = f^2 - \rho^2 \text{.}
\end{equation}

Differentiating $f$, we see that
\begin{equation} \label{eq.f_deriv}
\partial_\rho f = f \rho^{-1} \text{,} \qquad \partial_t f = - \tfr f^2 \rho^{-1} \cos (c t) \text{,}
\end{equation}
and that
\begin{equation} \label{eq.f_deriv_2}
\partial^2_{\rho \rho} f = 0 \text{,} \qquad \partial^2_{\rho t} f = - \tfr f^2 \rho^{-2} \cdot \cos (\tfr t) \text{,} \qquad \partial^2_{t t} f = \tfr^2 f \rho^{-2} ( 2 f^2 - \rho^2 ) \text{.}
\end{equation}
In particular, observe that for $f \ll_{n, \tfr} 1$, we have the trivial bounds
\begin{equation} \label{eq.f_deriv_trivial} f = \mc{O} (f) = \mc{O} (1) \text{,} \qquad \mc{O} ( f \rho ) = \mc{O} ( 1 ) \text{.} \end{equation}

We now collect some asymptotic identities involving $\nabla f$ and $\nabla^2 f$:

\begin{proposition} \label{thm.f_grad}
Suppose $f \ll_{n, \tfr} 1$.
Then, the gradient of $f$ satisfies\footnote{In particular, $\rho \ll_{n, \tfr} 1$ by \eqref{eq.f_trivial}.}
\begin{equation} \label{eq.f_grad}
\grad f = f [ \rho + \mc{O} ( \rho^3 ) ] \partial_\rho + f [ \tfr f \rho \cos (\tfr t) + \mc{O} ( \rho^3 ) ] \partial_t + \sum_{ A = 1 }^{n - 1} \mc{O} ( f^2 \rho^3 ) \cdot \partial_{ x^A } \text{.}
\end{equation}
In particular,
\begin{equation} \label{eq.f_grad_sq}
\nabla^\alpha f \nabla_\alpha f = f^2 [ 1 - \tfr^2 f^2 + \mc{O} ( \rho^2 ) ] = f^2 + \mc{O} ( f^4 ) \text{.}
\end{equation}
Furthermore, the components of $\nabla^2 f$ satisfy:
\begin{equation} \label{eq.f_hessian_pre} \begin{split}
\nabla_{\rho \rho} f &= f \rho^{-2} [ 1 + \rho^2 + \mc{O} ( \rho^3 ) ] \text{,} \\
\nabla_{t \rho} f &= f \rho^{-2} [ - 2 \tfr f \cos (\tfr t) - \tfr f \rho^2 \cos (\tfr t) \cdot \bar{g}_{tt} + \mc{O} ( \rho^3 ) ] \text{,} \\
\nabla_{t t} f &= f \rho^{-2} [ 1 + 2 \tfr^2 f^2 + ( 1 - \tfr^2 ) \rho^2 + \mc{O} ( \rho^3 ) ] \text{,} \\
\nabla_{A B} f &= f \rho^{-2} [ - ( 1 + \rho^2 ) \mathring{g}_{AB} + \mc{O} ( \rho^3 ) ] \text{,} \\
\nabla_{\rho A} f &= f \rho^{-2} [ - \tfr f \rho^2 \cos (\tfr t) \cdot \bar{g}_{t A} + \mc{O} ( \rho^3 ) ] \text{,} \\
\nabla_{t A} f &= f \rho^{-2} \cdot \mc{O} ( \rho^3 ) \text{.}
\end{split} \end{equation}
\end{proposition}

\begin{proof}
Expanding using \eqref{eq.g_inv}, \eqref{eq.g_inv_t}, and \eqref{eq.f_deriv}, we obtain
\begin{equation} \label{eq.f_grad_0} \begin{split}
\grad f &= ( g^{\rho\rho} \partial_\rho f + g^{\rho t} \partial_t f ) \partial_\rho + ( g^{\rho t} \partial_\rho f + g^{t t} \partial_t f) \partial_t \\
&\qquad + ( g^{\rho A} \partial_\rho f + g^{t A} \partial_t f) \partial_{ x^A } \\
&= [ f \rho + \mc{O} ( f \rho^3 ) ] \partial_\rho + \mc{O} ( f^2 \rho^4 ) \cdot \partial_\rho + [ f^2 \rho + \mc{O} ( f^2 \rho^3 ) ] \tfr \cos (\tfr t) \cdot \partial_t \\
&\qquad + \mc{O} ( f \rho^4 ) \cdot \partial_t + \sum_{ A = 1 }^{n - 1} \mc{O} ( f \rho^4 ) \cdot \partial_{ x^A } + \sum_{ A = 1 }^{n - 1} \mc{O} ( f^2 \rho^3 ) \cdot \partial_{ x^A } \text{.}
\end{split} \end{equation}
The identity \eqref{eq.f_grad} now follows from \eqref{eq.f_deriv_trivial} and \eqref{eq.f_grad_0}.
Similarly,
\begin{equation} \label{eq.f_grad_1} \begin{split}
\nabla^\alpha f \nabla_\alpha f &= g^{\rho\rho} ( \partial_\rho f )^2 + 2 g^{\rho t} \partial_\rho f \partial_t f + g^{tt} ( \partial_t f )^2 \\
&= f^2 [ 1 + \mc{O} ( \rho^2 ) ] + f^3 \cdot \mc{O} ( \rho^3 ) + [ -1 + \mc{O} ( \rho^2 ) ] \tfr^2 f^2 ( f^2 - \rho^2 ) \text{,}
\end{split} \end{equation}
where we also used \eqref{eq.f_trivial}.
Applying \eqref{eq.f_deriv_trivial} yields \eqref{eq.f_grad_sq}.

The identities for $\nabla^2 f$ are similarly derived.
Since
\[ \nabla_{\alpha\beta} f = \partial_\alpha \partial_\beta f - \Gamma^\mu_{\alpha\beta} \partial_\mu f \text{,} \]
for each pair $(\alpha, \beta)$, we can expand the right-hand sides using Proposition \ref{thm.g}, \eqref{eq.f_deriv}, and \eqref{eq.f_deriv_2}.
Again, keeping in mind \eqref{eq.f_trivial} and \eqref{eq.f_deriv_trivial} results in \eqref{eq.f_hessian_pre}.
\end{proof}

\subsubsection{Error Bounds}

Here, we present some simple ``error term" estimates that will be useful in the proof of our main result.

\begin{corollary} \label{thm.error_est}
Let $\xi = \mc{O} ( \zeta )$, where $\zeta: \mc{M} \rightarrow \R^+$.
Then, when $f \ll_{n, \tfr} 1$, we have
\begin{equation} \label{eq.error_est}
\Box \xi = \mc{O} ( \zeta ) \text{,} \qquad \nabla^\alpha f \nabla_\alpha \xi = \mc{O} ( f \zeta ) \text{.}
\end{equation}
\end{corollary}

\begin{proof}
For the first bound in \eqref{eq.error_est}, we write
\[
\Box \xi = g^{\alpha\beta} ( \partial_\alpha \partial_\beta \xi - \Gamma^\mu_{ \alpha \beta } \partial_\mu \xi ) \text{.}
\]
Using \eqref{eq.g}, \eqref{eq.g_inv}, \eqref{eq.Gamma}, and the definition of $\mc{O} ( \zeta )$, we obtain that
\begin{align*}
\Box \xi = \mc{O} ( \rho^2 ) [ \mc{O} ( \rho^{-2} \zeta ) + \mc{O} ( \rho^{-1} ) \cdot \mc{O} ( \rho^{-1} \zeta ) ] = \mc{O} ( \zeta ) \text{.}
\end{align*}
Similarly, for the remaining bound, we have
\[
\nabla^\alpha f \nabla_\alpha \xi = g^{\alpha\beta} \partial_\alpha f \partial_\beta \xi \text{.}
\]
Applying \eqref{eq.g_inv}, \eqref{eq.f_deriv}, and the definition of $\mc{O} ( \zeta )$, we obtain
\begin{align*}
\nabla^\alpha f \nabla_\alpha \xi &= \mc{O} ( \rho^2 ) \cdot \mc{O} ( f \rho^{-1} ) \cdot \mc{O} ( \rho^{-1} \zeta ) = \mc{O} ( f \zeta ) \text{.} \qedhere
\end{align*}
\end{proof}

\subsection{Pseudoconvexity} \label{sec.aads_pseudoconvex}

In this section, we will examine when the level sets of $f$ are pseudoconvex.
First, we recall the geometric definition of pseudoconvexity:

\begin{definition} \label{def.pseudoconvex}
Let $\Sigma \subseteq \mc{M}$ be a smooth hypersurface, and let $H$ be a smooth function defined on a neighborhood of $\mc{M}$ such that $\Sigma$ is precisely the level set $\{ H = 0 \}$.
We say that $\Sigma$ is \emph{pseudoconvex} (with respect to $\Box$ and the direction of increasing $H$) iff for any null vector field $X$ tangent to $\Sigma$, we have
\[ \nabla_{X X} H := \nabla^2 H (X, X) < 0 \text{.} \]
\end{definition}

Roughly, $\Sigma$ is pseudoconvex iff $-H$ is convex with respect to all tangent null directions.
Geometrically, this states that any null geodesic which intersects a point $P$ of $\Sigma$ tangentially will lie in $\{ H < 0 \}$ near $P$.
Note this definition is independent of the choice of $H$, as long as the side in which $H > 0$ does not change.

The following characterization of pseudoconvexity, which is an immediate consequence of Definition \ref{def.pseudoconvex}, will be simpler for computational purposes:

\begin{proposition} \label{thm.pseudoconvex_ex}
Let $\Sigma$ and $H$ be as in Definition \ref{def.pseudoconvex}.
Suppose there exists a smooth function $w$ such that the projection of $- (\nabla^2 H + w \cdot g)$ to $\Sigma$ is positive-definite.
Then, $\Sigma$ is pseudoconvex (with respect to the direction of increasing $H$).
\end{proposition}

\subsubsection{Adapted Frames}

To measure this pseudoconvexity, we define an orthonormal frame (with respect to $g$) adapted to $f$ in the following manner:
\begin{itemize}
\item The first such vector field is the unit normal to the level sets of $f$,
\begin{equation} \label{eq.N_pre}
N := | \nabla^\alpha f \nabla_\alpha f |^{- \frac{1}{2} } \grad f \text{.}
\end{equation}

\item On each level set of $(r, t)$, we fix a local ($g$-)orthonormal frame $E_1, \dots, E_{n-1}$.

\item For the final vector field, we note that both
\begin{equation} \label{eq.T_tilde}
\bar{T} := \partial_t + \tfr f \cos (\tfr t) \cdot \partial_\rho \text{,} \qquad \tilde{T} := \bar{T} - \sum_{ A = 1 }^{n - 1} g ( \bar{T}, E_A ) E_A
\end{equation}
are tangent to the level sets of $f$.
Furthermore, $\tilde{T}$ is normal to the $E_A$'s.
Consequently, we can define our final vector field in our frame to be
\begin{equation} \label{eq.T_pre}
T := | g ( \tilde{T}, \tilde{T} ) |^{-\frac{1}{2}} \tilde{T} \text{.}
\end{equation}
\end{itemize}

In addition, we can make a convenient computational simplification regarding the frame elements $E_A$.
By applying a bounded linear change of a bounded coordinate system $x^A$, we obtain coordinates $y^A$ for which the vector fields $\partial_{ y^A }$ are orthonormal at a single point of $\mc{M}$.
Moreover, the bounded coordinate assumption from Definition \ref{def.aads_horizontal} ensures that this transformation is bounded independently of the choice of $x^A$.
Next, note from \eqref{eq.aads} and \eqref{eq.aads_asymp} that $g_{A B} = \rho^{-2} \mathring{g}_{A B} + \mc{O} (1)$.
Then, by another linear change of coordinates (bounded as long as $\rho \ll_{n, \tfr} 1$), the new coordinates $y^A$ can be chosen to satisfy $E_A = r^{-1} \partial_{ y^A }$ at a single point.

Therefore, we can enlarge our class of bounded coordinate systems such that for each $P \in \mc{M}$, there is a bounded coordinate system satisfying
\begin{equation} \label{eq.EA}
E_A |_P = \rho \partial_{ x^A } |_P \text{.}
\end{equation}
Since we will only be engaging in pointwise tensorial computations, we can, for simplicity, assume that \eqref{eq.EA} holds at each point.

\begin{proposition} \label{thm.onf}
Suppose $f \ll_{n, \tfr} 1$.
Then:
\begin{equation} \label{eq.TN} \begin{split}
N &:= ( 1 - \tfr^2 f^2 )^{ - \frac{1}{2} } \rho \\
&\qquad \cdot \left\{ [ 1 + \mc{O} ( \rho^2 ) ] \partial_\rho + [ \tfr f \cos (\tfr t) + \mc{O} ( \rho^2 ) ] \partial_t + \sum_{ A = 1 }^{n - 1} \mc{O} ( f \rho^2 ) \cdot \partial_{ x^A } \right\} \text{,} \\
T &:= ( 1 - \tfr^2 f^2 )^{ - \frac{1}{2} } \rho \left[ 1 + \frac{1}{2} ( \bar{g}_{tt} - \tfr^2 ) \rho^2 + \mc{O} ( f \rho^2 ) \right] \\
&\qquad \cdot \left\{ \partial_t + \tfr f \cos (\tfr t) \cdot \partial_\rho - \rho^2 \sum_{ A = 1 }^{n - 1} [ \bar{g}_{t A} + \mc{O} ( f \rho ) ] \cdot \partial_{ x^A } \right\} \text{.}
\end{split} \end{equation}
Moreover, we have that
\begin{equation} \label{eq.vf_TN} \begin{split}
( 1 - \tfr^2 f^2 )^\frac{1}{2} \rho \cdot \partial_\rho &= [ 1 + \mc{O} ( \rho^2 ) ] N - [ \tfr f \cos (\tfr t) + \mc{O} ( \rho^2 ) ] T \\
&\qquad + \sum_{ A = 1 }^{n - 1} \mc{O} ( \rho^2 ) \cdot E_A \text{,} \\
( 1 - \tfr^2 f^2 )^\frac{1}{2} \rho \cdot \partial_t &= [ 1 + \mc{O} ( \rho^2 ) ] T - [ \tfr f \cos (\tfr t) + \mc{O} ( \rho^2 ) ] N \\
&\qquad + \sum_{ A = 1 }^{n - 1} \mc{O} ( \rho^2 ) \cdot E_A \text{.}
\end{split} \end{equation}
\end{proposition}

\begin{proof}
For the first identity in \eqref{eq.TN}, we apply \eqref{eq.f_grad} and \eqref{eq.f_grad_sq} to \eqref{eq.N_pre}:
\begin{equation} \label{eq.TN_10} \begin{split}
N &= [ 1 - \tfr^2 f^2 + \mc{O} ( \rho^2 ) ]^{- \frac{1}{2}} \\
&\qquad \cdot \left\{ [ \rho + \mc{O} ( \rho^3 ) ] \partial_\rho + [ \tfr f \rho \cos (\tfr t) + \mc{O} ( \rho^3 ) ] \partial_t + \sum_{ A = 1 }^{n - 1} \mc{O} ( f \rho^3 ) \cdot \partial_{ x^A } \right\} \text{.}
\end{split} \end{equation}
Recalling \eqref{eq.f_deriv_trivial}, we recover the asymptotic expansion
\begin{equation} \label{eq.TN_11}
[ 1 - \tfr^2 f^2 + \mc{O} ( \rho^2 ) ]^{- \frac{1}{2}} = ( 1 - \tfr^2 f^2 )^{-\frac{1}{2}} [ 1 + \mc{O} ( \rho^2 ) ] \text{,}
\end{equation}
and combining \eqref{eq.TN_10} with \eqref{eq.TN_11} yields the desired equality.

The corresponding expansion for $T$ requires more care, since we wish to obtain higher order asymptotics than for $N$.
By \eqref{eq.g}, \eqref{eq.g_t}, and \eqref{eq.T_tilde}, we have
\begin{equation} \label{eq.TN_20}
\tilde{T} = \partial_t + \tfr f \cos (\tfr t) \cdot \partial_\rho - \rho^2 \sum_{ A = 1 }^{n - 1} [ \bar{g}_{t A} + \mc{O} ( f \rho ) ] \partial_{ x^A } \text{.}
\end{equation}
Combining \eqref{eq.TN_20} with \eqref{eq.g}, \eqref{eq.g_t}, \eqref{eq.f_trivial}, and \eqref{eq.f_deriv_trivial}, we then see that
\begin{equation} \label{eq.TN_21} \begin{split}
g (\tilde{T}, \tilde{T}) &= \rho^{-2} [ - 1 + \tfr^2 f^2 + ( \bar{g}_{tt} - \tfr^2 ) \rho^2 + \mc{O} ( f \rho^2 ) ] \\
&= \rho^{-2} ( -1 + \tfr^2 f^2 ) [ 1 - ( \bar{g}_{tt} - \tfr^2 ) \rho^2 + \mc{O} ( f \rho^2 ) ] \text{.}
\end{split} \end{equation}
Recalling again \eqref{eq.f_deriv_trivial}, the above implies
\begin{equation} \label{eq.TN_22}
T = \rho ( 1 - \tfr^2 f^2 )^{- \frac{1}{2} } \left[ 1 + \frac{1}{2} ( \bar{g}_{tt} - \tfr^2 ) \rho^2 + \mc{O} ( f \rho^2 ) \right] \tilde{T} \text{,}
\end{equation}
and the second part of \eqref{eq.TN} follows now from the above and \eqref{eq.TN_20}.

Finally, for \eqref{eq.vf_TN}, we use \eqref{eq.f_deriv_trivial} and \eqref{eq.EA} to rewrite \eqref{eq.TN} as
\begin{equation} \label{eq.vf_TN_0} \begin{split}
\left[ \begin{matrix} N + \sum_{ A = 1 }^{n - 1} \mc{O} ( \rho^2 ) \cdot E_A \\ T + \sum_{ A = 1 }^{n - 1} \mc{O} ( \rho^2 ) \cdot E_A \end{matrix} \right] &= ( 1 - \tfr^2 f^2 )^{ - \frac{1}{2} } \rho \mathbf{A} \left[ \begin{matrix} \partial_\rho \\ \partial_t \end{matrix} \right] \text{,} \\
\mathbf{A} &:= \left[ \begin{matrix} 1 & \tfr f \cos (\tfr t) \\ \tfr f \cos (\tfr t) & 1 \end{matrix} \right] + \mc{O} ( \rho^2 ) \text{,}
\end{split} \end{equation}
and we invert this relation.
In particular, observe (using also \eqref{eq.f_deriv_trivial}) that
\begin{equation} \label{eq.vf_TN_1}
\mathbf{A}^{-1} = ( 1 - \tfr^2 f^2 )^{-1} \left\{ \left[ \begin{matrix} 1 & - \tfr f \cos (\tfr t) \\ - \tfr f \cos (\tfr t) & 1 \end{matrix} \right] + \mc{O} ( \rho^2 ) \right\} \text{.}
\end{equation}
Combining \eqref{eq.vf_TN_0} and \eqref{eq.vf_TN_1} and applying \eqref{eq.f_deriv_trivial} results in \eqref{eq.vf_TN}.
\end{proof}

\subsubsection{The Frame Expansion of $\nabla^2 f$}

The next step is to compute the Hessian of $f$ in terms of the aforementioned orthonormal frames.

\begin{proposition} \label{thm.f_hessian}
Suppose $f \ll_{n, \tfr} 1$.
Then:
\begin{equation} \label{eq.f_hessian} \begin{split}
\nabla_{E_A E_B} f &= - f \delta_{A B} + f \rho^2 ( \bar{g}_{A B} - \delta_{A B} ) + \mc{O} ( f \rho^3 ) \text{,} \\
\nabla_{T E_A} f &= f \rho^2 \bar{g}_{t A} + \mc{O} ( f^2 \rho^2 ) \text{,} \\
\nabla_{T T} f &= f + ( 1 + \bar{g}_{tt} + \tfr^2 ) f \rho^2 + \mc{O} ( f^3 \rho^2 ) \text{,} \\
\nabla_{T N} f &= \mc{O} ( f^2 \rho^2 ) \text{,} \\
\nabla_{N N} f &= f + \mc{O} ( f^3 ) \text{,} \\
\nabla_{N E_A} f &= \mc{O} ( f^2 \rho^2 ) \text{.}
\end{split} \end{equation}
In particular,
\begin{equation} \label{eq.f_box}
\Box f = - (n - 1) f + \mc{O} ( f^3 ) \text{.}
\end{equation}
\end{proposition}

\begin{proof}
First, from \eqref{eq.EA} and the identity for $\nabla_{A B} f$ in \eqref{eq.f_hessian_pre}, we have
\begin{equation} \label{eq.f_hessian_00}
\nabla_{E_A E_B} f = - f \mathring{g}_{A B} - f \rho^2 \mathring{g}_{AB} + \mc{O} ( f \rho^3 ) \text{.}
\end{equation}
Moreover, from \eqref{eq.g} and \eqref{eq.EA}, we see that
\begin{equation} \label{eq.g_ring_frame}
\mathring{g}_{A B} = \delta_{A B} - \rho^2 \bar{g}_{A B} + \mc{O} ( \rho^3 ) \text{.}
\end{equation}
Combining \eqref{eq.f_hessian_00} and \eqref{eq.g_ring_frame} results in the identity for $\nabla_{ E_A E_B } f$ in \eqref{eq.f_hessian}.

Next, by \eqref{eq.EA} and \eqref{eq.TN},
\begin{equation} \label{eq.f_hessian_10} \begin{split}
\nabla_{N E_A} f &= (1 - \tfr^2 f^2)^{- \frac{1}{2} } \mc{I}_{ N E_A } \text{,} \\
\mc{I}_{ N E_A } &:= [ 1 + \mc{O} ( \rho^2 ) ] \rho^2 \nabla_{ \rho A } f + [ \tfr f \cos (\tfr t) + \mc{O} ( \rho^2 ) ] \rho^2 \nabla_{t A} f \\
&\qquad + \sum_{ B = 1 }^{n - 1} \mc{O} ( f \rho^2 ) \cdot \rho^2 \nabla_{A B} f \text{,}
\end{split} \end{equation}
so that applications of \eqref{eq.f_deriv_trivial} and \eqref{eq.f_hessian_pre} yield
\begin{equation} \label{eq.f_hessian_11} \begin{split}
\mc{I}_{N E_A} &= \mc{O} (1) \cdot \rho^2 \nabla_{ \rho A } f + \mc{O} ( f ) \cdot \rho^2 \nabla_{t A} f + \sum_{ B = 1 }^{n - 1} \mc{O} ( f^2 \rho^2 ) \cdot \rho^2 \nabla_{A B} f \\
&= [ \mc{O} ( f^2 \rho^2 ) + \mc{O} ( f \rho^3 ) ] + \mc{O} ( f^2 \rho^3 ) + \mc{O} ( f^2 \rho^2 ) \\
&= \mc{O} ( f^2 \rho^2 ) \text{.}
\end{split} \end{equation}
Applying \eqref{eq.f_deriv_trivial}, \eqref{eq.f_hessian_10}, and \eqref{eq.f_hessian_11} yields the bound for $\nabla_{N E_A} f$ in \eqref{eq.f_hessian}.

The proof for $\nabla_{T E_A} f$ is similar, but we require a more careful expansion in this case.
Applying \eqref{eq.EA} and \eqref{eq.TN}, we can write
\begin{equation} \label{eq.f_hessian_20} \begin{split}
\nabla_{T E_A} f &= ( 1 - \tfr^2 f^2 )^{ - \frac{1}{2} } \left[ 1 + \frac{1}{2} ( \bar{g}_{tt} - \tfr^2 ) \rho^2 + \mc{O} ( f \rho^2 ) \right] \mc{I}_{T E_A} \text{,} \\
\mc{I}_{T E_A} &:= \rho^2 \nabla_{t A} f + \tfr f \cos (\tfr t) \cdot \rho^2 \nabla_{\rho A} f - \rho^2 \sum_{ B = 1 }^{n - 1} [ \bar{g}_{t B} + \mc{O} ( f \rho ) ] \rho^2 \nabla_{A B} f \text{.}
\end{split} \end{equation}
Appealing to \eqref{eq.f_deriv_trivial}, \eqref{eq.f_hessian_pre}, and \eqref{eq.g_ring_frame}, we can expand $\mc{I}$ as
\begin{equation} \label{eq.f_hessian_21} \begin{split}
\mc{I}_{T E_A} &= \mc{O} ( f \rho^3 ) + \mc{O} ( f ) [ \mc{O} ( f^2 \rho^2 ) + \mc{O} ( f \rho^3 ) ] \\
&\qquad - f \rho^2 \sum_{ B = 1 }^{n - 1} [ \bar{g}_{t B} + \mc{O} ( f \rho ) ] [ - \mathring{g}_{A B} + \mc{O} ( \rho^2 ) ] \\
&= \mc{O} ( f^2 \rho^2 ) + f \rho^2 \sum_{ B = 1 }^{n - 1} [ \bar{g}_{t B} + \mc{O} ( f \rho ) ] [ \delta_{A B} + \mc{O} ( \rho^2 ) ] \\
&= f \rho^2 \bar{g}_{t A} + \mc{O} ( f^2 \rho^2 ) \text{.}
\end{split} \end{equation}
Combining \eqref{eq.f_hessian_20} with \eqref{eq.f_hessian_21} yields, as desired,
\begin{equation} \label{eq.f_hessian_22}
\nabla_{T E_A} f = [ 1 + \mc{O} ( f^2 ) ] [ 1 + \mc{O} ( \rho^2 ) ] \mc{I}_{T E_A} = f \rho^2 \bar{g}_{t A} + \mc{O} ( f^2 \rho^2 ) \text{.}
\end{equation}

Next, by \eqref{eq.TN} in conjunction with \eqref{eq.f_trivial} and \eqref{eq.f_deriv_trivial}, we can write
\begin{equation} \label{eq.f_hessian_30} \begin{split}
\nabla_{T T} f &= ( 1 - \tfr^2 f^2 )^{ -1 } [ 1 + ( \bar{g}_{tt} - \tfr^2 ) \rho^2 + \mc{O} ( f \rho^2 ) ] \mc{I}_{T T} \text{,} \\
\mc{I}_{T T} &:= \rho^2 \nabla_{t t} f + 2 \tfr f \cos (\tfr t) \cdot \rho^2 \nabla_{\rho t} f + \tfr^2 ( f^2 - \rho^2 ) \cdot \rho^2 \nabla_{\rho \rho} f \\
&\qquad + \mc{O} ( \rho^2 ) \cdot \sum_{ A = 1 }^{n - 1} \rho^2 \nabla_{t A} f + \mc{O} ( f \rho^2 ) \cdot \sum_{ A = 1 }^{n - 1} \rho^2 \nabla_{ \rho A } f \\
&\qquad + \mc{O} ( \rho^4 ) \cdot \sum_{ A, B = 1 }^{n - 1} \rho^2 \nabla_{A B} f \text{.}
\end{split} \end{equation}
From \eqref{eq.f_trivial}, \eqref{eq.f_deriv_trivial}, and \eqref{eq.f_hessian_pre}, we then have
\begin{equation} \label{eq.f_hessian_31} \begin{split}
\mc{I}_{T T} &= f [ 1 + ( 1 - \tfr^2 ) \rho^2 + 2 \tfr^2 f^2 ] + 2 \tfr^2 ( f^2 - \rho^2 ) f ( -2 - \rho^2 \bar{g}_{tt} ) \\
&\qquad + \tfr^2 ( f^2 - \rho^2 ) f ( 1 + \rho^2 ) + \mc{O} ( f \rho^3 ) \\
&= f [ 1 - \tfr^2 f^2 + ( 1 + 2 \tfr^2 ) \rho^2 + \mc{O} ( f^2 \rho^2 ) ] \text{.}
\end{split} \end{equation}
Combining \eqref{eq.f_hessian_30} and \eqref{eq.f_hessian_31}, and noting that
\[
( 1 - \tfr^2 f^2 )^{ -1 } [ 1 - \tfr^2 f^2 + ( 1 + 2 \tfr^2 ) \rho^2 + \mc{O} ( f^2 \rho^2 ) ] = 1 + ( 1 + 2 \tfr^2 ) \rho^2 + \mc{O} ( f^2 \rho^2 ) \text{,}
\]
we compute, as desired,
\begin{equation} \label{eq.f_hessian_33} \begin{split}
\nabla_{T T} f &= f [ 1 + ( \bar{g}_{tt} - \tfr^2 ) \rho^2 + \mc{O} ( f \rho^2 ) ] [ 1 + ( 1 + 2 \tfr^2 ) \rho^2 + \mc{O} ( f^2 \rho^2 ) ] \\
&= f + ( 1 + \bar{g}_{tt} + \tfr^2 ) f \rho^2 + \mc{O} ( f^3 \rho^2 ) \text{.}
\end{split} \end{equation}

Similarly, again applying \eqref{eq.f_trivial} and \eqref{eq.TN}, we see that
\begin{equation} \label{eq.f_hessian_40} \begin{split}
\nabla_{T N} f &= ( 1 - \tfr^2 f^2 )^{ -1 } \left[ 1 + \frac{1}{2} ( \bar{g}_{tt} - \tfr^2 ) \rho^2 + \mc{O} ( f \rho^2 ) \right] \mc{I}_{T N} \text{,} \\
\mc{I}_{T N} &:= [ \tfr f \cos (\tfr t) + \mc{O} ( \rho^2 ) ] \rho^2 \nabla_{t t} f + [ \tfr f \cos (\tfr t) + \mc{O} ( f \rho^2 ) ] \rho^2 \nabla_{\rho \rho} f \\
&\qquad + [ 1 + \tfr^2 f^2 + \mc{O} ( \rho^2 ) ] \rho^2 \nabla_{\rho t} f + \sum_{ A = 1 }^{n - 1} \mc{O} ( f \rho^2 ) \cdot \rho^2 \nabla_{t A} f \\
&\qquad + \sum_{ A = 1 }^{n - 1} \mc{O} ( \rho^2 ) \cdot \rho^2 \nabla_{ \rho A } f + \sum_{ A, B = 1 }^{n - 1} \mc{O} ( \rho^4 ) \cdot \rho^2 \nabla_{A B} f \text{.}
\end{split} \end{equation}
Using \eqref{eq.f_deriv_trivial} and \eqref{eq.f_hessian_pre} yields
\begin{equation} \label{eq.f_hessian_41} \begin{split}
\mc{I}_{T N} &= \tfr f \cos (\tfr t) \cdot ( f + 2 \tfr^2 f^3 ) + \tfr f \cos (\tfr t) \cdot f \\
&\qquad + ( 1 + \tfr^2 f^2 ) \cdot [ -2 \tfr f^2 \cos (\tfr t ) ] + \mc{O} ( f^2 \rho^2 ) \\
&= \mc{O} ( f^2 \rho^2 ) \text{.}
\end{split} \end{equation}
Note in particular that the top-order terms in $\mc{I}_{T N}$ (containing $f$ but not $\rho$) cancel.
Combining \eqref{eq.f_hessian_40} and \eqref{eq.f_hessian_41} results in the identity for $\nabla_{T N} f$.

Lastly, for $\nabla_{N N} f$, equations \eqref{eq.f_trivial} and \eqref{eq.TN} yield
\begin{equation} \label{eq.f_hessian_50} \begin{split}
\nabla_{N N} f &= ( 1 - \tfr^2 f^2 )^{ -1 } \mc{I}_{N N} \text{,} \\
\mc{I}_{N N} &:= [ 1 + \mc{O} ( \rho^2 ) ] \rho^2 \nabla_{\rho \rho} f + [ \tfr^2 ( f^2 - \rho^2 ) + \mc{O} ( f \rho^2 ) ] \rho^2 \nabla_{t t} f \\
&\qquad + 2 [ \tfr f \cos (\tfr t) + \mc{O} ( \rho^2 ) ] \rho^2 \nabla_{\rho t} f + \sum_{ A = 1 }^{n - 1} \mc{O} ( f^2 \rho^2 ) \cdot \rho^2 \nabla_{ \rho A } f \\
&\qquad + \sum_{ A = 1 }^{n - 1} \mc{O} ( f^3 \rho^2 ) \cdot \rho^2 \nabla_{ t A } f + \sum_{ A, B = 1 }^{n - 1} \mc{O} ( f^4 \rho^4 ) \cdot \rho^2 \nabla_{ A B } f \text{.}
\end{split} \end{equation}
Applying \eqref{eq.f_trivial}, \eqref{eq.f_deriv_trivial}, and \eqref{eq.f_hessian_pre} results in the expansion
\begin{equation} \label{eq.f_hessian_51} \begin{split}
\mc{I}_{N N} &= ( f + f \rho^2 ) + \tfr^2 ( f^2 - \rho^2 ) ( f + 2 \tfr^2 f^3 ) - 2 \tfr^2 f ( f^2 - \rho^2 ) + \mc{O} ( f^2 \rho^2 ) \\
&= f + \mc{O} ( f^3 ) \text{.}
\end{split} \end{equation}
Combining \eqref{eq.f_hessian_50} and \eqref{eq.f_hessian_51} yields the final equation in \eqref{eq.f_hessian}.

Finally, \eqref{eq.f_box} follows from summing the relevant components in \eqref{eq.f_hessian}.
\end{proof}

\subsubsection{The Pseudoconvexity Criterion}

We are finally prepared to state the main criterion for pseudoconvexity of the level sets of $f$.
This is expressed below purely in terms of the asymptotic geometry at infinity.
Note that (aside from $\mathring{g}$ being static) the crucial point is a certain positivity condition for the term $\bar{g}$.

\begin{definition} \label{def.pseudoconvex_ass}
Given $\zeta \in C^\infty ( \mc{M} )$, we say that the $\zeta$-\emph{pseudoconvex property} holds, with constant $K > 0$, iff the symmetric covariant $2$-tensor\footnote{Recall from \eqref{eq.aads_static} that $\mathring{g}_{t t} \equiv -1$ and $\mathring{g}_{t A} \equiv 0$.}
\begin{equation} \label{eq.Pi}
\Pi_\zeta := - ( \bar{g} + \tfr^2 dt^2 + \zeta \mathring{g} )
\end{equation}
on $( 0, \tfr^{-1} \pi ) \times \mc{S}$ satisfies the bound
\begin{equation} \label{eq.pseudoconvex_ass}
\Pi_\zeta ( X, X ) \geq K [ ( X^t )^2 + \mathring{g}_{A B} X^A X^B ] \text{,}
\end{equation}
for any vector field $X := X^t \partial_t + X^A \partial_{ x^A }$ on $( 0, \tfr^{-1} \pi ) \times \mc{S}$.
\end{definition}

\begin{remark}
The property of being $\zeta$-pseudoconvex is a priori tied to the radial coordinate $r$ defining the hypersurfaces of constant $f$. However, one may check that if one makes a change of radial coordinate $r \mapsto \tilde{r} \left(r\right)$ which preserves the asymptotic form of the metric (\ref{eq.aads}), then the level sets of $\tilde{f}$---defined with respect to $\tilde{r}$---remain pseudoconvex near infinity if those of $f$ are. 
\end{remark}

\begin{remark}
As previously noted in Section \ref{sec:intro_aads}, in the case that $( \mc{M}, g )$ is vacuum, the above $\zeta$-pseudoconvexity property is equivalent to the level sets of $t$ at infinity having uniformly positive curvature (with respect to the metric induced from $\mathring{g}$). We will elaborate on this further in \cite{hol_shao:uc_ads_ns, ASfollowup}.
\end{remark}

The following theorem shows that the positivity of $\Pi_\zeta$ in \eqref{eq.Pi} is the crucial determinant of pseudoconvexity for level sets of $f$ (at least for small $f$-values):

\begin{theorem} \label{thm.pseudoconvex}
Let $\zeta \in C^\infty ( \mc{M} )$, and consider the corresponding function
\begin{equation} \label{eq.w}
w_\zeta := ( f + f \rho^2 + f \rho^2 \zeta ) \in C^\infty ( \mc{M} ) \text{,}
\end{equation}
Then, for $f \ll_{n, \tfr} 1$, we have that
\begin{equation} \label{eq.pseudoconvex} \begin{split}
- ( \nabla^2 f + w_\zeta \cdot g ) (E_A, E_B) &= - f \rho^2 ( \bar{g}_{A B} + \zeta \mathring{g}_{A B} ) + \mc{O} ( f \rho^3 ) \text{,} \\
- ( \nabla^2 f + w_\zeta \cdot g ) (T, E_A) &= - f \rho^2 ( \bar{g}_{t A} + \zeta \mathring{g}_{t A} ) + \mc{O} ( f^2 \rho^2 ) \text{,} \\
- ( \nabla^2 f + w_\zeta \cdot g ) (T, T) &= - f \rho^2 ( \bar{g}_{tt} + \tfr^2 + \zeta \mathring{g}_{t t} ) + \mc{O} ( f^3 \rho^2 ) \text{.}
\end{split} \end{equation}
In particular, if $\Pi_\zeta$ defined in \eqref{eq.Pi} is positive-definite, then $\{ f = \varepsilon \}$ is pseudoconvex (with respect to $\Box$ and the direction of increasing $f$) for $0 < \varepsilon \ll_{n, \tfr} 1$.
\end{theorem}

\begin{proof}
The identities \eqref{eq.pseudoconvex} follow immediately from \eqref{eq.f_hessian} and \eqref{eq.g_ring_frame}, since $T$, $N$, and the $E_A$'s are $g$-orthonormal.
For small enough $f$ (and hence small enough $\rho$), the error terms on the right-hand sides in \eqref{eq.pseudoconvex} become negligible.
Thus, the projection of $- ( \nabla^2 f + w_\zeta \cdot g )$ to the level sets of $f$ (i.e., spanned by $T$ and the $E_A$'s) is positive-definite for small $f$ if and only if \eqref{eq.Pi} is positive-definite.
\end{proof}

Later, we will require a more quantitative version of Theorem \ref{thm.pseudoconvex} which uses the full pseudoconvexity condition from Definition \ref{def.pseudoconvex_ass}.
For technical reasons, it will be convenient to work not with $\grad f$, but instead with the following reweighting:\footnote{By using $S$ instead of $\grad f$, one obtains additional cancellations which greatly simplify the proof of the Carleman estimate, Theorem \ref{thm.carleman}, later on. In particular, the factor $h_\zeta$ in \eqref{eq.Sw} is an error term with no leading-order contributions.}

\begin{definition} \label{def.S}
Define the vector field $S$ on $\mc{M}$ by
\begin{equation} \label{eq.S}
S := f^{n - 3} \grad f \text{.}
\end{equation}
\end{definition}

\begin{remark}
Observe that for vector fields $X, Y$ on $\mc{M}$ tangent to the level sets of $f$,
\[
( \nabla S + f^{n - 3} w \cdot g ) (X, Y) = f^{n - 3} ( \nabla^2 f + w \cdot g ) (X, Y) \text{.}
\]
Thus, $\nabla S$ conveys the same information about pseudoconvexity as $\nabla^2 f$.
\end{remark}

\begin{proposition} \label{thm.pseudoconvex_quant}
Suppose the $\zeta$-pseudoconvex property holds for some $\zeta = \mc{O} (1)$ and some $K > 0$.
Moreover, with $w_\zeta$ defined as in \eqref{eq.w}, we set
\begin{equation} \label{eq.pi}
\pi_\zeta := - ( \nabla S + f^{n - 3} w_\zeta \cdot g ) \text{.}
\end{equation}
Then, whenever $f \ll_{n, \tfr} 1$, we have, for any $1$-form $\theta$ on $\mc{M}$,
\begin{equation} \label{eq.pseudoconvex_quant} \begin{split}
\pi_\zeta^{\alpha \beta} \theta_\alpha \theta_\beta &\geq [ K f^{n - 2} \rho^2 + \mc{O} ( f^{n - 1} \rho^2 ) ] \left[ | \theta_T |^2 + \sum_{ A = 1 }^{ n - 1 } | \theta_{ E_A } |^2 \right] \\
&\qquad - [ (n - 1) f^{n - 2} + \mc{O} ( f^n ) ] | \theta_{ N } |^2 \text{.}
\end{split} \end{equation}
\end{proposition}

\begin{proof}
For brevity, we write $\pi$ for $\pi_\zeta$ and $\Pi$ for $\Pi_\zeta$.
First, we note that
\begin{equation} \label{eq.pseudoconvex_quant_0}
\pi^{\alpha\beta} = - f^{n - 3} ( \nabla^{\alpha \beta} f + w_\zeta \cdot g^{ \alpha \beta } ) - ( n - 3 ) f^{n - 4} \nabla^\alpha f \nabla^\beta f \text{.}
\end{equation}
We expand the left-hand side of \eqref{eq.pseudoconvex_quant} using the usual orthonormal frames.
Recalling \eqref{eq.f_deriv_trivial} and \eqref{eq.pseudoconvex} and letting $\Pi$ be as in \eqref{eq.Pi}, we obtain
\begin{equation} \label{eq.pseudoconvex_quant_1} \begin{split}
\pi^{\alpha \beta} \theta_\alpha \theta_\beta &= [ f^{n - 2} \rho^2 \Pi_{t t} + \mc{O} ( f^{n - 1} \rho^2 ) ] ( - \theta_T )^2 \\
&\qquad + \sum_{ A = 1 }^{n - 1} [ f^{n - 2} \rho^2 \Pi_{t A} + \mc{O} ( f^{n - 1} \rho^2 ) ] (- \theta_T ) \theta_{ E_A } \\
&\qquad + \sum_{ A, B = 1 }^{n - 1} [ f^{n - 2} \rho^2 \Pi_{A B} + \mc{O} ( f^{n - 1} \rho^2 ) ] \theta_{ E_A } \theta_{ E_B } \\
&\qquad + \sum_{ A = 1 }^{n - 1} \pi_{N E_A} \theta_N \theta_{ E_A } + \pi_{N T} \theta_N ( - \theta_T ) + \pi_{N N} | \theta_N |^2 \text{.}
\end{split} \end{equation}
Applying the assumption \eqref{eq.pseudoconvex_ass} yields
\begin{equation} \label{eq.pseudoconvex_quant_2} \begin{split}
\pi^{\alpha \beta} \theta_\alpha \theta_\beta &\geq [ K f^{n - 2} \rho^2 + \mc{O} ( f^{n - 1} \rho^2 ) ] \left( | \theta_T |^2 + \sum_{ A = 1 }^{n - 1} | \theta_{ E_A } |^2 \right) \\
&\qquad + \sum_{ A = 1 }^{n - 1} \pi_{N E_A} \theta_N \theta_{ E_A } - \pi_{N T} \theta_N \theta_T + \pi_{N N} | \theta_N |^2 \text{.}
\end{split} \end{equation}

For the remaining components of $\pi$ in \eqref{eq.pseudoconvex_quant_2}, we use \eqref{eq.f_hessian}, the definition \eqref{eq.w} of $w$, and the identity \eqref{eq.pseudoconvex_quant_0} in order to obtain
\begin{equation} \label{eq.pseudoconvex_quant_31}
\pi_{ N E_A } = \mc{O} ( f^{n - 1} \rho^2 ) \text{,} \qquad \pi_{ N T } = \mc{O} ( f^{n - 1} \rho^2 ) \text{.}
\end{equation}
Furthermore, by similar reasoning---in conjunction with the assumption $\zeta = \mc{O} (1)$, the observation $| \nabla_N f |^2 = \nabla^\alpha f \nabla_\alpha f$ and \eqref{eq.f_grad_sq}---we have
\begin{equation} \label{eq.pseudoconvex_quant_32}
\pi_{ N N } = - (n - 1) f^{n - 2} + \mc{O} ( f^n ) \text{.}
\end{equation}
Thus, combining \eqref{eq.pseudoconvex_quant_2}, \eqref{eq.pseudoconvex_quant_31}, and \eqref{eq.pseudoconvex_quant_32} yields
\begin{equation} \label{eq.pseudoconvex_quant_3} \begin{split}
\pi^{\alpha \beta} \theta_\alpha \theta_\beta &\geq [ K f^{n - 2} \rho^2 + \mc{O} ( f^{n - 1} \rho^2 ) ] \left( | \theta_T |^2 + \sum_{ A = 1 }^{n - 1} | \theta_{ E_A } |^2 \right) \\
&\qquad + \sum_{ A = 1 }^{n - 1} \mc{O} ( f^{n - 1} \rho^2 ) \cdot \theta_N \theta_{ E_A } + \mc{O} ( f^{n - 1} \rho^2 ) \cdot \theta_N \theta_T \\
&\qquad - [ ( n - 1 ) f^{n - 2} + \mc{O} ( f^n ) ] | \theta_N |^2 \text{.}
\end{split} \end{equation}
Finally, since \eqref{eq.f_deriv_trivial} implies
\[
\mc{O} ( f^{n - 1} \rho^2 ) \cdot \theta_N \theta_T \gtrsim \mc{O} ( f^{n - 1} \rho^2 ) \cdot | \theta_T |^2 + \mc{O} ( f^n ) \cdot | \theta_N |^2 \text{,}
\]
and similarly for $\mc{O} ( f^{n - 1} \rho^2 ) \cdot \theta_N \theta_{ E_A }$, then \eqref{eq.pseudoconvex_quant_3} becomes \eqref{eq.pseudoconvex_quant}.
\end{proof}

In particular, the above results apply to AdS spacetime---more specifically, to the segments defined in \eqref{eq.ads}, \eqref{eq.ads_g_ring}---as well as to its generalization, \eqref{eq.ads_gen}.

\begin{corollary} \label{thm.ads_pseudoconvex}
Consider the special case of AdS spacetime $( \mc{M}_\AdS, g_\AdS )$; see \eqref{eq.ads}.
Then, whenever $0 < \tfr < 1$, the $\zeta$-pseudoconvex property holds, with
\begin{equation} \label{eq.ads_zeta}
\zeta := - \frac{ 1 - \tfr^2 }{ 2 } \text{,} \qquad \Pi_\zeta = \frac{ 1 - \tfr^2 }{ 2 } ( dt^2 + \mathring{\gamma} ) \text{,}
\end{equation}
and with constant $K = \frac{1}{2} ( 1 - \tfr^2 )$.
\end{corollary}

\begin{remark}
On the other hand, the level sets of $f$ in $( \mc{M}_\AdS, g_\AdS )$ fail to be pseudoconvex if $\tfr > 1$, as one can see that $\Pi_\zeta$ fails to be positive-definite.
For the borderline case $\tfr = 1$, one must expand $\nabla^2 f$ (and hence $g$) to higher order to discern whether these have good sign.
While this is slightly involved, one can show that when $\tfr = 1$, the level sets of $f$ fail to be pseudoconvex.
\end{remark}

\subsection{Horizontal and Mixed Tensors} \label{sec.aads_tensor}

For our upcoming applications to the Einstein equations, we will need to apply our unique continuation results to objects which are tensorial on each level set of $(t, r)$.
Here, we briefly discuss these tensorial objects we will encounter in our results, and we state precisely how covariant derivatives---in particular the wave operator $\Box$---are defined on such objects.

The formalism here is analogous to those found in \cite{shao:bdc_nv, shao:ksp, shao:bdc_nvp}, where similar objects were constructed on null cones and time foliations.

Assume $(\mc{M}, g)$ are as in Definitions \ref{def.aads_manifold} and \ref{def.aads}:
\begin{itemize}
\item A tensor $W$ is ($\mc{S}$-)\emph{horizontal} iff $W$ identifies with a tensor on some level set of $(t, r)$ in $\mc{M}$---that is, some copy of $\mc{S}$.

\item We denote by $T^\mu_\lambda \mc{M}$ the usual $(\mu, \lambda)$-tensor bundle over $\mc{M}$, consisting of all tensors at all points of $\mc{M}$ of rank $(\mu, \lambda)$.

\item We denote by $\ul{T}^m_l \mc{M}$ the ($\mc{S}$-)\emph{horizontal bundle} over $\mc{M}$, consisting of all horizontal tensors of rank $(m, l)$ at all points of $\mc{M}$.
\end{itemize}
In general, for a vector bundle $\mc{V}$ over $\mc{M}$, we let $\Gamma \mc{V}$ denote the space of all smooth sections of $\mc{V}$.
According to this formalism:
\begin{itemize}
\item $\Gamma T^\mu_\lambda \mc{M}$ denotes the usual space of \emph{tensor fields} of rank $(\mu, \lambda)$ over $\mc{M}$.

\item $\Gamma \ul{T}^m_l \mc{M}$ is the space of \emph{horizontal tensor fields} of rank $(m, l)$ over $\mc{M}$.
\end{itemize}
Note in particular that $\Gamma T^0_0 \mc{M} = \Gamma \ul{T}^0_0 \mc{M} = C^\infty ( \mc{M} )$.

Next, we define the \emph{mixed bundles} to be the tensor product bundles
\begin{equation} \label{eq.mixed_bundle} T^\mu_\lambda \ul{T}^m_l \mc{M} := T^\mu_\lambda \mc{M} \otimes \ul{T}^m_l \mc{M} \text{.} \end{equation}
Similarly, we will call an element of $\Gamma T^\mu_\lambda \ul{T}^m_l \mc{M}$ a \emph{mixed tensor field}.
By the duality formulation, we can consider $A \in \Gamma T^\mu_\lambda \ul{T}^m_l \mc{M}$ as a $C^\infty ( \mc{M} )$-multilinear map on the appropriate number of standard and horizontal vector fields and $1$-forms.

\begin{remark}
Readers who are interested only in scalar wave equations can skip the subsequent discussion on covariant formulations.
In particular, the reader can simply replace all instances of horizontal sections $\Gamma \ul{T}^m_l \mc{M}$ by the space $C^\infty ( \mc{M} )$ of smooth scalar functions.
Furthermore, in this case, the mixed tensor bundles $T^\mu_\lambda \ul{T}^m_l \mc{M}$ reduce to the usual (spacetime) tensor bundles $T^\mu_\lambda \mc{M}$.
\end{remark}

\subsubsection{Covariant Structures}

Recall that $g$ induces a bundle metric on any tensor bundle $T^\mu_\lambda \mc{M}$.
Similarly, the induced metrics $\slashed{g}$ on the level sets of $(t, r)$ induce bundle metrics on any horizontal bundle $\ul{T}^m_l \mc{M}$.
From the above, one can now naturally define the \emph{mixed bundle metric} $g$ on $T^\mu_\lambda \ul{T}^m_l \mc{M}$, first by
\begin{equation} \label{eq.mixed_metric} g ( A_1 \otimes B_1, A_2 \otimes B_2 ) := g (A_1, A_2) \cdot \slashed{g} (B_1, B_2) \text{,} \end{equation}
and then linearly extended to all mixed tensors.
In terms of indices, this corresponds precisely to $g$-metric contractions for all the spacetime components and $\slashed{g}$-metric contractions for all the horizontal components.

Next, recall the Levi-Civita connection $\nabla$ induces a bundle connection---also denoted $\nabla$---on any $T^\mu_\lambda \mc{M}$, and this $\nabla$ is compatible with the metric, i.e.,
\[
\nabla g = 0 \text{.}
\]
We can also define analogous \emph{horizontal connections}---also denoted by $\nasla$---on the horizontal bundles.
Given any vector field $X \in \Gamma T^1_0 \mc{M}$, we define the following:
\begin{itemize}
\item For a scalar $f \in \Gamma \ul{T}^m_l \mc{M} = C^\infty ( \mc{M} )$, we define $\nasla_X f = Xf$, as usual.

\item For a horizontal vector field $Y \in \Gamma \ul{T}^1_0 \mc{M}$, we define $\nasla_X Y$ to be the orthogonal projection of $\nabla_X Y$ onto the tangent spaces of the $(t, r)$-level sets.

\item From the above, $\nasla$ can then be defined on all $\Gamma \ul{T}^m_l \mc{M}$ in the usual way, via Leibniz rule considerations.
In particular, for a covariant $A \in \Gamma \ul{T}^0_l \mc{M}$,
\begin{align*}
\nasla_X A ( Y_1, \ldots, Y_l ) &:= X [ A ( Y_1, \ldots, Y_l ) ] - A ( \nasla_X Y_1, Y_2, \ldots, Y_l ) \\
&\qquad - \ldots - A ( Y_1, \ldots, \nasla_X Y_l ) \text{.}
\end{align*}
where $Y_1, \ldots, Y_k$ are arbitrary horizontal vector fields.
\end{itemize}
Note that if $X$ is itself horizontal, then $\nasla_X$ is precisely the induced covariant derivative on the level sets of $(t, r)$, thus our choice of notation $\nasla$ makes sense.
Furthermore, we note that $\nasla \slashed{g} \equiv 0$, i.e., $\nasla$ is compatible with the induced metric.
More specifically, $\nasla_X \slashed{g}$ vanishes \emph{for all spacetime directions $X$}.

We can now canonically combine the connections $\nabla$ and $\nasla$ to obtain \emph{mixed connections} $\nabla$ on the mixed bundles.
The basic idea is to have $\nabla$ behave like the usual spacetime connection $\nabla$ on the spacetime components and like $\nasla$ on the horizontal components.
Indeed, we first define the mixed connection $\nabla$ by
\begin{equation} \label{eq.mixed_connection} \nabla_X ( A \otimes B ) := \nabla_X A \otimes B + A \otimes \nasla_X B \text{,} \end{equation}
where $A$ and $B$ denote spacetime and horizontal tensor fields, respectively, and then we linearly extend $\nabla_X$ to all the remaining mixed tensor fields.
Again, one can show that $\nabla$ is compatible with the mixed metrics:
\begin{equation} \label{eq.mixed_covariant} \nabla g \equiv 0 \text{.} \end{equation}

In practice, the most useful representation of mixed covariant derivatives is the following differentiation formula for covariant mixed tensor fields: if $A \in \Gamma T^0_\lambda \ul{T}^0_l \mc{M}$, if $X, Z_1, \dots, Z_\lambda \in \Gamma T^1_0 \mc{M}$, and if $Y_1, \ldots, Y_l \in \Gamma \ul{T}^1 \mc{M}$, then
\begin{equation} \label{eq.mixed_deriv} \begin{split}
\nabla_X A ( Z_1, \dots, Z_\lambda; Y_1, \dots, Y_l ) &= X [ A ( Z_1, \dots, Z_\lambda; Y_1, \dots, Y_l ) \\
&\qquad - A ( \nabla_X Z_1, \dots, Z_\lambda; Y_1, \dots, Y_l ) - \dots \\
&\qquad - A ( Z_1, \dots, \nabla_X Z_\lambda; Y_1, \dots, Y_l ) \\
&\qquad - A ( Z_1, \dots, Z_\lambda; \nasla_X Y_1, \dots, Y_l ) - \ldots \\
&\qquad - A ( Z_1, \dots, Z_\lambda; Y_1, \dots, \nasla_X Y_l ) \text{.}
\end{split} \end{equation}

Now, given any $A \in \Gamma T^\mu_\lambda \ul{T}^m_l \mc{M}$:
\begin{itemize}
\item We define its \emph{mixed covariant differential} $\nabla A \in \Gamma T^\mu_{ \lambda + 1 } \ul{T}^m_l \mc{M}$ to be the mixed tensor field mapping a vector field $X$ to $\nabla_X A$.

\item We can then define higher-order mixed differentials of $A$ by iterating this operator $\nabla$.
For instance, for the second differential, $\nabla^2 A := \nabla \nabla A$, the outer ``$\nabla$" is a mixed differential acting on $\Gamma T^\mu_{ \lambda + 1 } \ul{T}^m_l \mc{M}$.

\item Consequently, we can define $\Box A \in \Gamma T^\mu_\lambda \ul{T}^m_l \mc{M}$ as the $g$-trace of $\nabla^2 A$, with the trace being applied to the two $\nabla^2$-components.
\end{itemize}
Finally, we define the \emph{mixed curvature} by the failure of second covariant differentials to commute: given $A$ as above and $X, Y \in \Gamma T^1_0 \mc{M}$, we define
\begin{equation} \label{eq.mixed_curv} \mc{R} A \in \Gamma T^\mu_{ \lambda + 2 } \ul{T}^m_l \mc{M} \text{,} \qquad \mc{R}_{XY} [A] := \nabla_{XY} A - \nabla_{YX} A \text{.} \end{equation}

\begin{proposition} \label{thm.mixed_curv_hor}
Let $\phi \in \Gamma \ul{T}^0_l \mc{M}$.
Then, given any spacetime vector fields $X, Y$ and horizontal vector fields $Z_1, \dots, Z_l$, we have the following formula:
\begin{equation} \label{eq.mixed_curv_hor} \begin{split}
\mc{R}_{XY} \phi (Z_1, \dots, Z_l) &= - \phi ( \nasla_X ( \nasla_Y Z_1 ) - \nasla_Y ( \nasla_X Z_1 ) - \nasla_{ [X, Y] } Z_1, \dots, Z_l ) \\
&\qquad - \dots \\
&\qquad - \phi ( Z_1, \dots, \nasla_X ( \nasla_Y Z_l ) - \nasla_Y ( \nasla_X Z_l ) - \nasla_{ [X, Y] } Z_l ) \text{.}
\end{split} \end{equation}
In particular, if both $X$ and $Y$ are also horizontal, then \eqref{eq.mixed_curv_hor} reduces to the usual Riemann curvature on the level sets of $(t, r)$.
\end{proposition}

\begin{proof}
This follows from \eqref{eq.mixed_deriv} and a direct computation.
\end{proof}

\subsubsection{Tensor Notations}

In the upcoming development, we will work with horizontal tensor fields of arbitrary rank.
For future convenience, we introduce some notational conventions to simplify how these objects are expressed.

First, we will use capital Latin letters to denote \emph{horizontal multi-indices}, i.e., a collection of zero or more horizontal indices.
The number of indices represented will be apparent from context.
For instance, if $\phi \in \Gamma \ul{T}^0_l \mc{M}$, then $\phi_I$ denotes a single scalar component of $\phi$, and $I$ represents $l$ horizontal indices.
Repeated indices represent summations over all individual indices; for $\phi$ as above, then
\[
\phi^I \phi_I := \phi^{A_1 \dots A_l} \phi_{A_1 \dots A_l} \text{.}
\]

Furthermore, for spherical tensors, we let $| \cdot |$ denote the pointwise tensor norm, with respect to the induced spherical metrics $\slashed{g}$.
Thus, if $\psi \in \Gamma \ul{T}^m_l \mc{M}$, then
\begin{equation} \label{eq.tensor_norm}
| \psi |^2 := \psi^I{}_J \psi_I{}^J \text{.}
\end{equation}

\begin{remark}
For readers interested only in the scalar case, multi-indices can essentially be ignored.
Indeed, $| \psi |$ is simply the absolute value, and $\phi^I{}_J \psi_I{}^J = \phi \cdot \psi$.
Moreover, in this case, the curvature operator $\phi \mapsto \mc{R} \phi$ is trivially zero.
\end{remark}

One consequence of \eqref{eq.mixed_covariant} is that mixed covariant derivatives commute with contractions in both spacetime and horizontal components.
Thus, the usual product rule considerations hold for mixed tensor fields.
For instance, if $\psi \in \Gamma \ul{T}^0_l \mc{M}$, then
\begin{equation} \label{eq.product_rule}
\nabla_\beta ( \nabla^\alpha \psi^I \nabla_\alpha \psi_I ) = 2 \nabla^\alpha \psi^I \nabla_{\beta \alpha} \psi_I \text{,}
\end{equation}
where ``$\nabla_\beta$" on the left-hand side is the usual directional derivative (i.e., the mixed derivative on $\Gamma T^0_0 \ul{T}^0_0 \mc{M}$), while ``$\nabla_\beta$" on the right-hand side is the mixed derivative on $\Gamma T^0_1 \ul{T}^0_l \mc{M}$.
Furthermore, by integrating formulae such as \eqref{eq.product_rule} over spacetime regions, one sees that the usual integration by parts processes that hold for spacetime tensor fields can be directly extended to these mixed tensor fields.

\subsubsection{Curvature Estimates}

Finally, we prove some estimates involving the curvature operator $\mc{R}$ that will be used in our main results.

\begin{proposition} \label{thm.mixed_curv_nor}
Let $\phi \in \Gamma \ul{T}^0_l \mc{M}$.
Then, for $f \ll_{n, \tfr} 1$, we have
\begin{equation} \label{eq.mixed_curv_f}
| \mc{R}_{N T} \phi | \lesssim \rho^3 | \phi | \text{,} \qquad | \mc{R}_{N E_A} \phi | \lesssim \rho^3 | \phi | \text{.}
\end{equation}
\end{proposition}

\begin{proof}
Again, as we are dealing with pointwise tensorial computations, we will for convenience always work with frame and coordinate systems which are related at the point in question by \eqref{eq.EA}.
The first step will be to prove that
\begin{equation} \label{eq.mixed_curv_nor}
| \mc{R}_{\rho a} \phi | \lesssim \rho | \phi | \text{,} \qquad | \mc{R}_{t A} \phi | \lesssim \rho | \phi | \text{.}
\end{equation}

Since $\nasla_\alpha \partial_{ x^A }$ is the orthogonal projection of $\nabla_\alpha \partial_{ x^A }$ to the $(t, r)$-level sets, then
\begin{equation} \label{eq.mixed_curv_nor_0}
\nasla_\alpha \partial_{ x^A } = \sum_{ B = 1 }^{ n - 1 } g ( \Gamma^\mu_{ \alpha A } \partial_\mu, E_B ) E_B = \sum_{ B = 1 }^{ n - 1 } \mc{O} ( \rho^2 ) \cdot \partial_{ x^B } + \Gamma^B_{ \alpha A } \partial_{ x^B } \text{,}
\end{equation}
where we applied \eqref{eq.g}, \eqref{eq.Gamma}, and \eqref{eq.EA}.
Now, by \eqref{eq.Gamma} and \eqref{eq.mixed_curv_nor_0},
\begin{equation} \label{eq.mixed_curv_nor_10} \begin{split}
\nasla_\rho ( \nasla_a \partial_{ x^A } ) - \nasla_a ( \nasla_\rho \partial_{ x^A } ) &= \partial_\rho \Gamma^B_{a A} \cdot \partial_{ x^B } - \partial_a \Gamma^B_{\rho A} \cdot \partial_{ x^B } + \Gamma^B_{a A} \nasla_\rho \partial_{ x^B } \\
&\qquad - \Gamma^B_{\rho A} \nasla_a \partial_{ x^B } + \sum_{ B = 1 }^{n - 1} \mc{O} ( \rho^2 ) \cdot \nasla_\rho \partial_{ x^B } \\
&\qquad + \sum_{ B = 1 }^{n - 1} \mc{O} ( \rho^2 ) \cdot \nasla_a \partial_{ x^B } + \sum_{ C = 1 }^{n - 1} \mc{O} ( \rho ) \cdot \partial_{ x^C } \\
&= \partial_\rho \Gamma^B_{a A} \cdot \partial_{ x^B } - \partial_a \Gamma^B_{\rho A} \cdot \partial_{ x^B } + \Gamma^B_{a A} \Gamma^C_{\rho B} \cdot \partial_{ x^C } \\
&\qquad - \Gamma^B_{\rho A} \Gamma^C_{a B} \cdot \partial_{ x^C } + \sum_{ C = 1 }^{n - 1} \mc{O} ( \rho ) \cdot \partial_{ x^C } \text{.}
\end{split} \end{equation}
By \eqref{eq.Gamma_deriv}, we see that
\begin{equation} \label{eq.mixed_curv_nor_11}
\partial_\rho \Gamma^B_{a A} = \mc{O} ( \rho ) \text{,} \qquad \partial_a \Gamma^B_{\rho A} = \mc{O} ( \rho ) \text{.}
\end{equation}
Moreover, by \eqref{eq.Gamma},
\begin{equation} \label{eq.mixed_curv_nor_12}
\Gamma^B_{a A} \Gamma^C_{ \rho B } - \Gamma^B_{\rho A} \Gamma^C_{a B} = - \rho^{-1} \Gamma^C_{a A} + \rho^{-1} \Gamma^C_{a A} + \mc{O} ( \rho ) = \mc{O} ( \rho ) \text{.}
\end{equation}
Thus, combining \eqref{eq.mixed_curv_nor_10}-\eqref{eq.mixed_curv_nor_12} yields
\begin{equation} \label{eq.mixed_curv_nor_1}
\nasla_\rho ( \nasla_a \partial_{ x^A } ) - \nasla_a ( \nasla_\rho \partial_{ x^A } ) = \sum_{ B = 1 }^{n - 1} \mc{O} ( \rho ) \cdot \partial_{ x^B } \text{.}
\end{equation}
Recalling \eqref{eq.mixed_curv_hor} and applying \eqref{eq.EA} and \eqref{eq.mixed_curv_nor_1} results in the first estimate in \eqref{eq.mixed_curv_nor}.

For the remaining bound in \eqref{eq.mixed_curv_nor}, we apply a similar computation to obtain
\begin{equation} \label{eq.mixed_curv_nor_20} \begin{split}
\nasla_t ( \nasla_A \partial_{ x^B } ) - \nasla_A ( \nasla_t \partial_{ x^B } ) &= \partial_t \Gamma^C_{A B} \cdot \partial_{ x^C } - \partial_A \Gamma^C_{t B} \cdot \partial_{ x^C } + \Gamma^C_{A B} \Gamma^D_{t C} \cdot \partial_{ x^D } \\
&\qquad - \Gamma^C_{t B} \Gamma^D_{A C} \cdot \partial_{ x^D } + \sum_{ C = 1 }^{n - 1} \mc{O} ( \rho ) \cdot \partial_{ x^C } \text{.}
\end{split} \end{equation}
Applying \eqref{eq.Gamma}, \eqref{eq.Gamma_t}, and \eqref{eq.Gamma_deriv} yields
\begin{equation} \label{eq.mixed_curv_nor_2}
\nasla_t ( \nasla_A \partial_{ x^B } ) - \nasla_A ( \nasla_t \partial_{ x^B } ) = \sum_{ C = 1 }^{n - 1} \mc{O} ( \rho ) \cdot \partial_{ x^C } \text{,}
\end{equation}
and the second bound in \eqref{eq.mixed_curv_nor} follows.
This completes the proof of \eqref{eq.mixed_curv_nor}.

Next, we use \eqref{eq.f_deriv_trivial}, \eqref{eq.EA}, \eqref{eq.TN} to bound:
\begin{equation} \label{eq.mixed_curv_f_1}
| R_{N E_A} \phi | \lesssim \rho^2 | R_{ \rho A } \phi | + f \rho^2 | R_{ t A } \phi | + f^2 \rho^4 \sum_{ B = 1 }^{n - 1} | R_{A B} \phi | \text{.}
\end{equation}
The first two terms on the right-hand side are controlled using \eqref{eq.mixed_curv_nor}.
For the last term, we observe $R_{A B} \phi$ is the Riemann curvature operator for the $(t, r)$-level sets applied to $\phi$ (see Proposition \ref{thm.mixed_curv_hor}), and it follows (via \eqref{eq.aads_bdd} and \eqref{eq.Gamma}) that
\begin{equation} \label{eq.mixed_curv_curv}
| R_{A B} \phi | \lesssim | \phi | \text{.}
\end{equation}
Combining \eqref{eq.mixed_curv_f_1} with the above results in the second bound in \eqref{eq.mixed_curv_f}.

Similarly, by \eqref{eq.f_deriv_trivial} and \eqref{eq.TN},
\begin{equation} \label{eq.mixed_curv_f_2}
| R_{N T} \phi | \lesssim \rho^2 | R_{ \rho t } \phi | + \sum_{ A = 1 }^{n - 1} ( \rho^4 | R_{ \rho A } \phi | + f \rho^4 | R_{ t A } \phi | ) + f^2 \rho^6 \sum_{ A, B = 1 }^{n - 1} | R_{A B} \phi | \text{.}
\end{equation}
Applying \eqref{eq.mixed_curv_nor}, \eqref{eq.mixed_curv_curv}, \eqref{eq.mixed_curv_f_2} yields the first bound in \eqref{eq.mixed_curv_f}.
\end{proof}

\begin{corollary} \label{thm.mixed_curv_ex}
Let $\phi \in \Gamma \ul{T}^0_l \mc{M}$.
Then, for $f \ll_{n, \tfr} 1$, we have
\begin{equation} \label{eq.mixed_curv_ex}
| S^\alpha \nabla^\beta \phi^I \mc{R}_{\alpha \beta} \phi_I | \lesssim \mc{O} ( f^{n - 2} \rho^3 ) \cdot \left( | \nabla_T \phi |^2 + \sum_{ A = 1 }^{n - 1} | \nabla_{ E_A } \phi |^2 + | \phi |^2 \right) \text{,}
\end{equation}
where $S$ is the vector field from Definition \ref{def.S}.
\end{corollary}

\begin{proof}
Expanding in terms of the usual orthonormal frames, we have
\begin{equation} \label{eq.mixed_curv_ex_0}
S^\alpha \nabla^\beta \phi^I \mc{R}_{\alpha \beta} \phi_I = f^{n - 3} \nabla_N f \left( - \nabla_T \phi^I \mc{R}_{ N T } \phi_I + \sum_{ A = 1 }^{n - 1} \nabla_{ E_A } \phi^I \mc{R}_{ N E_A } \phi_I \right) \text{.}
\end{equation}
Now, by \eqref{eq.f_grad_sq}, we have that $| \nabla_N f | = \mc{O} ( f )$.
As a result,
\begin{equation} \label{eq.mixed_curv_ex_1} \begin{split}
| S^\alpha \nabla^\beta \phi^I \mc{R}_{\alpha \beta} \phi_I | \lesssim \mc{O} ( f^{n - 2} ) \cdot \left( | \nabla_T \phi | | \mc{R}_{ N T } \phi | + \sum_{ A = 1 }^{n - 1} | \nabla_{ E_A } \phi | | \mc{R}_{ N E_A } \phi | \right) \text{,}
\end{split} \end{equation}
and combining \eqref{eq.mixed_curv_f} and \eqref{eq.mixed_curv_ex_1} yields \eqref{eq.mixed_curv_ex}.
\end{proof}

\section{The Carleman Estimates} \label{sec.carleman}

In this section, we prove the following Carleman inequalities:

\begin{theorem} \label{thm.carleman}
Let $n \in \N$ and $y, r_0 > 0$, and fix constants $p, \kappa \in \R$ satisfying
\begin{equation} \label{eq.p_kappa}
0 < p < 1 \text{,} \qquad \kappa \geq \frac{n - 1}{2} \text{.}
\end{equation}
Let $( \mc{I}_\tfr, \mathring{g} )$ be an $n$-dimensional segment of bounded static AdS infinity (see Definition \ref{def.aads_infinity}), and let $\mc{M} := (r_0, \infty) \times \mc{I}_\tfr$.
Suppose in addition that:
\begin{itemize}
\item $( \mc{M}, g )$ is an admissible aAdS segment, as described in Definition \ref{def.aads}.

\item The $\zeta$-pseudoconvex property (see Definition \ref{def.pseudoconvex_ass}) is satisfied for some function $\zeta = \mc{O} (1)$ on $\mc{M}$, with constant $K > 0$.
\end{itemize}
In addition, fix sufficiently small constants $f_0, \rho_0$ satisfying
\begin{equation} \label{eq.f0_rho0}
0 < \rho_0 \ll f_0 \ll_{n, \tfr, p, K} 1 \text{,}
\end{equation}
and let $\Omega_{ f_0, \rho_0 }$ denote the region\footnote{Note $\Omega_{ f_0, \rho_0 }$ is relatively compact, hence all integrals we consider over $\Omega_{ f_0, \rho_0 }$ will be finite.} 
\begin{equation} \label{eq.Omega}
\Omega_{ f_0, \rho_0 } := \{ f < f_0 \text{, } \rho > \rho_0 \} \text{,}
\end{equation}
Then, for any $\phi \in \Gamma \ul{T}^0_l \mc{M}$, $l \geq 0$ with both $\phi$ and $\nabla \phi$ vanishing on $\{ f = f_0 \}$:
\begin{itemize}
\item If $\sigma \in \R$ and $\lambda \in [1 + \kappa, \infty)$, then there exist constants $C, \mc{C} > 0$, depending on $n$, $\tfr$, $p$, and $K$, such that the following inequality holds:
\begin{equation} \label{eq.carleman_lambda} \begin{split}
&\int_{ \Omega_{ f_0, \rho_0 } } f^{n - 2 - 2 \kappa} e^\frac{ - 2 \lambda f^p }{p} f^{-p} | ( \Box + \sigma ) \phi |^2 \\
&\qquad + \mc{C} \lambda ( \lambda^2 + | \sigma | ) \int_{ \{ \rho = \rho_0 \} } [ | \nabla_t ( \rho^{ - \kappa } \phi ) |^2 + | \nabla_\rho ( \rho^{ - \kappa } \phi ) |^2 + | \rho^{- \kappa - 1} \phi |^2 ] d \mathring{g} \\
&\quad \geq C \lambda \int_{ \Omega_{ f_0, \rho_0 } } f^{n - 2 - 2 \kappa} e^\frac{ - 2 \lambda f^p }{p} ( \rho^4 | \nabla_t \phi |^2 + \rho^4 | \nabla_\rho \phi |^2 + \rho^2 | \nasla \phi |^2 ) \\
&\quad \qquad + \lambda [ \kappa^2 - ( n - 2 ) \kappa+ \sigma - (n - 1) ] \int_{ \Omega_{ f_0, \rho_0 } } f^{n - 2 - 2 \kappa} e^\frac{ - 2 \lambda f^p }{p} | \phi |^2 \\
&\quad \qquad + C \lambda^3 \int_{ \Omega_{ f_0, \rho_0 } } f^{n - 2 - 2 \kappa} e^\frac{ - 2 \lambda f^p }{p} f^{2p} | \phi |^2 \text{.}
\end{split} \end{equation}

\item If $\kappa \gg n$, then there exist $C, \mc{C} > 0$, depending on $n$, $\tfr$, and $K$, such that
\begin{equation} \label{eq.carleman_beta} \begin{split}
&\int_{ \Omega_{ f_0, \rho_0 } } f^{n - 2 - 2 \kappa} | \Box \phi |^2 \\
&\qquad + \mc{C} \kappa^3 \int_{ \{ \rho = \rho_0 \} } [ | \nabla_t ( \rho^{ - \kappa } \phi ) |^2 + | \nabla_\rho ( \rho^{ - \kappa } \phi ) |^2 + | \rho^{- \kappa - 1} \phi |^2 ] d \mathring{g} \\
&\quad \geq C \kappa \int_{ \Omega_{ f_0, \rho_0 } } f^{n - 2 - 2 \kappa} ( \rho^4 | \nabla_t \phi |^2 + \rho^4 | \nabla_\rho \phi |^2 + \rho^2 | \nasla \phi |^2 ) \\
&\quad \qquad + C \kappa^3 \int_{ \Omega_{ f_0, \rho_0 } } f^{n - 2 - 2 \kappa} | \phi |^2 \text{.}
\end{split} \end{equation}
\end{itemize}
\end{theorem}

\begin{remark}
Again, readers interested only in scalar equations can assume throughout that $\phi \in C^\infty (\mc{M})$ and ignore all multi-indices ``$I$" in the upcoming proof.
\end{remark}

\begin{remark}
As indicated by the notation in \eqref{eq.f0_rho0}, the constant $f_0$ will be determined in the course of the proof depending only on the fixed parameters $n, y, p, K$.
It corresponds to choosing a region sufficiently close to the boundary where both the AdS asymptotics and the pseudoconvexity property can be quantitatively exploited. 
\end{remark}

The remainder of this section is dedicated to the proof of Theorem \ref{thm.carleman}.
We focus our attention on \eqref{eq.carleman_lambda}, which is proved in Sections \ref{sec.carleman_conj} and \ref{sec.carleman_comp}.
The proof of \eqref{eq.carleman_beta} is analogous but simpler; we only briefly summarize this in Section \ref{sec.carleman_beta}.
From now on, we adopt the assumptions of Theorem \eqref{thm.carleman}, and we fix $\sigma \in \R$ and $\lambda \geq 1 + \kappa$.

\subsection{The Conjugated Inequality} \label{sec.carleman_conj}

As is standard in the proofs of Carleman-type estimates, the main idea is to not work with $\phi$ itself, but rather with $\phi$ multiplied by a specific weight.
To be more specific, we define the function
\begin{equation} \label{eq.psi_phi}
\psi := e^{-F} \phi \text{,}
\end{equation}
where $F$ is the following reparametrization of $f$:
\begin{equation} \label{eq.F}
F := \kappa \cdot \log f + \lambda p^{-1} f^p \text{.}
\end{equation}
Letting $'$ denote differentiation with respect to $f$, then $F$ satisfies
\begin{equation} \label{eq.F_deriv}
F' = \kappa f^{-1} + \lambda f^{-1 + p} \text{,}  \qquad F'' = - \kappa f^{-2} - \lambda (1 - p) f^{-2 + p} \text{.}
\end{equation}

The aim of this subsection is to prove a preliminary inequality for $\psi$.
In order to state this estimate succintly, we first define the following:
\begin{itemize}
\item Recalling $w_\zeta$ and $S$, as given in \eqref{eq.w} and \eqref{eq.S}, respectively, we define
\begin{equation} \label{eq.Sw}
S_\zeta \psi := \nabla_S \psi + h_\zeta \psi \text{,} \qquad h_\zeta := f^{n - 3} w_\zeta + \frac{1}{2} \nabla^\alpha S_\alpha \in C^\infty ( \mc{M} ) \text{.}
\end{equation}

\item We also define the conjugated wave operator $\mc{L}$:
\begin{equation} \label{eq.L}
\mc{L} \psi := e^{-F} ( \Box + \sigma ) ( e^F \psi ) = e^{-F} ( \Box + \sigma ) \phi \text{.}
\end{equation}

\item Finally, as in Proposition \ref{thm.rho_grad}, we let $\mc{N} := | \nabla^\alpha \rho \nabla_\alpha \rho |^{- \frac{1}{2}} \grad \rho$ denote the outer-pointing unit normal to the level sets of $\rho$.
\end{itemize}
The main estimate of this subsection can now be expressed as follows.

\begin{lemma} \label{thm.psi_est}
There exists $C > 0$, depending on $n$, $\tfr$, $p$, $K$, such that when $f < f_0$,
\begin{equation} \label{eq.psi_est} \begin{split}
\lambda^{-1} f^{n - 2 - p} | \mc{L} \psi |^2 &\geq C \lambda f^{n - 2 + p} | \nabla_N \psi |^2 + C f^{n - 2} \rho^2 ( | \nabla_T \psi |^2 + | \nasla \psi |^2 ) \\
&\qquad + [ \kappa^2 - ( n - 2 ) \kappa + \sigma - (n - 1) ] f^{n - 2} | \psi |^2 \\
&\qquad + C ( \lambda f^{n - 2 + p} + \lambda^2 f^{n - 2 + 2p} ) | \psi |^2 + \nabla^\beta P_\beta \text{,}
\end{split} \end{equation}
where the $1$-form $P$ satisfies, for some $\mc{C} > 0$ depending on $n$, $y$, $p$,
\begin{equation} \label{eq.psi_est_boundary}
P ( \mc{N} ) \leq \mc{C} f^{n - 2} \rho^2 ( | \nabla_t \psi |^2 + | \nabla_\rho \psi |^2 ) + \mc{C} ( \lambda^2 + | \sigma | ) f^{n - 2} | \psi |^2 \text{.}
\end{equation}
\end{lemma}

Lemma \ref{thm.psi_est} is proved in the remainder of this subsection.

\subsubsection{The Stress-Energy Tensor}

Recall the tensor field $\pi_\zeta$ defined in \eqref{eq.pi}, and let $Q$ be the stress-energy tensor for the (free) wave equation with respect to $\psi$:
\begin{equation} \label{eq.emt}
Q_{\alpha\beta} := \nabla_\alpha \psi^I \nabla_\beta \psi_I - \frac{1}{2} g_{\alpha\beta} \nabla^\mu \psi^I \nabla_\mu \psi_I \text{.}
\end{equation}
Direct computations yield that
\begin{equation} \label{eq.set_1} \begin{split}
S^\alpha \nabla^\beta Q_{\alpha\beta} &= \Box \psi^I \nabla_S \psi_I + S^\alpha \nabla^\beta \psi_I \mc{R}_{\alpha \beta} \psi^I \text{,} \\
\nabla^\beta ( Q_{\alpha\beta} S^\alpha ) &= \Box \psi^I \nabla_S \psi_I - \pi_\zeta^{\alpha \beta} \nabla_\alpha \psi^I \nabla_\beta \psi_I - h_\zeta \cdot \nabla^\mu \psi^I \nabla_\mu \psi_I \\
&\qquad + S^\alpha \nabla^\beta \psi_I \mc{R}_{\alpha \beta} \psi^I \text{.}
\end{split} \end{equation}
Furthermore, defining the current
\begin{equation} \label{eq.P_sharp}
P^Q_\beta := Q_{\alpha\beta} S^\alpha + \frac{1}{2} h_\zeta \cdot \nabla_\beta | \psi |^2 - \frac{1}{2} \nabla_\beta h_\zeta \cdot | \psi |^2 \text{,}
\end{equation}
we see from \eqref{eq.set_1} that
\begin{equation} \label{eq.set_2}
\nabla^\beta P^Q_\beta = \Box \psi^I S_\zeta \psi_I - \pi_\zeta^{\alpha \beta} \nabla_\alpha \psi^I \nabla_\beta \psi_I + S^\alpha \nabla^\beta \psi_I \mc{R}_{\alpha \beta} \psi^I - \frac{1}{2} \Box h_\zeta \cdot | \psi |^2 \text{.}
\end{equation}

Next, recall that the $\zeta$-pseudoconvexity criterion implies that the hypotheses of Proposition \ref{thm.pseudoconvex_quant} hold.
Thus, using \eqref{eq.pseudoconvex_quant}, we can bound\footnote{Recall that in our $\mc{O}$-notation, the constants are allowed to depend on $n$ and $\tfr$.}
\begin{equation} \label{eq.set_pi} \begin{split}
\pi_\zeta^{\alpha \beta} \nabla_\alpha \psi^I \nabla_\beta \psi_I &\geq [ K f^{n - 2} \rho^2 + \mc{O} ( f^{n - 1} \rho^2 ) ] ( | \nabla_T \psi |^2 + | \nasla \psi |^2 ) \\
&\qquad - [ (n - 1) f^{n - 2} + \mc{O} ( f^n ) ] | \nabla_N \psi |^2 \text{.}
\end{split} \end{equation}
Moreover, an application of \eqref{eq.mixed_curv_ex} yields
\begin{equation} \label{eq.set_R}
- S^\alpha \nabla^\beta \psi_I \mc{R}_{\alpha \beta} \psi^I \geq \mc{O} ( f^{n - 2} \rho^3 ) \cdot ( | \nabla_T \psi |^2 + | \nasla \psi |^2 + | \psi |^2 ) \text{.}
\end{equation}
From \eqref{eq.f_grad_sq} and \eqref{eq.f_box}, we can see that
\begin{equation} \label{eq.S_divg}
\nabla^\alpha S_\alpha = - 2 f^{n - 2} + \mc{O} ( f^n ) \text{.}
\end{equation}
Thus, using \eqref{eq.f_deriv_trivial}, \eqref{eq.error_est}, \eqref{eq.w}, \eqref{eq.Sw}, and the assumption $\zeta = \mc{O} (1)$, we obtain
\begin{equation} \label{eq.h_est}
h_\zeta = \mc{O} ( f^n ) \text{,} \qquad \Box h_\zeta = \mc{O} ( f^n ) \text{.}
\end{equation}
Applying \eqref{eq.set_pi}-\eqref{eq.h_est} to \eqref{eq.set_2} and recalling \eqref{eq.f_deriv_trivial} yields
\begin{equation} \label{eq.set} \begin{split}
\Box \psi^I S_\zeta \psi_I &\geq \nabla^\beta P^Q_\beta + \mc{O} ( f^n ) \cdot | \psi |^2 - [ (n - 1) f^{n - 2} + \mc{O} ( f^n ) ] \cdot | \nabla_N \psi |^2 \\
&\qquad + [ K f^{n - 2} \rho^2 + \mc{O} ( f^{n - 1} \rho^2 ) ] \cdot ( | \nabla_T \psi |^2 + | \nasla \psi |^2 ) \text{.}
\end{split} \end{equation}

\subsubsection{The Conjugate Operator}

Expanding, we see that
\begin{equation} \label{eq.conj_1} \begin{split}
\mc{L} \psi &= e^{-F} \nabla^\alpha ( F' e^F \nabla_\alpha f \cdot \psi ) + e^{-F} \nabla^\alpha ( e^F \nabla_\alpha \psi ) + \sigma \psi \\
&= \Box \psi + 2 F' f^{-n + 3} \cdot S \psi + \mc{A}_0 \cdot \psi \text{,}
\end{split} \end{equation}
where we observe using \eqref{eq.f_grad_sq}, \eqref{eq.f_box}, \eqref{eq.F_deriv}, and \eqref{eq.h_est} that
\begin{equation} \label{eq.conj_A0} \begin{split}
\mc{A}_0 &:= [ ( F^\prime )^2 + F^{\prime\prime} ] \nabla^\alpha f \nabla_\alpha f + F^\prime \Box f + \sigma \\
&= ( \kappa^2 - n \kappa + \sigma ) + \lambda ( 2 \kappa - n + p ) f^p + \lambda^2 f^{2 p} + \lambda^2 \cdot \mc{O} ( f^2 ) \text{.}
\end{split} \end{equation}

Contracting \eqref{eq.conj_1} with $S_\zeta \psi$ and recalling \eqref{eq.h_est}, we have
\begin{equation} \label{eq.conj_2} \begin{split}
\mc{L} \psi^I S_\zeta \psi_I &= \Box \psi^I S_\zeta \psi_I + 2 F' f^{-n + 3} | \nabla_S \psi |^2 + \mc{A} \cdot \psi^I \nabla_S \psi_I + h_\zeta \mc{A}_0 | \psi |^2 \\
&= \Box \psi^I S_\zeta \psi_I + 2 F' f^{-n + 3} | \nabla_S \psi |^2 + \mc{A} \cdot \psi^I \nabla_S \psi_I \\
&\qquad + \lambda^2 \cdot \mc{O} ( f^n ) \cdot | \psi |^2 \text{.}
\end{split} \end{equation}
where by \eqref{eq.h_est},
\begin{equation} \label{eq.conj_A} \begin{split}
\mc{A} &:= \mc{A}_0 + 2 F' f^{-n + 3} h_\zeta \\
&= ( \kappa^2 - n \kappa + \sigma ) + \lambda ( 2 \kappa - n + p ) f^p + \lambda^2 f^{2 p} + \lambda^2 \cdot \mc{O} ( f^2 ) \text{.}
\end{split} \end{equation}

Next, letting
\begin{equation} \label{eq.P_flat}
P^S_\beta := \frac{1}{2} \mc{A} S_\beta \cdot | \psi |^2 \text{,}
\end{equation}
and applying the product rule yields
\begin{equation} \label{eq.A_ibp} \begin{split}
\mc{A} \cdot \psi^I \nabla_S \psi_I &= \nabla^\beta P^S_\beta - \frac{1}{2} ( S \mc{A} + \mc{A} \nabla^\alpha S_\alpha ) \cdot | \psi |^2 \text{.}
\end{split} \end{equation}
Applying \eqref{eq.f_grad_sq}, \eqref{eq.error_est}, and \eqref{eq.S_divg} to \eqref{eq.conj_A} yields
\begin{equation} \label{eq.GG} \begin{split}
- \frac{1}{2} ( S \mc{A} + \mc{A} \nabla^\alpha S_\alpha ) &= ( \kappa^2 - n \kappa + \sigma ) f^{n - 2} + \frac{2 - p}{2} \lambda ( 2 \kappa - n + p ) f^{n - 2 + p} \\
&\qquad + (1 - p) \lambda^2 f^{n - 2 + 2p} + \lambda^2 \cdot \mc{O} ( f^n ) \text{.}
\end{split} \end{equation}
Thus, combining \eqref{eq.F_deriv}, \eqref{eq.conj_2}, \eqref{eq.A_ibp}, and \eqref{eq.GG}, we obtain
\begin{equation} \label{eq.conj_3} \begin{split}
\mc{L} \psi^I S_\zeta \psi_I &= \nabla^\beta P^S_\beta + \Box \psi^I S_\zeta \psi_I + 2 ( \kappa f^{-n + 2} + \lambda f^{-n + 2 + p} ) | \nabla_S \psi |^2 \\
&\qquad + ( \kappa^2 - n \kappa + \sigma ) f^{n - 2} | \psi |^2 + \frac{2 - p}{2} \lambda ( 2 \kappa - n + p ) f^{n - 2 + p} | \psi |^2 \\
&\qquad + (1 - p) \lambda^2 f^{n - 2 + 2p} | \psi |^2 + \lambda^2 \cdot \mc{O} ( f^n ) \cdot | \psi |^2 \text{.}
\end{split} \end{equation}

Note that \eqref{eq.f_grad_sq} and \eqref{eq.S} imply that
\begin{equation} \label{eq.S_N} | \nabla_S \psi |^2 = [ f^{2 n - 4} + \mc{O} ( f^{2 n - 2} ) ] | \nabla_N \psi |^2 \text{.} \end{equation}
Applying \eqref{eq.set}, \eqref{eq.conj_3}, and \eqref{eq.S_N} results in the identity
\begin{equation} \label{eq.conj} \begin{split}
\mc{L} \psi^I S_\zeta \psi_I &\geq [ ( 2 \kappa - n + 1 ) f^{n - 2} + 2 \lambda f^{n - 2 + p} + \lambda \cdot \mc{O} (f^n) ] | \nabla_N \psi |^2 \\
&\qquad + [ K f^{n - 2} \rho^2 + \mc{O} ( f^{n - 1} \rho^2 ) ] ( | \nabla_T \psi |^2 + | \nasla \psi |^2 ) \\
&\qquad + \left[ ( \kappa^2 - n \kappa + \sigma ) f^{n - 2} + \frac{2 - p}{2} \lambda ( 2 \kappa - n + p ) f^{n - 2 + p} \right] | \psi |^2 \\
&\qquad + [ (1 - p) \lambda^2 f^{n - 2 + 2p} + \lambda^2 \cdot \mc{O} ( f^n ) ] | \psi |^2 + \nabla^\beta ( P^Q_\beta + P^S_\beta ) \text{.}
\end{split} \end{equation}

\subsubsection{A Hardy-Type Inequality}

We can generate extra positive in the $| \psi |^2$-terms in \eqref{eq.conj} by using the positivity already present in the $| \nabla_N \psi |^2$.
This is done via a Hardy-type inequality, for which we give the pointwise precursor below.

\begin{lemma} \label{thm.hardy_ptwise}
For any $q \in \R$, the following inequality holds:
\begin{equation} \label{eq.hardy_ptwise} \begin{split}
f^q | \nabla_N \psi |^2 &\geq \frac{1}{4} (q - n)^2 f^q \cdot | \psi |^2 + \mc{O} ( f^{q + 2} ) \cdot ( | \nabla_N \psi |^2 + | \psi |^2 ) \\
&\qquad - \frac{1}{2} (q - n) \nabla^\beta ( f^{q - 1} \nabla_\beta f \cdot | \psi |^2 ) \text{.}
\end{split} \end{equation}
\end{lemma}

\begin{proof}
Let $b \in \R$ be a constant to be fixed later, and observe that
\begin{align*}
0 &\leq f^{q - 2} | \nabla^\beta f \nabla_\beta \psi + b f \cdot \psi |^2 \\
&= f^{q - 2} | \nabla^\beta f \nabla_\beta \psi |^2 + b f^{q - 1} \cdot \nabla^\beta f \nabla_\beta ( \psi^2 ) + b^2 f^q | \psi |^2 \\
&= f^{q - 2} | \nabla^\beta f \nabla_\beta \psi |^2 + [ b^2 f^q - b (q - 1) f^{q - 2} \nabla^\beta f \nabla_\beta f - b f^{q - 1} \Box f ] | \psi |^2 \\
&\qquad + \nabla^\beta ( b f^{q - 1} \nabla_\beta f \cdot | \psi |^2 ) \text{.}
\end{align*}
Recalling \eqref{eq.f_grad_sq} and \eqref{eq.f_box}, the above becomes
\begin{equation} \label{eq.hardy_ptwise_1} \begin{split}
[ f^q + \mc{O} ( f^{q + 2} ) ] | \nabla_N \psi |^2 &\geq [ - b ( b + n - q ) f^q + O ( f^{q + 2} ) ] | \psi |^2 \\
&\qquad - \nabla^\beta ( b f^{q - 1} \nabla_\beta f \cdot | \psi |^2 ) \text{.}
\end{split} \end{equation}
Finally, we observe that the constant $- b ( b + n - q )$ in \eqref{eq.hardy_ptwise_1} is maximized when $b = \frac{1}{2} (q - n)$.
Taking this choice of $b$ results in \eqref{eq.hardy_ptwise}.
\end{proof}

In particular, taking $q$ to be $n - 2$ and $n - 2 + p$ in \eqref{eq.hardy_ptwise}, we obtain
\begin{equation} \label{eq.hardy_1} \begin{split}
f^{n - 2} | \nabla_N \psi |^2 &\geq f^{n - 2} | \psi |^2 + \nabla^\beta ( f^{n - 3} \nabla_\beta f \cdot | \psi |^2 ) \\
&\qquad + \mc{O} ( f^n ) \cdot ( | \nabla_N \psi |^2 + | \psi |^2 ) \text{,} \\
f^{n - 2 + p} | \nabla_N \psi |^2 &\geq \frac{ (2 - p)^2 }{ 4 } f^{n - 2} | \psi |^2 + \frac{2 - p}{2} \nabla^\beta ( f^{n - 3} \nabla_\beta f \cdot | \psi |^2 ) \\
&\qquad + \mc{O} ( f^n ) \cdot ( | \nabla_N \psi |^2 + | \psi |^2 ) \text{.}
\end{split} \end{equation}
Letting
\begin{equation} \label{eq.P_natural}
P^H_\beta := ( 2 \kappa - n + 1 ) f^{n - 3} \nabla_\beta f \cdot | \psi |^2 + \frac{2 - p}{2} \lambda f^{n - 3 + p} \nabla_\beta f \cdot | \psi |^2 \text{,}
\end{equation}
then applying \eqref{eq.hardy_1} to \eqref{eq.conj} (note $2 \kappa - n + 1 \geq 0$ by \eqref{eq.p_kappa}) yields
\begin{equation} \label{eq.hardy_2} \begin{split}
\mc{L} \psi^I S_\zeta \psi_I &\geq \nabla^\beta ( P^Q_\beta + P^S_\beta + P^H_\beta ) + [ \lambda f^{n - 2 + p} + \lambda \cdot \mc{O} (f^n) ] | \nabla_N \psi |^2 \\
&\qquad + [ K f^{n - 2} \rho^2 + \mc{O} ( f^{n - 1} \rho^2 ) ] ( | \nabla_T \psi |^2 + | \nasla \psi |^2 ) \\
&\qquad + [ \kappa^2 - ( n - 2 ) \kappa + \sigma - (n - 1) ] f^{n - 2} | \psi |^2 \\
&\qquad + \frac{2 - p}{2} \, \lambda \left( 2 \kappa - n + 1 + \frac{p}{2} \right) f^{n - 2 + p} | \psi |^2 \\
&\qquad + [ (1 - p) \lambda^2 f^{n - 2 + 2p} + \lambda^2 \cdot \mc{O} ( f^n ) ] | \psi |^2 \text{,}
\end{split} \end{equation}

Next, noting from \eqref{eq.f_grad_sq}, \eqref{eq.S}, and \eqref{eq.h_est} that
\begin{align*}
\mc{L} \psi^I S_\zeta \psi_I &= \mc{L} \psi^I S \psi_I + \mc{L} \psi^I \cdot h_\zeta \psi_I \\
&\leq \lambda^{-1} f^{n - 2 - p} | \mc{L} \psi |^2 + \frac{1}{2} \lambda [ f^{n - 2 + p} + \mc{O} ( f^n ) ] | \nabla_N \psi |^2 + \lambda \cdot \mc{O} ( f^n ) \cdot | \psi |^2 \text{,}
\end{align*}
then combining the above with \eqref{eq.hardy_2} yields
\begin{equation} \label{eq.hardy_3} \begin{split}
\lambda^{-1} f^{n - 2 - p} | \mc{L} \psi |^2 &\geq \nabla^\beta ( P^Q_\beta + P^S_\beta + P^H_\beta ) \\
&\qquad + \left[ \frac{1}{2} \lambda f^{n - 2 + p} + \lambda \cdot \mc{O} (f^n) \right] | \nabla_N \psi |^2 \\
&\qquad + [ K f^{n - 2} \rho^2 + \mc{O} ( f^{n - 1} \rho^2 ) ] ( | \nabla_T \psi |^2 + | \nasla \psi |^2 ) \\
&\qquad + [ \kappa^2 - ( n - 2 ) \kappa + \sigma - (n - 1) ] f^{n - 2} | \psi |^2 \\
&\qquad + \frac{2 - p}{2} \, \lambda \left( 2 \kappa - n + 1 + \frac{p}{2} \right) f^{n - 2 + p} | \psi |^2 \\
&\qquad + [ (1 - p) \lambda^2 f^{n - 2 + 2p} + \lambda^2 \cdot \mc{O} ( f^n ) ] | \psi |^2 \text{.}
\end{split} \end{equation}
Thus, by recalling that $f \ll_{n, \tfr, p, K} 1$ and by recalling \eqref{eq.p_kappa}, we can find a constant $C > 0$, depending on $n$, $\tfr$, $p$, and $K$, such that
\begin{equation} \label{eq.hardy} \begin{split}
\lambda^{-1} f^{n - 2 - p} | \mc{L} \psi |^2 &\geq \nabla^\beta ( P^Q_\beta + P^S_\beta + P^H_\beta ) + C \lambda f^{n - 2 + p} | \nabla_N \psi |^2 \\
&\qquad + C f^{n - 2} \rho^2 ( | \nabla_T \psi |^2 + | \nasla \psi |^2 ) \\
&\qquad + [ \kappa^2 - ( n - 2 ) \kappa + \sigma - (n - 1) ] f^{n - 2} | \psi |^2 \\
&\qquad + C ( \lambda f^{n - 2 + p} + \lambda^2 f^{n - 2 + 2p} ) | \psi |^2 \text{.}
\end{split} \end{equation}

\subsubsection{Boundary Expansions}

By setting
\begin{equation} \label{eq.P}
P := P^Q + P^S + P^H \text{,}
\end{equation}
then \eqref{eq.hardy} is identical to \eqref{eq.psi_est}.
Thus, the proof of Lemma \ref{thm.psi_est} would be complete provided we verify $P$, as defined in \eqref{eq.P}, satisfies the estimate \eqref{eq.psi_est_boundary}.

Recall first that $\mc{N}$ has the asymptotic expansion \eqref{eq.rho_normal}.
Applying \eqref{eq.rho_normal}, \eqref{eq.f_deriv}, and \eqref{eq.f_deriv_trivial} to \eqref{eq.conj_A} and \eqref{eq.P_flat}, we see that
\begin{equation} \label{eq.boundary_11} \begin{split}
P^S ( \mc{N} ) &= \frac{1}{2} ( \kappa^2 - n \kappa + \sigma ) f^{n - 2} \cdot | \psi |^2 + \frac{1}{2} \lambda ( 2 \kappa - n + p ) f^{n - 2 + p} \cdot | \psi |^2 \\
&\qquad + \frac{1}{2} \lambda^2 f^{n - 2 + 2p} \cdot | \psi |^2 + \lambda^2 \cdot \mc{O} ( f^n ) \cdot | \psi |^2 \text{.}
\end{split} \end{equation}
Similarly, another application of \eqref{eq.rho_normal} and \eqref{eq.f_deriv} to \eqref{eq.P_natural} yields that
\begin{equation} \label{eq.boundary_12} \begin{split}
P^H ( \mc{N} ) &= ( 2 \kappa - n + 1 ) f^{n - 2} \cdot | \psi |^2 + \frac{2 - p}{2} \lambda f^{n - 2 + p} \cdot | \psi |^2 \\
&\qquad + \lambda \cdot \mc{O} ( f^n ) \cdot | \psi |^2 \text{.}
\end{split} \end{equation}
Combining \eqref{eq.boundary_11} and \eqref{eq.boundary_12}, we see that as long as $f$ is sufficiently small, there is some constant $\mc{C}$---depending on $n$, $\tfr$, $p$---such that
\begin{equation} \label{eq.boundary_1}
P^S ( \mc{N} ) + P^H ( \mc{N} ) \leq \mc{C} ( \lambda^2 + | \sigma | ) f^{n - 2} \cdot | \psi |^2 \text{.}
\end{equation}

For $P^Q$, we expand using \eqref{eq.P_sharp}:
\begin{equation} \label{eq.boundary_20} \begin{split}
P^Q ( \mc{N} ) &= \nabla_S \psi^I \nabla_{ \mc{N} } \psi_I - \frac{1}{2} g (S, \mc{N}) \cdot \nabla^\mu \psi^I \nabla_\mu \psi_I \\
&\qquad + h_\zeta \cdot \psi^I \nabla_{ \mc{N} } \psi_I - \frac{1}{2} \mc{N} h_\zeta \cdot | \psi |^2 \\
&:= B_1 + B_2 + B_3 + B_4 \text{.}
\end{split} \end{equation}
Using \eqref{eq.rho_normal}, \eqref{eq.f_deriv}, and \eqref{eq.h_est}, we see that
\begin{equation} \label{eq.boundary_21} \begin{split}
B_3 &\leq \mc{O} ( f^n ) \cdot ( \rho^2 | \nabla_\rho \psi |^2 + \rho^8 | \nabla_t \psi |^2 + \rho^8 | \nasla \psi |^2 + | \psi |^2 ) \text{,} \\
B_4 &= \mc{O} ( f^n ) \cdot | \psi |^2 \text{.}
\end{split} \end{equation}
For the remaining, we note from \eqref{eq.aads}, \eqref{eq.f_deriv_trivial}, \eqref{eq.f_grad}, and \eqref{eq.S} that
\begin{align*}
\nabla_S \psi &= f^{n - 2} \rho \left\{ [ 1 + \mc{O} (\rho^2) ] \cdot \nabla_\rho \psi + \mc{O} (f) \cdot \nabla_t \psi + \sum_{ A = 1 }^{n - 1} \mc{O} ( f \rho^2 ) \cdot \nabla_A \psi \right\} \text{,} \\
\nabla^\mu \psi^I \nabla_\mu \psi_I &= [ \rho^2 + O ( \rho^4 ) ] ( - | \nabla_t \psi |^2 + | \nabla_\rho \psi |^2 ) + | \nasla \psi |^2 \text{.}
\end{align*}
As a result, we see using the above, \eqref{eq.rho_normal}, and \eqref{eq.f_deriv_trivial} that
\begin{equation} \label{eq.boundary_22} \begin{split}
B_1 &\leq [ f^{n - 2} \rho^2 + \mc{O} ( f^{n - 2} \rho^4 ) ] \cdot | \nabla_\rho \psi |^2 + \mc{O} ( f^{n - 1} \rho^2 ) \cdot | \nabla_\rho \psi | | \nabla_t \psi | \\
&\qquad + \mc{O} ( f^{n - 1} \rho^3 ) \cdot | \nabla_\rho \psi | | \nasla \psi | + \mc{O} ( f^{n - 1} \rho^5 ) \cdot | \nabla_t \psi |^2 \\
&\qquad + \mc{O} ( f^{n - 1} \rho^4 ) \cdot | \nabla_t \psi | | \nasla \psi | + \mc{O} ( f^{n - 1} \rho^5  ) \cdot | \nasla \psi |^2 \text{,} \\
B_2 &= \frac{1}{2} f^{n - 2} [ \rho^2 + O (\rho^4) ] ( | \nabla_t \psi |^2 - | \nabla_\rho \psi |^2 ) - \frac{1}{2} f^{n - 2} | \nasla \psi |^2 \text{.}
\end{split} \end{equation}

From \eqref{eq.boundary_20}-\eqref{eq.boundary_22} (note the $| \nasla \psi |^2$-term is negative for $f < f_0$), we conclude
\begin{equation} \label{eq.boundary_2} \begin{split}
P^Q ( \mc{N} ) &\leq \left[ \frac{1}{2} f^{n - 2} \rho^2 + \mc{O} ( f^{n - 1} \rho^2 ) \right] ( | \nabla_t \psi |^2 + | \nabla_\rho \psi |^2 ) + \mc{O} ( f^n ) \cdot | \psi |^2 \\
&\leq \mc{C} f^{n - 2} \rho^2 ( | \nabla_t \psi |^2 + | \nabla_\rho \psi |^2 ) + \mc{C} f^{n - 2} \cdot | \psi |^2 \text{,}
\end{split} \end{equation}
for some $\mc{C} > 0$.
Finally, combining \eqref{eq.boundary_1} and \eqref{eq.boundary_2} completes the derivation of \eqref{eq.psi_est_boundary}, which completes the proof of Lemma \ref{thm.psi_est}.

\subsection{Completion of the Proof of \eqref{eq.carleman_lambda}} \label{sec.carleman_comp}

The next step is to convert the estimates for $\psi$ in Lemma \ref{thm.psi_est} into corresponding estimates for the original function $\phi$.
More specifically, we prove the following pointwise bound for $\phi$.

\begin{lemma} \label{thm.phi_est}
There exists $C > 0$, depending on $n$, $\tfr$, $p$, $K$, such that when $f < f_0$,
\begin{equation} \label{eq.phi_est} \begin{split}
\lambda^{-1} f^{-p} E^p_{\kappa, \lambda} | ( \Box + \sigma ) \phi |^2 &\geq C E^p_{\kappa, \lambda} ( \rho^4 | \nabla_t \phi |^2 + \rho^4 | \nabla_\rho \phi |^2 + \rho^2 | \nasla \phi |^2 ) \\
&\qquad + E^p_{\kappa, \lambda} [ \kappa ^2 - ( n - 2 ) \kappa + \sigma - (n - 1) ] | \phi |^2 \\
&\qquad + C \lambda^2 E^p_{\kappa, \lambda} f^{2p} | \phi |^2 + \nabla^\beta P_\beta \text{,}
\end{split} \end{equation}
where
\begin{equation} \label{eq.carleman_exp}
E^p_{\kappa, \lambda} := e^{-2 F} f^{n - 2} = f^{n - 2 - 2 \kappa} e^\frac{ - 2 \lambda f^p }{p} \text{,}
\end{equation}
and where the $1$-form $P$ satisfies, for some $\mc{C} > 0$ (depending on $n$, $\tfr$, $p$),
\begin{equation} \label{eq.phi_est_boundary}
\rho^{-n} \cdot P ( \mc{N} ) \leq \mc{C} [ | \nabla_t ( \rho^{-\kappa} \phi ) |^2 + | \nabla_\rho ( \rho^{-\kappa} \phi ) |^2 ] + \mc{C} ( \lambda^2 + | \sigma | ) | \rho^{-\kappa - 1} \phi |^2 \text{.}
\end{equation}
\end{lemma}

\subsubsection{Proof of Lemma \ref{thm.phi_est}}

We begin with the bulk estimate \eqref{eq.psi_est}.
First, we use the largeness of $\lambda$ and the smallness of $f$ in order to obtain
\begin{equation} \label{eq.psi_est_0} \begin{split}
\lambda^{-1} f^{n - 2 - p} | \mc{L} \psi |^2 &\geq C f^{n - 2 + 2p} | \nabla_N \psi |^2 + C f^{n - 2} \rho^2 ( | \nabla_T \psi |^2 + | \nasla \psi |^2 ) \\
&\qquad + [ \kappa^2 - ( n - 2 ) \kappa + \sigma - (n - 1) ] f^{n - 2} | \psi |^2 \\
&\qquad + C \lambda^2 f^{n - 2 + 2p} | \psi |^2 + \nabla^\beta P_\beta \text{.}
\end{split} \end{equation}
In particular, we shrank the $| \nabla_N \psi |^2$-term and dropped one of the $| \psi |^2$-terms.

We now write \eqref{eq.psi_est_0} in terms of $\phi$.
By \eqref{eq.f_grad_sq}, \eqref{eq.F_deriv}, and the assumption $\lambda \geq 1 + \kappa$,
\begin{equation} \label{eq.psi_phi_deriv} e^{-2 F} | \nabla_N \phi |^2 = | \nabla_N \psi + F' \nabla_N f \cdot \psi |^2 \lesssim | \nabla_N \psi |^2 + \lambda^2 | \psi |^2 \text{.} \end{equation}
As a result, applying \eqref{eq.carleman_exp} and \eqref{eq.psi_phi_deriv} to \eqref{eq.psi_est_0} yields, for $f < f_0$,
\begin{equation} \label{eq.lower_1} \begin{split}
\lambda^{-1} f^{-p} E^p_{\kappa, \lambda} | ( \Box + \sigma ) \phi |^2 &\geq C E^p_{\kappa, \lambda} ( \rho^2 | \nabla_T \phi |^2 + \rho^2 | \nasla \phi |^2 + f^{2p} | \nabla_N \phi |^2 ) \\
&\qquad + E^p_{\kappa, \lambda} [ \kappa^2 - ( n - 2 ) \kappa + \sigma - (n - 1) ] | \phi |^2 \\
&\qquad + C \lambda^2 E^p_{\kappa, \lambda} f^{2p} | \phi |^2 + \nabla^\beta P_\beta \text{,}
\end{split} \end{equation}
where $C$ is some (possibly different) constant depending on $n$, $\tfr$, $p$, and $K$.

Next, applying \eqref{eq.vf_TN}, we observe that
\begin{equation} \label{eq.TN_trho} \begin{split}
\rho^4 | \nabla_t \phi |^2 &\lesssim \rho^2 | \nabla_T \phi |^2 + f^2 \rho^2 | \nabla_N \phi |^2 + \rho^6 | \nasla \phi |^2 \text{,} \\
\rho^4 | \nabla_\rho \phi |^2 &\lesssim \rho^2 | \nabla_N \phi |^2 + f^2 \rho^2 | \nabla_T \phi |^2 + \rho^6 | \nasla \phi |^2 \text{.}
\end{split} \end{equation}
Thus, recalling \eqref{eq.lower_1} and \eqref{eq.TN_trho}---and noting that $f \geq \rho$ by \eqref{eq.f_trivial}---results in \eqref{eq.phi_est}.\footnote{Again, the constant $C$ may change, but it depends still on $n$, $\tfr$, $p$, and $K$.}
(Here, we take $P$ to be the same $1$-form as was in Lemma \ref{thm.psi_est}.)

To complete the proof of Lemma \ref{thm.phi_est}, it remains to express the boundary estimate \eqref{eq.psi_est_boundary} in terms of $\phi$.
For convenience, we define the shorthand
\begin{equation} \label{eq.carly_exp}
\mc{E}_{p, \lambda} := e^{- \frac{ \lambda f^p }{p} } \leq 1 \text{.}
\end{equation}
Note that by \eqref{eq.f_deriv},
\begin{equation} \label{eq.carly_exp_deriv}
| \partial_\rho \mc{E}_{p, \lambda} | \lesssim \lambda f^p \rho^{-1} \text{,} \qquad | \partial_t \mc{E}_{p, \lambda} | \lesssim \lambda f^{p + 1} \rho^{-1} \text{.}
\end{equation}

Since $\rho \leq f$ by \eqref{eq.f_trivial}, then \eqref{eq.p_kappa} and \eqref{eq.carly_exp} imply
\begin{equation} \label{eq.upper_30}
f^{n - 2} \rho^{-n} \cdot | \psi |^2 = f^{n - 2 - 2 \kappa} \rho^{-n + 2 + 2 \kappa} \cdot | \mc{E}_{p, \lambda} \rho^{ - \kappa - 1 } \phi |^2 \leq | \rho^{ - \kappa - 1 } \phi |^2 \text{.}
\end{equation}
Next, since $f^{-1} \rho = \sin (\tfr t)$ depends only on $t$, then \eqref{eq.f_deriv}, \eqref{eq.carly_exp}, and \eqref{eq.carly_exp_deriv} yield
\begin{align*}
| \nabla_\rho \psi | &\lesssim f^{-\kappa} \rho^\kappa [ \mc{E}_{p, \lambda} | \nabla_\rho ( \rho^{-\kappa} \phi ) | + | \partial_\rho \mc{E}_{p, \lambda} | | \rho^{-\kappa} \phi | ] \\
&\lesssim f^{-\kappa} \rho^\kappa [ | \nabla_\rho ( \rho^{-\kappa} \phi ) | + \lambda f^p | \rho^{-\kappa - 1} \phi | ] \text{.}
\end{align*}
Again, \eqref{eq.f_trivial} and \eqref{eq.p_kappa} yields, for $f < f_0$,
\begin{equation} \label{eq.upper_31} \begin{split}
f^{n - 2} \rho^{-n + 2} | \nabla_\rho \psi |^2 &\lesssim f^{-2 \kappa + n - 2} \rho^{2 \kappa - n + 2} [ | \nabla_\rho ( \rho^{-\kappa} \phi ) |^2 + \lambda^2 | \rho^{-\kappa - 1} \phi |^2 ] \\
&\lesssim | \nabla_\rho ( \rho^{-\kappa} \phi ) |^2 + \lambda^2 | \rho^{-\kappa - 1} \phi |^2 \text{.}
\end{split} \end{equation}

Moreover, a similar computation using \eqref{eq.f_deriv} and \eqref{eq.carly_exp_deriv} yields
\begin{align*}
| \nabla_t \psi | &\lesssim f^{-\kappa} \rho^\kappa | \nabla_t ( \rho^{-\kappa} \phi ) | + | \nabla_t ( f^{-\kappa} \rho^\kappa ) | | \rho^{-\kappa} \phi | + f^{-\kappa} \rho^\kappa | \partial_t \mc{E}_{p, \lambda} | | \rho^{-\kappa} \phi | \\
&\lesssim f^{-\kappa} \rho^\kappa [ | \nabla_t ( \rho^{-\kappa} \phi ) | + \lambda f | \rho^{-\kappa - 1} \phi | ] \text{,}
\end{align*}
from which, along with \eqref{eq.p_kappa}, we obtain,
\begin{equation} \label{eq.upper_32} \begin{split}
f^{n - 2} \rho^{-n + 2} | \nabla_t \psi |^2 &\lesssim f^{-2 \kappa + n - 1} \rho^{2 \kappa - n + 1} [ | \nabla_t ( \rho^{-\kappa} \phi ) |^2 + \lambda^2 | \rho^{-\kappa - 1} \phi |^2 ] \\
&\lesssim | \nabla_\rho ( \rho^{-\kappa} \phi ) |^2 + \lambda^2 | \rho^{-\kappa - 1} \phi |^2 \text{.}
\end{split} \end{equation}
Combining \eqref{eq.psi_est_boundary} with \eqref{eq.upper_30}-\eqref{eq.upper_32} yields \eqref{eq.phi_est_boundary}.

\subsubsection{The Integrated Estimate}

To complete the proof of \eqref{eq.carleman_lambda}, we integrate the pointwise inequality \eqref{eq.phi_est} over $\Omega_{ f_0, \rho_0 }$, and we then apply the divergence theorem.
Observe $\partial \Omega_{ \rho_0 }$ contains two components: one from $\{ f = f_0 \}$, and the other from $\{ \rho = \rho_0 \}$.
Since both $\phi$ and $\nabla \phi$ vanish on $\{ f = f_0 \}$, then by \eqref{eq.phi_est_boundary}, we need only consider $\{ \rho = \rho_0 \}$.
The result of this computation is
\begin{equation} \label{eq.integral_1} \begin{split}
&\lambda^{-1} \int_{ \Omega_{ f_0, \rho_0 } } f^{-p} E^p_{\kappa, \lambda} | ( \Box + \sigma ) \phi |^2 + \int_{ \{ \rho = \rho_0 \} } P ( \mc{N} ) \\
&\quad \geq C \int_{ \Omega_{ f_0, \rho_0 } } E^p_{\kappa, \lambda} ( \rho^4 | \nabla_t \phi |^2 + \rho^4 | \nabla_\rho \phi |^2 + \rho^2 | \nasla \phi |^2 ) \\
&\quad \qquad + \int_{ \Omega_{ f_0, \rho_0 } } E^p_{\kappa, \lambda} [ \kappa^2 - ( n - 2 ) \kappa + \sigma - (n - 1) ] | \phi |^2 \\
&\quad \qquad + C \lambda^2 \int_{ \Omega_{ f_0, \rho_0 } } E^p_{\kappa, \lambda} f^{2p} | \phi |^2 \text{,}
\end{split} \end{equation}

It remains to estimate the boundary term in \eqref{eq.integral_1}.
Let $| \mathring{g} |$ and $| \mf{g} |$ denote the square roots of the determinants of the matrices $[ \mathring{g}_{a b} ]$ and $[ \mf{g}_{ab} ]$, respectively, defined with respect to some fixed coordinates $x^a$.
Then,
\begin{equation} \label{eq.boundary_volume} \begin{split}
\int_{ \{ \rho = \rho_0 \} } P ( \mc{N} ) &\leq \int_{ \{ \rho = \rho_0 \} } \rho^{-n} P ( \mc{N} ) \cdot | \mf{g} | d t d x^1 \dots d x^{n-1} \\
&= \int_{ \{ \rho = \rho_0 \} } \rho^{-n} P ( \mc{N} ) [ | \mathring{g} | + \mc{O} ( \rho^2 ) ] d t d x^1 \dots d x^{n-1} \text{,}
\end{split} \end{equation}
Recalling \eqref{eq.aads_bdd}, the estimate \eqref{eq.phi_est_boundary}, and the fact that the coordinate systems we deal with are bounded, we see, for sufficiently small $\rho_0$, that
\begin{equation} \label{eq.integral_2} \begin{split}
\int_{ \{ \rho = \rho_0 \} } P ( \mc{N} ) &\leq \mc{C} \int_{ \{ \rho = \rho_0 \} } [ | \nabla_t ( \rho^{-\kappa} \phi ) |^2 + | \nabla_\rho ( \rho^{-\kappa} \phi ) |^2 ] \\
&\qquad + \mc{C} ( \lambda^2 + | \sigma | ) \int_{ \{ \rho = \rho_0 \} } | \rho^{-\kappa - 1} \phi |^2 \text{.}
\end{split} \end{equation}
Finally, combining \eqref{eq.integral_1} and \eqref{eq.integral_2} results in \eqref{eq.carleman_lambda}.

\subsection{Proof of \eqref{eq.carleman_beta}} \label{sec.carleman_beta}

We briefly sketch the proof of \eqref{eq.carleman_beta}, which is mostly analogous to that of \eqref{eq.carleman_lambda}.
The setting is as before, except we take $\sigma = \lambda = 0$.
With this $\sigma$ and $\lambda$, we retrace the proof of Lemma \ref{thm.psi_est} up to \eqref{eq.conj}, from which we obtain\footnote{Here, we define $F := \kappa \cdot \log f$ and $\psi := e^{-F} \phi = f^{-\kappa} \phi$, while $S_\zeta$ and $h_\zeta$ are as before.}
\begin{equation} \label{eq.proof_beta_1} \begin{split}
\mc{L} \psi^I S_\zeta \psi_I &\geq \nabla^\beta P'_\beta + [ ( 2 \kappa - n + 1 ) f^{n - 2} + \kappa \cdot \mc{O} (f^n) ] | \nabla_N \psi |^2 \\
&\qquad + [ K f^{n - 2} \rho^2 + \mc{O} ( f^{n - 1} \rho^2 ) ] ( | \nabla_T \psi |^2 + | \nasla \psi |^2 ) \\
&\qquad + \left[ ( \kappa^2 - n \kappa ) f^{n - 2} + \kappa^2 \cdot \mc{O} ( f^n ) \right] | \psi |^2 \text{,}
\end{split} \end{equation}
where $P'$ now satisfies
\begin{equation} \label{eq.proof_beta_1P} \begin{split}
P' ( \mc{N} ) &\leq \mc{C} f^{n - 2} \rho^2 ( | \nabla_t \psi |^2 + | \nabla_\rho \psi |^2 ) + \mc{C} \kappa^2 f^{n - 2} | \psi |^2 \text{,}
\end{split} \end{equation}
for some constant $\mc{C} > 0$ depending on $n$ and $\tfr$.\footnote{Essentially, the only difference from the corresponding proof of \eqref{eq.carleman_lambda} is that $\sigma = \lambda = 0$, and that any factors of $\lambda$ in the error terms can be replaced by $\kappa$.}

From here, the proof proceeds slightly differently.
Applying once again \eqref{eq.f_grad_sq}, \eqref{eq.S}, and \eqref{eq.h_est}, we estimate the left-hand side of \eqref{eq.proof_beta_1}:
\begin{align*}
\mc{L} \psi^I S_\zeta \psi_I &= \mc{L} \psi^I S \psi_I + \mc{L} \psi^I \cdot h_\zeta \psi_I \\
&\leq \kappa^{-1} f^{n - 2} | \mc{L} \psi |^2 + \frac{1}{2} \kappa [ f^{n - 2} + \mc{O} ( f^n ) ] | \nabla_N \psi |^2 + \kappa \cdot \mc{O} ( f^n ) \cdot | \psi |^2 \text{.}
\end{align*}
Combining the above with \eqref{eq.proof_beta_1} and using that $\kappa \gg n$ yields
\begin{equation} \label{eq.proof_beta_2} \begin{split}
\kappa^{-1} f^{n - 2} | \mc{L} \psi |^2 &\geq \nabla^\beta P'_\beta + [ \kappa f^{n - 2} + \kappa \cdot \mc{O} (f^n) ] | \nabla_N \psi |^2 \\
&\qquad + \left[ K f^{n - 2} \rho^2 + \mc{O} ( f^{n - 1} \rho^2 ) \right] ( | \nabla_T \psi |^2 + | \nasla \psi |^2 ) \\
&\qquad + \left[ \frac{1}{2} \kappa^2 f^{n - 2} + \kappa^2 \cdot \mc{O} ( f^n ) \right] | \psi |^2 \text{.}
\end{split} \end{equation}

Using analogues of \eqref{eq.psi_phi_deriv} and \eqref{eq.TN_trho} as before, we obtain from \eqref{eq.proof_beta_2},
\begin{equation} \label{eq.proof_beta_3} \begin{split}
\kappa^{-1} f^{n - 2 - 2 \kappa} | \Box \phi |^2 &\geq C f^{n - 2 - 2 \kappa} ( \rho^4 | \nabla_t \phi |^2 + \rho^4 | \nabla_\rho \phi |^2 + \rho^2 | \nasla \phi |^2 ) \\
&\qquad + C f^{n - 2 - 2 \kappa} \kappa^2 | \phi |^2 + \nabla^\beta P'_\beta \text{,}
\end{split} \end{equation}
where the constant $C > 0$ depends on $n$, $\tfr$, and $K$.
Furthermore, from analogues of \eqref{eq.upper_30}, \eqref{eq.upper_31}, and \eqref{eq.upper_32}, we can bound
\begin{equation} \label{eq.proof_beta_4}
\rho^{-n} \cdot P ( \mc{N} ) \leq \mc{C} \kappa^2 [ | \nabla_t ( \rho^{-\kappa} \phi ) |^2 + | \nabla_\rho ( \rho^{-\kappa} \phi ) |^2 + | \rho^{-\kappa - 1} \phi |^2 ] \text{,} \end{equation}
where the constant $\mc{C} > 0$ again depends on $n$ and $\tfr$.

We now integrate \eqref{eq.proof_beta_3} over $\Omega_{ f_0, \rho_0 }$ and apply the divergence theorem.
Moreover, from \eqref{eq.proof_beta_4}, we obtain the analogue of \eqref{eq.boundary_volume}, from which we obtain
\begin{equation} \label{eq.proof_beta_5}
\int_{ \{ \rho = \rho_0 \} } P' ( \mc{N} ) \leq \mc{C} \kappa^2 \int_{ \{ \rho = \rho_0 \} } [ | \nabla_t ( \rho^{-\kappa} \phi ) |^2 + | \nabla_\rho ( \rho^{-\kappa} \phi ) |^2 + | \rho^{-\kappa - 1} \phi |^2 ] \text{.}
\end{equation}
This handles the boundary term and results in the desired estimate \eqref{eq.carleman_beta}.

\section{Proof of Uniqueness} \label{sec.proof}

In this section, we provide the argument which proves unique continuation using the Carleman estimates of Theorem \ref{thm.carleman}.
In Section \ref{sec.proof_gen}, we state and prove the most general version (Theorem \ref{theo:ads}) of our uniqueness results.
In Section \ref{sec.proof_inf}, we briefly discuss a variant of our result (Theorem \ref{theo:ads5}) which deals with infinite-order vanishing and generally bounded potentials.

\subsection{The General Result} \label{sec.proof_gen}

Before stating our most general result, we first recall the analogues of the constants $\beta_\pm$ discussed in Section \ref{sec:vanishing} in $(n + 1)$-dimensions:

\begin{definition}
In $(n+1)$-dimensions, we define the constants $\beta_\pm$ by\footnote{These can be derived in an analogous manner as in Section \ref{sec:vanishing} for $n = 3$.}
\begin{equation} \label{eq.beta_pm}
\beta_\pm := \frac{n}{2} \pm \sqrt{ \frac{n^2}{4} - \sigma } \text{.}
\end{equation}
\end{definition}

Our general unique continuation can now be stated as follows:

\begin{theorem} \label{theo:ads}
Let $n \in \N$ and $y, r_0 > 0$.
Let $( \mc{I}_\tfr, \mathring{g} )$ be an $n$-dimensional segment of bounded static AdS infinity (see Definition \ref{def.aads_infinity}), and let $\mc{M} := (r_0, \infty) \times \mc{I}_\tfr$.
Suppose in addition that:
\begin{itemize}
\item $( \mc{M}, g )$ is an admissible aAdS segment, as described in Definition \ref{def.aads}.

\item The $\zeta$-pseudoconvex property (see Definition \ref{def.pseudoconvex_ass}) is satisfied for some function $\zeta = \mc{O} (1)$ on $\mc{M}$, with constant $K > 0$.
\end{itemize}
Moreover, suppose $\phi \in \Gamma \underline{T}^0_l \mathcal{M}$ satisfies the following:
\begin{enumerate}[(i)]
\item There exist constants $0 < p < 1$, $\sigma \in \R$, and $C > 0$ such that
\begin{equation} \label{we}
| \Box_g \phi + \sigma \phi |^2 \leq C \rho^p ( \rho^4 | \nabla_t \phi |^2 + \rho^4 | \nabla_\rho \phi |^2 + \rho^2 | \nasla \phi |^2 + \rho^{2p} | \phi |^2 ) \text{.}
\end{equation}

\item The following vanishing condition holds:
\begin{equation}  \label{bca}
\lim_{ \rho_0 \searrow 0 } \int_{ \mc{M} \cap \{ \rho = \rho_0 \} } [ | \nabla_t ( \rho^{ - \kappa } \phi ) |^2 + | \nabla_\rho ( \rho^{ - \kappa } \phi ) |^2  + | \rho^{- \kappa - 1} \phi |^2 ] d \mathring{g} = 0 \text{,}
\end{equation}
where $\kappa$ is given by
\begin{equation}
\label{bckappa} \kappa = \begin{cases} (n - 1) - \beta_- = \beta_+ - 1 & \text{if } \sigma \leq \frac{n^2 - 1}{4} \text{,} \\ \frac{n-1}{2} & \text{if } \frac{n^2 - 1}{4} < \sigma \text{.} \end{cases}
\end{equation}

\item The following finiteness condition holds:
\begin{equation} \label{ficoa}
\int_{ \mc{M} } \rho^{2+p} | \nasla \phi |^2 < \infty \text{.} 
\end{equation}
\end{enumerate}
Then, there exists $0 < f_0 \ll_{n, y, p, K} 1$ such that $\phi \equiv 0$ in $\mathcal{M} \cap \{ f < f_0 / 2 \}$.
\end{theorem}

\begin{remark}
The choice of $\kappa$ in \eqref{bckappa} in particular ensures both $\kappa \geq \frac{n-1}{2}$ and that the quantity $\kappa^2 - (n - 2) \kappa + \sigma - (n - 1)$ on the right hand side of \eqref{eq.carleman_lambda} is non-negative.
\end{remark}

\begin{proof}
Let $f_0 \ll_{n, y, p, K} 1$ be as in the statement of Theorem \ref{thm.carleman}, where $n$, $y$, $p$, $K$ are as given in the assumptions of the present theorem.
Let $\chi : \left[0,f_0\right] \rightarrow [0, 1]$ be a smooth cut-off function satisfying
\[ \chi ( s ) = \begin{cases} 1 & 0 \leq s \leq \frac{f_0}{2} \text{,} \\ 0 & s > \frac{3 f_0}{4} \text{.} \end{cases} \]
We wish to derive a wave equation for $\chi (f) \cdot \phi$; for notational convenience, we write $\chi$ for $\chi (f) = \chi \circ f$, and we use $\prime$ to denote differentiation in $f$.

A direct computation yields the identity
\begin{align} \label{we2}
\Box_g ( \chi \cdot \phi ) + \sigma ( \chi \cdot \phi ) = \mathcal{F} \text{,}
\end{align}
with the right hand side $\mathcal{F}$ given by
\begin{align}
\mathcal{F} &= 2 \nabla^\alpha \chi \nabla_\alpha \phi + \phi \Box_g \chi + \chi \left(\Box_g \phi + \sigma \phi\right) \nonumber \\
&= \chi^\prime ( 2 \nabla^\alpha f \nabla_\alpha \phi + \Box_g f \cdot \phi ) + \chi^{\prime \prime} \nabla^\alpha f \nabla_\alpha f \cdot \phi + \chi ( \Box_g \phi + \sigma \phi ) \text{.}
\end{align}
Note that $\chi^\prime$ and $\chi^{\prime \prime}$ are both supported in $[ \frac{1}{2} f_0, \frac{3}{4} f_0 ]$ only.
Thus, recalling \eqref{eq.f_grad}, \eqref{eq.f_grad_sq}, \eqref{eq.error_est}, \eqref{we}, and the smallness of $f$, we compute
\begin{equation} \label{int_ext}
|\mathcal{F}|^2 \begin{cases} \lesssim_{n, y, f_0} ( \rho^2 | \nabla_\rho \phi |^2 + \rho^2 | \nabla_t \phi |^2 + \rho^{2+p} | \nasla \phi |^2 + |\phi|^2 ) & \frac{f_0}{2} \leq f \leq \frac{3f_0}{4} \text{,} \\ \lesssim \rho^p ( \rho^4 | \nabla_t \phi |^2 + \rho^4 | \nabla_\rho \phi |^2 + \rho^2 | \nasla \phi |^2 + \rho^{2p} | \phi |^2 ) & 0 \leq f \leq \frac{f_0}{2} \text{.} \end{cases}
\end{equation}

Recall the definition of the region $\Omega_{f_0, \rho_0}$ from \eqref{eq.f0_rho0}, and define the subregions
\[ \Omega_{f_0, \rho_0}^{int} = \Omega_{f_0, \rho_0} \cap \left\{ f < \frac{f_0}{2} \right\} \text{,} \qquad \Omega_{f_0, \rho_0}^{ext} = \Omega_{f_0, \rho_0} \cap \left\{ \frac{f_0}{2} < f < \frac{3 f_0}{4} \right\} \text{,} \]
for sufficiently small $\rho_0 > 0$.
Since the function $\bar{\phi} := \chi \phi$ vanishes near $\{ f = f_0 \}$ and satisfies \eqref{we2}, we can apply \eqref{eq.carleman_lambda} of Theorem \ref{thm.carleman} for $\bar{\phi}$, with $\kappa$ chosen as in \eqref{bckappa}.
The left-hand side $L$ of \eqref{eq.carleman_lambda} can then be estimated
\begin{equation} \label{lhsd} \begin{split}
L &\lesssim \int_{ \Omega_{f_0, \rho_0}^{ext} } f^{n - 2 - 2 \kappa} e^\frac{ - 2 \lambda f^p }{p} f^{-p} | \mathcal{F} |^2 + \int_{ \Omega_{f_0, \rho_0}^{int} } f^{n - 2 - 2 \kappa} e^\frac{ - 2 \lambda f^p }{p} f^{-p} | \mathcal{F} |^2 \\
\qquad &\qquad + \lambda ( \lambda^2 + | \sigma | ) \int_{ \{ \rho = \rho_0 \} } [ | \nabla_t ( \rho^{ - \kappa } \bar{\phi} ) |^2 + | \nabla_\rho ( \rho^{ - \kappa } \bar{\phi} ) |^2 + | \rho^{- \kappa - 1} \bar{\phi} |^2 ] d \mathring{g} \\
&:= L_1 + L_2 + L_3 \text{.}
\end{split} \end{equation}
For the right-hand side $R$ of \eqref{eq.carleman_lambda}, we recall \eqref{eq.f_trivial} to estimate
\begin{equation} \label{rhsd} \begin{split}
R &\gtrsim_{n, y, p, K} \lambda^3 \int_{ \Omega_{f_0, \rho_0}^{int} } f^{n - 2 - 2 \kappa} e^\frac{ - 2 \lambda f^p }{p} \rho^{2p} | \phi |^2 \\
&\qquad + \lambda \int_{ \Omega_{f_0,\rho_0}^{int} } f^{n - 2 - 2 \kappa} e^\frac{ - 2 \lambda f^p }{p} ( \rho^4 | \nabla_t \phi |^2 + \rho^4 | \nabla_\rho \phi |^2 + \rho^2 | \nasla \psi |^2 ) \\
&:= R_1 + R_2 \text{.}
\end{split} \end{equation}
Here, we used that all terms on the right hand side of \eqref{eq.carleman_lambda} are non-negative by our choice of $\kappa$.
Moreover, we restricted all integrals to $\Omega^{int}_{f_0, \rho_0}$, where $\bar{\phi} = \phi$.

By \eqref{lhsd} and \eqref{rhsd}, we can write \eqref{eq.carleman_lambda} as
\[ L_1 + L_2 + L_3 \gtrsim_{n, y, p, K} R_1 + R_2 \text{.} \]
Using the second part of \eqref{int_ext} and the bound $| f^{-1} \rho | \leq 1$ from \eqref{eq.f_trivial}, we can then absorb $L_2$ into $R_1 + R_2$ for $\lambda$ large (independently of $\rho_0$!).
Thus, for large $\lambda$,
\begin{equation}
\label{LR_intermediate} L_1 + L_3 \gtrsim_{n, y, p, K} R_1 + R_2 \text{.}
\end{equation}
Next, we take the limit $\rho_0 \searrow 0$, and we claim that $L_3 \rightarrow 0$ in this limit.\footnote{Note it is important that the absorption step occurs before this limit $\rho_0 \searrow 0$. In particular, because of the weight $f^{n - 2 - 2 \kappa - p}$, it is not a priori clear that the limit of $L_2$ as $\rho_0 \searrow 0$ is finite. In contrast, such an issue does not arise with $L_1$, since $f$ is bounded from below on $\Omega_{f_0, \rho_0}^{ext}$.}
To see this, we express $L_3$ in terms of $\phi$.
Noting in particular from \eqref{eq.f_deriv} that
\[ | \partial_\rho \chi | + | \partial_t \chi | \lesssim_{ n, y, f_0 } \rho^{-1} \text{,} \]
then the vanishing assumption \eqref{bca} indeed implies the claim.

From the above argument, we can conclude that
\begin{equation} \label{LR_interm_2} \begin{split}
&\int_{ \left\{ \frac{f_0}{2} < f < \frac{3 f_0}{4} \right\} } f^{n - 2 - 2 \kappa} e^\frac{ - 2 \lambda f^p }{p} \left( |\phi|^2 + \rho^2 | \nabla_\rho \phi|^2 + \rho^2 | \nabla_t \phi|^2  + \rho^{2+p} | \slashed{\nabla} \phi|^2 \right) \\
&\quad \gtrsim_{n, y, p, K, f_0} \lambda^3 \int_{ \left\{ f < \frac{ f_0 }{2} \right\} } f^{n - 2 - 2 \kappa} e^\frac{ - 2 \lambda f^p }{p} \rho^{2p} | \phi |^2 \\
&\quad \qquad + \lambda \int_{ \left\{ f < \frac{ f_0 }{2} \right\} } f^{n - 2 - 2 \kappa} e^\frac{ - 2 \lambda f^p }{p} ( \rho^4 | \nabla_t \phi |^2 + \rho^4 | \nabla_\rho \phi |^2 + \rho^2 | \nasla \psi |^2 ) \text{.}
\end{split} \end{equation}
We can now pull out the $f$-weights in \eqref{LR_interm_2}: since $n - 2 - 2 \kappa \leq 0$ by assumption,
\[ f^{n-2-2\kappa} e^\frac{ - 2 \lambda f^p }{p} \begin{cases} \leq \left( \frac{f_0}{2} \right)^{n - 2 - 2 \kappa} e^\frac{ - 2 \lambda \left( \frac{f_0}{2} \right)^p }{p} & \frac{f_0}{2} \leq f \leq \frac{3 f_0}{4} \text{,} \\ \geq \left( \frac{f_0}{2} \right)^{n - 2 - 2 \kappa} e^\frac{ - 2 \lambda \left( \frac{f_0}{2} \right)^p }{p} & f < \frac{f_0}{2} \end{cases} \text{.} \]
As a result, we obtain that (recall that $\lambda$ is large)
\begin{equation} \label{LR_interm_3} \begin{split}
&\int_{ \left\{ \frac{f_0}{2} < f < \frac{3 f_0}{4} \right\} } \left( |\phi|^2 + \rho^2 | \nabla_\rho \phi|^2 + \rho^2 | \nabla_t \phi|^2+ \rho^{2+p} | \slashed{\nabla} \phi|^2  \right) \\
&\quad \gtrsim \lambda \int_{ \left\{ f < \frac{ f_0 }{2} \right\} } ( \rho^4 | \nabla_t \phi |^2 + \rho^4 | \nabla_\rho \phi |^2 + \rho^2 | \nasla \psi |^2 + \rho^{2p}|\phi|^2 ) \text{.}
\end{split} \end{equation}

Clearly, boundedness of the left-hand side of \eqref{LR_interm_3} will imply that $\phi$ vanishes on $\{ f < f_0 / 2 \}$ by taking the limit $\lambda \nearrow \infty$.
The former is in turn implied directly by \eqref{ficoa} for the $| \nasla \phi |^2$-term, and can be deduced from \eqref{bca} for the remaining terms.
For instance, for the $\nabla_t \phi$-term, \eqref{bca} implies that
\begin{align*}
\infty &> \int_0^{r_0^{-1}} \int_{ \{ \rho = \rho' \} } \bar{\rho}^{-1 + \varepsilon} \bar{\rho}^{-2 \kappa} | \nabla_t \phi |^2 d \mathring{g} d \bar{\rho} \gtrsim \int_{ \mc{M} } \rho^{n + 1} \rho^{-1 + \varepsilon} \rho^{-2 \kappa} | \nabla_t \phi |^2
\end{align*}
for any $\varepsilon > 0$.
Since $n - 2 \kappa + \varepsilon \leq 1 + \varepsilon$ by our assumptions on $\kappa$, this indeed yields
\[ \int_{ \mc{M} } \rho^2 | \nabla_t \phi |^2 < \infty \text{,} \]
as desired.
The remaining terms can be handled analogously.
\end{proof}

\subsection{The Infinite-Order Vanishing Theorem} \label{sec.proof_inf}

We saw in Theorem \ref{theo:ads} above that the order of vanishing for $\phi$ at the boundary ensuring unique continuation depends crucially on the mass $\sigma$: the smaller (more negative) $\sigma$ is, the larger $\kappa$, and hence the required order of vanishing, is.
The next theorem, which is a consequence of the second Carleman estimate \eqref{eq.carleman_beta}, states in particular that if we only assume\footnote{Recall that in the setting of Theorem \ref{theo:ads} the potential $V$ in $\Box_g \phi + V\phi=0$ is given by an exact mass term plus a decaying part which vanishes at a quantitative rate near the boundary.} $L^\infty$-boundedness for $V$ in $\Box_g \phi + V \phi = 0$, then infinite-order vanishing of $\phi$ near the boundary ensures unique continuation.

\begin{theorem} \label{theo:ads5}
Let $n \in \N$ and $y, r_0 > 0$.
Let $( \mc{I}_\tfr, \mathring{g} )$ be an $n$-dimensional segment of bounded static AdS infinity (see Definition \ref{def.aads_infinity}), and let $\mc{M} := (r_0, \infty) \times \mc{I}_\tfr$.
Suppose in addition that:
\begin{itemize}
\item $( \mc{M}, g )$ is an admissible aAdS segment, as described in Definition \ref{def.aads}.

\item The $\zeta$-pseudoconvex property (see Definition \ref{def.pseudoconvex_ass}) is satisfied for some function $\zeta = \mc{O} (1)$ on $\mc{M}$, with constant $K > 0$.
\end{itemize}
Moreover, suppose $\phi \in \Gamma \underline{T}^0_l \mathcal{M}$ satisfies the following:
\begin{enumerate}[(i)]
\item There exists $C > 0$ such that
\begin{equation} \label{we5}
| \Box_g \phi |^2 \leq C ( \rho^4 | \nabla_t \phi |^2 + \rho^4 | \nabla_\rho \phi |^2 + \rho^2 | \nasla \phi |^2 + | \phi |^2 ) \text{.}
\end{equation}

\item $\phi$ vanishes to infinite order on the boundary, i.e., for any $\kappa \gg n$, we have
\begin{equation}  \label{bca5}
\lim_{ \tilde{\rho} \searrow 0 } \int_{ \mc{M} \cap \{ \rho = \tilde{\rho} \} } [ | \nabla_t ( \rho^{ - \kappa } \phi ) |^2 + | \nabla_\rho ( \rho^{ - \kappa } \phi ) |^2  + | \rho^{- \kappa - 1} \phi |^2 ] d \mathring{g} = 0 \text{.}
\end{equation}

\item The following finiteness condition holds:
\begin{equation}
\int_{ \mathcal{M} } \rho^2 | \nasla \phi |^2 < \infty \text{.}
\end{equation}
\end{enumerate}
Then, there exists $0 < f_0 \ll_{n, y, K} 1$ such that $\phi \equiv 0$ in $\mathcal{M} \cap \{ f < f_0 / 2 \}$.
\end{theorem}

The proof of Theorem \ref{theo:ads5} is analogous to that of Theorem \ref{theo:ads}, except we now apply the Carleman estimate \eqref{eq.carleman_beta} instead of \eqref{eq.carleman_lambda}, and we now take at the end the limit $\kappa \nearrow \infty$ (as opposed to $\lambda \nearrow \infty$ in the proof of Theorem \ref{theo:ads}).

The assumption of infinite-order vanishing is also necessary unless further restrictions on the size of $V$ are being imposed.
Indeed, for any $\beta \in \R$, one has
\begin{equation} \label{eq.r_box_beta} \Box \rho^\beta = ( \beta^2 - n \beta ) \rho^\beta + O ( \rho^{ \beta + 2 } ) \text{.} \end{equation}
In particular, by letting $\beta \nearrow \infty$, we see that infinite-order vanishing is required in order to handle wave equations with arbitrary bounded potentials.\footnote{More generally, since $\Box \rho^\beta + \sigma \rho^\beta = ( \beta^2 - n \beta + \sigma ) \rho^\beta + O ( \rho^{ \beta + 2 } ) \text{,}$
we see that prescribing vanishing of order $\beta$ allows one to handle at best wave equations with bounded potentials which are not too large in the $L^\infty$-norm (with the size determined by $\beta^2 - n \beta + \sigma$).}

\subsection{Global Argument for Pure AdS} \label{sec:gloads}

Here, we briefly sketch a proof of Corollary \ref{thm.global}.
The main idea is to use the conformal equivalence of a portion of AdS spacetime with Minkowski spacetime, and to take advantage of the unique continuation results and related observations of \cite{alex_schl_shao:uc_inf}.
Recall in particular that AdS spacetime conformally embeds into half of the Einstein cylinder, $\R \times \Sph^n$, while Minkowski spacetime conformally embeds into a relatively compact ``triangular" region in the same Einstein cylinder; see Figure \ref{fig.mink_ads}.

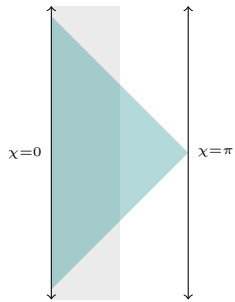
\begin{figure}
\centering
\begin{tikzpicture}[scale=1.3]
\fill[color=lightgray, opacity=0.3] (-0.7, -1.5) rectangle (0, 1.5);
\fill[color=teal, opacity=0.3] (-0.7, 1.4) -- (0.7, 0) -- (-0.7, -1.4) -- cycle;
\draw[<->] (-0.7, 1.5) -- (-0.7, -1.5);
\draw[<->] (0.7, 1.5) -- (0.7, -1.5);
\node[left] at (-0.7, 0) {$\scriptscriptstyle \chi = 0$};
\node[right] at (0.7, 0) {$\scriptscriptstyle \chi = \pi$};
\end{tikzpicture}
\caption{Part of Einstein cylinder $\R \times \Sph^n$, modulo spherical symmetry. The gray region indicates a conformally embedded AdS spacetime, while the teal region indicates a conformally embedded Minkowski spacetime.}
\label{fig.mink_ads}
\end{figure}

From the local unique continuation result, we know that $\phi$ and $d \phi$ vanishes on a region $\{ f < f_0 \}$ for some small $f_0$; this is the gray-shaded region in Figure \ref{fig.setup}.
Since the time interval used to generate $f$ has length greater than $\pi$, then this shaded region contains a point of both future and past null infinity in the conformal embedding of Minkowski spacetime; these points are denoted by $P_\pm$ in Figure \ref{fig.setup}.
Now, for some $\varepsilon > 0$, the level sets of the function
\[ \mf{f} = \frac{1}{4} ( r - t + 2 \varepsilon ) ( r + t + 2 \varepsilon ) \]
in \emph{Minkowski} spacetime form hyperboloids which terminate at $P_\pm$.

\begin{figure}
\centering
\begin{tikzpicture}[scale=1.5]
\fill[color=red, opacity=0.2] (-0.5, 2) -- (0.5, 1) -- (0.5, -1) -- (-0.5, -2) -- cycle;
\fill[color=gray, opacity=0.3] (0.5, 1.5) -- (0.5, -1.5) arc (250:110:0.7 and 1.6) -- cycle;
\draw[<->] (-0.5, 2) -- (-0.5, -2);
\node[left] at (-0.5, 0) {$\scriptscriptstyle r = 0$};
\draw[style=dashed, <->] (0.5, 2) -- (0.5, -2);
\node[right] at (0.5, 0) {$\scriptscriptstyle r = \infty$};
\draw (0.5, -1.5) arc (250:110:0.7 and 1.6);
\node at (0, 0) {$\scriptscriptstyle f = f_0$};
\draw[style=dashed, color=red] (-0.5, 2) -- (0.5, 1);
\draw[style=dashed, color=red] (-0.5, -2) -- (0.5, -1);
\node[draw, circle, fill, inner sep=0.7pt, color=purple] at (0.36, 1.14) {};
\node[draw, circle, fill, inner sep=0.7pt, color=purple] at (0.36, -1.14) {};
\node[above, right, color=purple] at (0.36, 1.14) {$\scriptscriptstyle P_+$};
\node[below, right, color=purple] at (0.36, -1.14) {$\scriptscriptstyle P_-$};
\node[right] at (0.5, 1.5) {$\scriptscriptstyle t = t_1$};
\node[right] at (0.5, -1.5) {$\scriptscriptstyle t = t_0$};

\fill[color=gray, opacity=0.3] (5.5, 1.5) -- (5.5, -1.5) arc (250:110:0.7 and 1.6) -- cycle;
\draw[<->] (4.5, 2) -- (4.5, -2);
\node[left] at (4.5, 0) {$\scriptscriptstyle r = 0$};
\draw[style=dashed, <->] (5.5, 2) -- (5.5, -2);
\node[right] at (5.5, 0) {$\scriptscriptstyle r = \infty$};
\draw (5.5, -1.5) arc (250:110:0.7 and 1.6);
\draw[style=dashed, color=red] (4.5, 2) -- (5.5, 1);
\draw[style=dashed, color=red] (4.5, -2) -- (5.5, -1);
\node[draw, circle, fill, inner sep=0.7pt, color=purple] at (5.36, 1.14) {};
\node[draw, circle, fill, inner sep=0.7pt, color=purple] at (5.36, -1.14) {};
\node[above, right, color=purple] at (5.36, 1.14) {$\scriptscriptstyle P_+$};
\node[below, right, color=purple] at (5.36, -1.14) {$\scriptscriptstyle P_-$};
\node[above, left, color=purple] at (5.1, 1) {$\scriptscriptstyle \mf{f} = c$};
\draw[color=purple] (5.36, -1.14) arc (250:110:1 and 1.23);
\node[color=blue] at (4.9, 0) {$U$};
\end{tikzpicture}
\caption{{\bf Left:} Part of AdS spacetime, with an embedded part of Minkowski spacetime (shaded red). The local unique continuation result applies to the gray-shaded region $f < f_0$. The purple $P_\pm$ are the terminal points of the pseudoconvex hyperboloids. {\bf Right:} Same setting, with purple line denoting a level set of $\mf{f}$. The region $U$ lies between the black and purple lines.}
\label{fig.setup}
\end{figure}
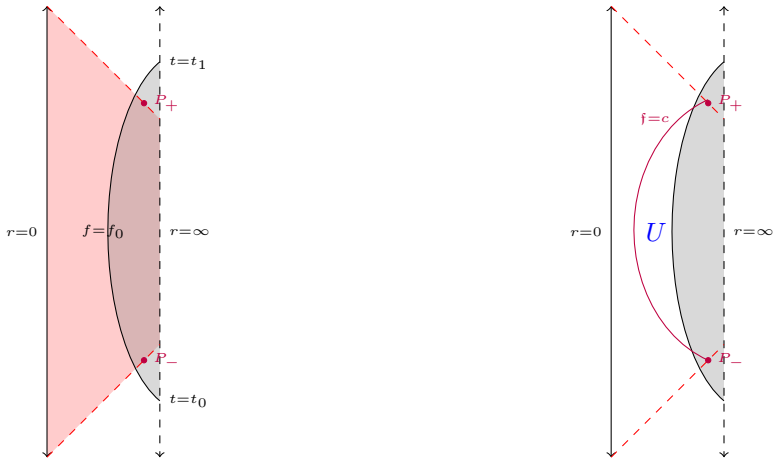

Recall from \cite[Section 3.1]{alex_schl_shao:uc_inf} that these level hyperboloids are strongly pseudoconvex, with this pseudoconvexity degenerating as one approaches $P_\pm$.
The idea is now to uniquely continue $\phi$ from the boundary $\{ f = f_0 \}$, where $\phi$ and $d \phi$ vanish, leftward to a level set $\{ \mf{f} = c \}$.
Consider the region $U$ bounded by these two hypersurfaces; see Figure \ref{fig.setup}.
We make the following observations:
\begin{itemize}
\item On $U$, the pseudoconvexity of the level sets of $\mf{f}$ are uniformly positive.
In other words, there is no degeneration of pseudoconvexity.

\item On $U$, the conformal factors associated with the embeddings of both AdS and Minkowski spacetimes into the Einstein cylinder are uniformly bounded from both above and below.

\item On $U$, the lower-order coefficients $a^\alpha$ and $V$ of the wave equation satisfied by $\phi$ are uniformly bounded.
\end{itemize}

As a result, classical uniqueness arguments (see, e.g., Proposition \ref{thm.uc_classical}) imply that $\phi$ can be uniquely continued into $U$, that is, $\phi$ vanishes identically on $U$.\footnote{This argument can be applied either directly to AdS spacetime or to the corresponding conformally related wave equation on Minkowski spacetime.}
Furthermore, the above observations, and hence the uniqueness argument, can be applied as long as the hyperboloid $\{ \mf{f} = c \}$ forming the left boundary of $U$ does not hit the line $\{ r = 0 \}$ in Minkowski spacetime.
Thus, by varying $c$, we obtain that $\phi$ and $d \phi$ vanishes in the dark gray region of Figure \ref{fig.globalize}.
Finally, resorting to the standard well-posedness theory (for Dirichlet boundary conditions), we obtain that $\phi$ indeed vanishes on all of AdS spacetime, as desired.

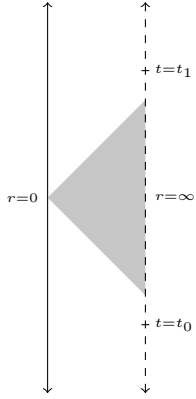
\begin{figure}
\centering
\begin{tikzpicture}[scale=1.3]
\fill[color=darkgray, opacity=0.3] (0.5, 1) -- (0.5, -1) -- (-0.5, 0) -- cycle;
\draw[<->] (-0.5, 2) -- (-0.5, -2);
\node[left] at (-0.5, 0) {$\scriptscriptstyle r = 0$};
\draw[style=dashed, <->] (0.5, 2) -- (0.5, -2);
\node[right] at (0.5, 0) {$\scriptscriptstyle r = \infty$};
\draw (0.46, 1.3) -- (0.54, 1.3);
\node[right] at (0.5, 1.3) {$\scriptscriptstyle t = t_1$};
\draw (0.46, -1.3) -- (0.54, -1.3);
\node[right] at (0.5, -1.3) {$\scriptscriptstyle t = t_0$};
\end{tikzpicture}
\caption{Classical uniqueness arguments using level sets of $\mf{f}$ show $\phi$ vanishes in the dark gray region. Well-posedness theory then implies $\phi$ vanishes everywhere in AdS spacetime.}
\label{fig.globalize}
\end{figure}

\section{Connections to the well-posedness theory} \label{sec.wp}

In this section, we deduce the assumptions of Theorem \ref{theo:ads} from natural assumptions in the context of the forward well-posedness theory for wave equations.
Recall that in this theory, given an arbitrary aAdS spacetime, one specifies initial data on a hypersurface of constant $t$ and seeks to contruct a solution near the conformal boundary (possibly depending on boundary conditions).
For simplicity, we restrict to the case where $\phi$ is scalar and satisfies the wave equation (\ref{freew}).
Generalizations to tensorial $\phi$ and more general wave equations are straightforward.

\begin{remark}
On the other hand, there is no established well-posedness theory for the wave equation $\Box_g \phi + V \phi = 0$ for general bounded potentials $V$.
As a result, there is no apparent restatement of Theorem \ref{theo:ads5} in terms of a well-posedness theory.
\end{remark}

Recall that in the context of the forward problem, one distinguishes three cases depending on the value of the mass $\sigma$:

\subsection{The range $n^2/4-1<\sigma<n^2/4$}

For such $\sigma$, which corresponds to the range $\frac{n}{2} - 1 < \beta_- < \frac{n}{2}$ and includes the conformally coupled case $\sigma = ( n^2 - 1 ) / 4$, i.e., $\beta_- = \frac{n - 1}{2}$, Warnick \cite{Warnick:2012fi} developed a general well-posedness theory for aAdS spacetimes based on propagation of the following energies on slices of constant $t$:
\begin{align*} 
E^{(1)}_{tw} [ \phi ] (\tau) &:= \int_{ \mc{M} \cap \{ t = \tau \} } \rho^{-1} \left\{ \rho^2 ( \partial_t \phi )^2 + \rho^{ 2 + 2 \beta_- } [ \partial_\rho ( \rho^{- \beta_-} \phi ) ]^2 + | \nasla \phi |^2 + \rho^2 \phi^2 \right\} \text{,} \\
E^{(2)}_{tw} [ \phi ] (\tau) &:= \int_{ \mc{M} \cap \{ t = \tau \} } \rho^{-1} \left\{ \rho^2 ( \partial_t \partial_t \phi )^2 + \rho^{ 2 + 2 \beta_- } [ \partial_\rho ( \rho^{-\beta_-} \partial_t \phi ) ]^2 + | \nasla \partial_t \phi|^2 \right. \\ 
&\qquad + \rho^{2+2\beta_-} | \rho^{-1} \nasla \partial_\rho ( \rho^{-\beta_-} \phi ) |^2  + | \rho^{-1} \nasla^2 \phi|^2 \\
&\qquad \left. + \rho^{ 2 n - 2 \beta_- } | \partial_\rho [ \rho^{2 \beta_- - n + 1} \partial_\rho ( \rho^{-\beta_-} \phi ) ] |^2 \right\} + E^{(1)}_{tw} [ \phi ] (t) \text{.}
\end{align*}
The above integrals are expressed in terms of the induced volume forms.\footnote{In pure AdS, integrals with respect to the induced metric can be stated more explicitly as
\[ \int_{ \{ t = \tau \} } \Phi \simeq \int_0^{ r_0^{-1} } \bar{\rho}^{-n} \int_{ \Sph^{n-1} } \Phi |_{ (t, \rho) = ( \tau, \bar{\rho} ) } d \mathring{\gamma} d \bar{\rho} \text{.} \]
An analogous expression holds for aAdS spacetimes.}

These renormalized energies (and their associated norms) are propagated by the forward evolution \emph{both} for the Dirichlet and the Neumann problem, that is, if one imposes \emph{either} of the two boundary conditions\footnote{Inhomogeneous versions of these conditions, as well as mixed Robin conditions, are also possible but omitted here for simplicity of the discussion.}
\begin{align} \label{bouc}
\rho^{-\beta_-} \phi \rightarrow 0 \quad \text{(Dirichlet),} \qquad \rho^{-2 + 2\beta_-} \partial_\rho ( \rho^{- \beta_-} \phi ) \rightarrow 0 \quad \text{(Neumann).}
\end{align}

\begin{remark}
This key observation goes back to the classical paper of Breitenlohner and Freedman \cite{BF} in the pure AdS case. The mathematical theory of ``twisted" Sobolev spaces and their associated energy estimates has been developed (and generalized to all aAdS spaces) by Warnick \cite{Warnick:2012fi}.
Note that the ``usual" energy arising from the timelike Killing field of AdS,
\[
E_{AdS}^{(1)} [\phi] (\tau) = \int_{ \{ t = \tau \} } \rho [ ( \partial_t \phi )^2 + ( \partial_\rho \phi )^2 + | \rho^{-1} \nasla \phi |^2 + \rho^{-2} \phi^2 ] \text{,}
\]
is only finite if homogenous Dirichlet conditions are imposed.
\end{remark}

The following theorem implies the informal Theorem \ref{thm.uc_rough_optimal2} as a special case:

\begin{theorem} \label{prop:deduce}
Consider a fixed $n$-dimensional aAdS spacetime $\left(\mathcal{M},g\right)$ as in Theorem \ref{theo:ads} and consider a solution of the scalar wave equation
\begin{align} \label{freew}
\Box_g \phi + \sigma \phi = V \phi \text{,} \qquad \textrm{with $V$ satisfying $|V| \leq C \rho^q$ for some $q > 0$,}
\end{align}
where $n^2 / 4 - 1 < \sigma < n^2 / 4$.
Assume that $\phi \in \mc{C}^\infty ( \mc{M} )$ satisfies:
\begin{itemize}
\item Dirichlet and Neumann conditions, i.e.,
\[
\lim_{ \tilde{\rho} \searrow 0 } \int_{ \mc{M} \cap \{ \rho = \tilde{\rho} \} } [ | \rho^{-\beta_-} \phi |^2 + | \nabla_t ( \rho^{ -\beta_- } \phi ) |^2 + | \rho^{-2 + 2 \beta_-} \nabla_\rho ( \rho^{ - \beta_- } \phi ) |^2 ] d \mathring{g} = 0 \text{.}
\]
\item $\sup\limits_{ t \in (0, y^{-1} \pi) } E^{(2)}_{tw} [ \phi ] (t) < \infty$ holds in $\mathcal{M}$.
\end{itemize}
Then, the assumptions (i)-(iii) in Theorem \ref{theo:ads} hold true.
\end{theorem}

\begin{remark}
Clearly, imposing \emph{both} Dirichlet \emph{and} Neumann conditions is a necessary condition for unique continuation to apply, because the well-posedness theory discussed above constructs (at least for potentials $V$ decaying like $|V| \leq C \rho^2$) a large class of nontrivial solutions of finite renormalized energy $E_{tw}^{(2)} [\phi]$ with only one of these conditions satisfied.
\end{remark}

\begin{proof}
Assumption (i) is immediately seen to be satisfied. Assumption (iii) already holds by the uniform boundedness of $E^{(1)}_{tw}\left[\phi\right]\left(t\right)$. Let 
\[
X = \phi \rho^{-\beta_-} \text{,} \qquad Y = \rho^{2 \beta_- - n + 1} \partial_\rho ( \phi \rho^{-\beta_-} ) \text{.}
\]

From the fundamental theorem of calculus, we derive, for any fixed $\rho_1 \in (0, r_0^{-1})$,
\begin{align*}
\int_{ \{ \rho = \rho_1 \} } |Y|^2 d \mathring{g} &= \int_0^{ \rho_1 } d \bar{\rho} \int_{ \{ \rho = \bar{\rho} \} } d \mathring{g} \partial_\rho ( | Y |^2 ) \\
&\leq 2 \sup_{ \bar{\rho} \in (0, \rho_1) } \sqrt{ \int_{ \{ \rho = \bar{\rho} \} } |Y|^2 d \mathring{g} } \int_0^{ \rho_1 } d \bar{\rho} \sqrt{ \int_{ \{ \rho = \bar{\rho} \} } | \partial_\rho Y |^2 d \mathring{g} } \text{,}
\end{align*}
and after another application of the Cauchy-Schwarz inequality,
\begin{equation} \label{eq.vanish_Y}
\sup_{ \bar{\rho} \in (0, \rho_1) } \int_{ \{ \rho = \bar{\rho} \} } | Y |^2 d \mathring{g} \leq 4 y^{-1} \pi \left[ \sup_{ t \in (0, y^{-1} \pi) } E^{(2)}_{tw} [ \phi ] (t) \right] \cdot \frac{ \rho_1^{2 \beta_- - n + 2} }{ 2 \beta_- - n + 2} \text{.}
\end{equation}
Recalling that $0 < 2 \beta_- - n + 2 < 2$, then dropping the supremum in \eqref{eq.vanish_Y} gives a quantitative vanishing condition for $Y$ as $\rho_1 \rightarrow 0$.

Repeating the above argument with $X$, using the bound for $Y = \rho^{2 \beta_- - n + 1} \partial_\rho X$ in the last step, yields also a quantitative estimate for $X$, namely,
\begin{equation} \label{eq.vanish_X}
\int_{ \{ \rho = \rho_1 \} } |X|^2 d \mathring{g} \lesssim_{\beta_-, y} \left[ \sup_{ t \in (0, y^{-1} \pi) } E^{(2)}_{tw} [ \phi ] (t) \right] \cdot \rho_1^{-2 \beta_- + n + 2} \text{.}
\end{equation}
Finally, repeating the estimate for $\partial_t X$ yields
\begin{align} \label{eq.vanish_Xt}
\int_{ \{ \rho = \rho_1 \} } | \partial_t X |^2 d \mathring{g} \lesssim_{\beta_-, y} \left[ \sup_{ t \in (0, y^{-1} \pi) } E^{(2)}_{tw} [ \phi ] (t) \right] \cdot \rho_1^{-2 \beta_- + n} \text{.}
\end{align}
It is now easy to see from \eqref{eq.vanish_Y}-\eqref{eq.vanish_Xt} that assumption (ii) in Theorem \ref{theo:ads} holds.  
\end{proof}

\subsection{The range $\sigma \leq n^2/4 - 1$} \label{sec:5/4comments}

In this range, the finiteness of $E^{(1)}_{tw} [ \phi ]$ already excludes the Neumann solution: the latter behaves like $\psi \sim \rho^{\beta_-}$ near the boundary, producing a divergence in $E^{(1)}_{tw} [ \phi ]$ for $\beta_- \leq \frac{n}{2} - 1$.
Consequently, there is no freedom to impose boundary conditions in the forward well-posedness theory of (\ref{freew}) when working with the energy $E^{(1)}_{tw} [\phi]$.

Turning to Theorem \ref{theo:ads}, we see that the vanishing assumption on $\phi$ ensures the necessary condition that the $\rho^{\beta_+}$-branch of the solution vanishes.\footnote{Solutions to the forward problem with $\rho^{-\beta_+} \phi$ having nontrivial trace on the boundary are easily constructed from \cite{Hol:wp}.}
The remaining conditions on $\partial_\rho ( \rho^{-\beta_+} \phi )$, $\partial_t ( \rho^{-\beta_+} \phi )$, and $\rho^{-1} \nasla ( \rho^{-\beta_+} \phi )$ in (ii) and (iii) are very mild, since one actually expects these quantities to extend continuously to the boundary for sufficiently regular solutions.

\subsection{The range $\sigma \geq \frac{n^2}{4}$} \label{sec:9/4comments}

There is no classical forward well-posedness theory for this range of $\sigma$.

\section{Application to the linearized Einstein equations on AdS} \label{sec.linear}

We conclude the paper with an application of our results to gravitational perturbations of AdS spacetime in $(3+1)$-dimensions. Recall that if one linearizes the Einstein equation near pure AdS, the linearized Bianchi equation takes (in view of the fact that pure AdS is conformally flat) the form
\begin{equation} \label{Weyl}
\nabla^a W_{abcd} = 0 \text{,}
\end{equation}
where $\nabla$ is the background connection of pure AdS, and where $W$ is a tensor having the symmetries and algebraic properties of a Weyl tensor (also called a Weyl field).

Using the standard orthonormal frame for the pure AdS metric (\ref{ads}),
\[
e_0 = \frac{1}{\sqrt{1+r^2}} \cdot \partial_t \text{,} \qquad e_{ \bar{r} } = \sqrt{1 + r^2} \cdot \partial_r \text{,} \qquad e_1 \text{,} \qquad e_2 \text{,}
\]
where $e_A$, $1 \leq A \leq 2$, is an orthonormal frame on the level spheres $S_{t, r}$ of $(t, r)$, we can decompose $W$ into its electric and magnetic part:
\[
E = W ( e_0, \cdot, e_0, \cdot ) \text{,} \qquad H = {}^\star W ( e_0, \cdot, e_0, \cdot ) \text{.}
\]
$E$ can then be further decomposed into horizontal fields on the $S_{t, r}$'s:
\begin{itemize}
\item A horizontal $2$-tensor $E_{A B}$.

\item A horizontal $1$-form $E_{\bar{r} A}$.

\item A scalar $E_{\bar{r} \bar{r}}$.
\end{itemize}
The dual magnetic part $H$ can be similarly decomposed.

Furthermore, we let the horizontal $2$-tensors $\hat{E}, \hat{H} \in \Gamma \ul{T}^0_2 \mc{M}_\AdS$ denote the traceless parts of $E$ and $H$ (with respect to the induced metrics $\slashed{g} = r^2 \mathring{\gamma}$ on the $S_{t, r}$'s),
\begin{equation} \label{eq.traceless}
\hat{E}_{AB} = E_{AB} - \frac{1}{2} \slashed{g}^{CD} E_{CD} \cdot \slashed{g}_{AB} \text{,}
\end{equation}
and analogously for $\hat{H}$.

Our main results then imply the following:

\begin{corollary} \label{cor:adscft}
Let $W$ be a smooth Weyl field satisfying the Bianchi equation (\ref{Weyl}) on pure AdS spacetime.
Suppose $W$ satisfies, on a segment $I$ of infinity of time length $T > \pi$, the vanishing condition
\begin{equation} \label{eq.W_vanish}
r^3 |W| + r^3|\nabla_t W| + r^3| \nabla_{\mathbb{S}^2}W| \rightarrow 0 \text{.}
\end{equation}
Then $W$ vanishes globally in the interior of the segment $I$. 
\end{corollary}

\begin{remark}
The vanishing assumption \eqref{eq.W_vanish} imposes that the suitably weighted tensor $W$ and its derivatives tangential to the boundary vanish.
Recall (see, for instance, \cite{adsdissipative}) that fixing the conformal class of the metric on the boundary to be (to linear order) that of AdS itself corresponds to the boundary condition
\[
|r^3 \hat{H}_{AB}| + |r^3 E_{A\bar{r}}| + |r^3 H_{\bar{r}\bar{r}}| \rightarrow 0 \text{,}
\]
while fixing the holographic stress energy tensor of a solution on the boundary to be zero corresponds to
\[
|r^3 \hat{E}_{AB}| + |r^3 H_{A\bar{r}}| + |r^3 E_{\bar{r}\bar{r}}| \rightarrow 0 \text{.}
\]
This justifies the nomenclature used in Section \ref{sec:tensorwave} of the introduction.
\end{remark}

\begin{proof}
From (\ref{Weyl}), one can see that the quantities $\Phi^\pm = \hat{E} \pm \hat{H}^\star$, where $\star$ denotes the Hodge star with respect to $\slashed{g}$, each satisfy a decoupled tensorial wave equation (in the sense of Section \ref{sec.aads_tensor}).\footnote{We use here the notation of \cite[Sect. 6.2]{adsdissipative}, where these decoupled equations are written out explicitly using spherical coordinates.} In fact, considering the (weighted) fields $r^2 \Phi^\pm$, we see that the resulting tensorial wave equations are of the form
\begin{equation} \label{apw}
\Box_g ( r^2 \Phi^\pm ) + 2 ( r^2 \Phi^\pm ) = \pm \frac{4}{ (1 + r^2) r } \nabla_t ( r^2 \Phi^\pm ) + V^\pm ( r^2 \Phi^\pm ) \text{,}
\end{equation}
where the potential $V^\pm$ decays like $r^{-2}$ near infinity (see also Section \ref{sec.linear_remark} below for further discussions on \eqref{apw} and its equivalent forms).

The tensorial equation (\ref{apw}) is of a form to which Theorem \ref{theo:ads} applies.
As shown in \cite{adsdissipative}, the assumption \eqref{eq.W_vanish} together with the Bianchi equations imply that $r^3 \Phi^+$ and $r^3 \Phi^-$ both satisfy the boundary conditions
\begin{equation} \label{bcfull}
| r^2 \nabla_r ( r^3 \Phi^\pm ) | \rightarrow 0 \text{,} \qquad | r^3 \Phi^\pm | \rightarrow 0 \text{.}
\end{equation}
Moreover, a calculus argument similar to that in the proof of Theorem \ref{prop:deduce} implies additional decay for $\Phi^\pm$, so that the vanishing assumptions of Theorem \ref{theo:ads} hold for $r^2 \Phi^\pm$.
Thus, by Theorem \ref{theo:ads} and Corollary \ref{thm.global}, we must have $\Phi^\pm = 0$ (and hence $\hat{E} = \hat{H} = 0$) in the interior of the segment $I$.

With $\hat{E}$ and $\hat{H}$ globally vanishing, we turn to the constraint equations arising from (\ref{Weyl}) on constant $t$-hypersurfaces in the interior of $I$.
From the vanishing of $\hat{E}$ and $\hat{H}$, one derives homogeneous second order elliptic equations for $E_{\bar{r}\bar{r}}$ and $H_{\bar{r}\bar{r}}$, with zero boundary conditions leading to the conclusion that $H_{\bar{r} \bar{r}} = 0$ and $E_{\bar{r} \bar{r}}=0$ globally for any hypersurface of constant $t$ in the interior of $I$.
A further integration of the constraint equations from the boundary with zero boundary data yields $H_{\bar{r}A}=0$ and $E_{\bar{r}A}=0$ as well.
\end{proof}

\subsection{A Remark on the Teukolsky Equations} \label{sec.linear_remark}

While the tensorial Teukolsky equations \eqref{apw} can be derived directly from \eqref{Weyl}, in physics literature, the equations are usually written in terms of complex scalar quantities.
Here, we demonstrate that \eqref{apw} is in fact equivalent to the scalar representation.

For this, we use the scalar complex representation given in \cite[Eq. (65)]{adsdissipative},
\begin{equation} \label{eq.teukolsky} \begin{split}
0 &= - \frac{r^2}{1 + r^2} \partial_t^2 \psi^\pm \pm \frac{4 r}{1 + r^2} \partial_t \psi^\pm + \frac{1 + r^2}{ r^3 } \partial_r \left\{ \frac{r^4}{1 + r^2} \partial_r [ r (1 + r^2) \psi^\pm ] \right\} \\
&\qquad + \frac{1}{ \sin \theta } \partial_\theta ( \sin \theta \cdot \partial_\theta \psi^\pm ) + \frac{1}{ \sin^2 \theta } \partial_\varphi^2 \psi^\pm - \frac{4 i \cos \theta}{ \sin^2 \theta } \partial_\varphi \psi^\pm \\
&\qquad - \left( \frac{4}{ \sin^2 \theta } - 2 \right) \psi^\pm \text{,}
\end{split} \end{equation}
where the scalars $\psi^\pm$ can be connected to the $\Phi^\pm$ in \eqref{apw} by the formulas
\begin{equation} \label{eq.teukonnection}
\psi^\pm = 2 ( \Phi^\mp ) ( e_1, e_1 ) - 2 i ( \Phi^\mp ) ( e_1, e_2 ) \text{.}
\end{equation}
Our objective below will be to \emph{assume \eqref{apw} and then derive \eqref{eq.teukolsky}.}

By \eqref{eq.mixed_deriv}, we can express $\Box_g \Phi^\pm$, contracted with respect to two frame elements $e_A$, $e_B$, in terms of the scalar wave operator on the function $\Phi^\pm ( e_A, e_B )$:
\begin{equation} \label{eq.tensor_wave} \begin{split}
\Box_g \Phi^\pm ( e_A, e_B ) &= \partial^\alpha [ \nabla_{ \partial_\alpha } \Phi^\pm ( e_A, e_B ) ] - \nabla_{ \nabla_{ \partial^\alpha } \partial_\alpha } \Phi^\pm ( e_A, e_B ) \\
&\qquad - \nabla_{ \partial_\alpha } \Phi^\pm ( \nasla_{ \partial^\alpha } e_A, e_B ) - \nabla_{ \partial_\alpha } \Phi^\pm ( e_A, \nasla_{ \partial^\alpha } e_B ) \\
&= \Box_g [ \Phi^\pm (e_A, e_B) ] - 2 \partial^\alpha [ \Phi^\pm ( \nasla_{ \partial_\alpha } e_A, e_B ) ] \\
&\qquad - 2 \partial^\alpha [ \Phi^\pm (e_A, \nasla_{ \partial_\alpha } e_B ) ] + \Phi^\pm ( \nasla_{ \nabla_{ \partial^\alpha } \partial_\alpha } e_A, e_B ) \\
&\qquad + \Phi^\pm (e_A, \nasla_{ \nabla_{ \partial^\alpha } \partial_\alpha } e_B ) + \Phi^\pm ( \nasla_{ \partial_\alpha } ( \nasla_{ \partial^\alpha } e_A ), e_B ) \\
&\qquad + \Phi^\pm ( e_A, \nasla_{ \partial_\alpha } ( \nasla_{ \partial^\alpha } e_B ) ) + 2 \Phi^\pm ( \nasla_{ \partial_\alpha } e_A, \nasla_{ \partial^\alpha } e_B ) \text{.}
\end{split} \end{equation}
On pure AdS, we can restrict our attention to the orthonormal frame
\begin{equation} \label{eq.ads_frame} e_1 := \frac{1}{r} \cdot \partial_\theta \text{,} \qquad e_2 := \frac{ 1 }{ r \sin \theta } \cdot \partial_\varphi \text{.} \end{equation}
Using \eqref{eq.ads_frame}, the explicit expression \eqref{eq.ads} for $g_\AdS$, and the fact that $\Phi_\pm$ is symmetric and trace-free, we can show that \eqref{eq.tensor_wave} reduces to
\begin{equation} \label{eq.tensor_wave_0} \begin{split}
( \Box_g \Phi^\pm )_{A B} &= - \frac{1}{ 1 + r^2 } \partial_t^2 ( \Phi^\pm_{A B} ) + r^{-2} \partial_r [ r^2 (1 + r^2) \cdot \partial_r ( \Phi^\pm_{A B} ) ] \\
&\qquad + \frac{1}{r^2 \sin \theta } \partial_\theta [ \sin \theta \cdot \partial_\theta ( \Phi^\pm_{A B} ) ] + \frac{1}{r^2 \sin^2 \theta } \partial_\varphi^2 ( \Phi^\pm_{A B} ) \\
&\qquad - \frac{ 4 \cos \theta }{ r^2 \sin^2 \theta } \partial_\varphi ( \Phi^{\pm \star}_{A B} ) + 4 r^{-2} ( 1 - \sin^{-2} \theta ) \cdot \Phi^\pm_{AB} \text{.}
\end{split} \end{equation}

The left-hand side of \eqref{eq.tensor_wave_0} can then be expanded using \eqref{apw}.
Applying the resulting equation, with $A = 1$ and $B = 1, 2$, and recalling \eqref{eq.teukonnection} results in \eqref{eq.teukolsky}.

\raggedright
\bibliographystyle{amsplain}
\bibliography{AdS}

\providecommand{\bysame}{\leavevmode\hbox to3em{\hrulefill}\thinspace}
\providecommand{\MR}{\relax\ifhmode\unskip\space\fi MR }
\providecommand{\MRhref}[2]{%
  \href{http://www.ams.org/mathscinet-getitem?mr=#1}{#2}
}
\providecommand{\href}[2]{#2}
\begin{thebibliography}{10}

\bibitem{alex_schl_shao:uc_inf}
S.~Alexakis, V.~Schlue, and A.~Shao, \emph{Unique continuation from infinity
  for linear waves}, Adv. Math. \textbf{286} (2016), 481--544.

\bibitem{alex_shao:uc_global}
S.~Alexakis and A.~Shao, \emph{Global uniqueness theorems for linear and
  nonlinear waves}, J. Func. Anal. \textbf{269} (2015), no.~11, 3458--3499.

\bibitem{alin_baou:non_unique}
S.~Alinhac and M.~S. Baouendi, \emph{A non uniqueness result for operators of
  principal type}, Math. Z. \textbf{220} (1995), no.~1, 561--568.

\bibitem{Anderson}
Michael~T. Anderson, \emph{{On the uniqueness and global dynamics of AdS
  spacetimes}}, Class.Quant.Grav. \textbf{23} (2006), 6935--6954.

\bibitem{Bachelot}
Alain Bachelot, \emph{The {K}lein-{G}ordon equation in the anti-de {S}itter
  cosmology}, J. Math. Pures Appl. (9) \textbf{96} (2011), no.~6, 527--554.
  \MR{2851681}

\bibitem{BF}
Peter Breitenlohner and Daniel~Z. Freedman, \emph{{Stability in gauged extended
  supergravity}}, Annals Phys. \textbf{144} (1982), 249.

\bibitem{carl:uc_strong}
T.~Carleman, \emph{Sur un probl\`{e}me d'unicit\'e pour les syst\`{e}mes
  d\'equations aux d\'eriv\'ees partielles \`a deux variables ind\'ependentes},
  Ark. Mat., Astr. Fys. \textbf{26} (1939), no.~17, 1--9.

\bibitem{ChrKla}
D.~Christodoulou and S.~Klainerman, \emph{{The Global Nonlinear Stability of
  the Minkowski Space}}, Princeton NJ, 1993.

\bibitem{Enciso}
Alberto Enciso and Niky Kamran, \emph{{Lorentzian Einstein metrics with
  prescribed conformal infinity}}, arXiv:1412.4376, 2014.

\bibitem{EncisoK}
\bysame, \emph{A singular initial-boundary value problem for nonlinear wave
  equations and holography in asymptotically anti-de {S}itter spaces}, J. Math.
  Pures Appl. (9) \textbf{103} (2015), no.~4, 1053--1091. \MR{3318179}

\bibitem{Friedrich}
H.~Friedrich, \emph{{Einstein equations and conformal structure -- Existence of
  anti de Sitter type space-times}}, J.Geom.Phys. \textbf{17} (1995), 125--184.

\bibitem{Hadamard}
Jacques Hadamard, \emph{{Lectures on the Cauchy Problem in Linear Partial
  Differential Equations}}, Yale University Press, New Haven, 1923.

\bibitem{Hartnoll}
Sean~A. Hartnoll, \emph{{Lectures on holographic methods for condensed matter
  physics}}, Class. Quant. Grav. \textbf{26} (2009), 224002.

\bibitem{hol_shao:uc_ads_ns}
G.~Holzegel and A.~Shao, \emph{Unique continuation from infinity in
  asymptotically {Anti-de Sitter} spacetimes: {Non-static} boundaries}, in
  preparation, 2016.

\bibitem{ASfollowup}
\bysame, \emph{{Unique continuation in asymptotically Anti-de Sitter
  spacetimes: The Einstein equations}}, in preparation, 2016.

\bibitem{Hol:wp}
Gustav Holzegel, \emph{{Well-posedness for the massive wave equation on
  asymptotically anti-de Sitter spacetimes}}, J. Hyperbolic Differ. Equ.
  \textbf{9} (2012), 239--261.

\bibitem{adsdissipative}
Gustav Holzegel, Jonathan Luk, Jacques Smulevici, and Claude Warnick,
  \emph{{Asymptotic properties of linear field equations in anti-de Sitter
  space}}, arXiv:1502.04965, 2015.

\bibitem{hor:lpdo2}
L.~H{\"o}rmander, \emph{The analysis of linear partial differential operators
  {II:} {Differential} operators with constant coefficients}, Springer-Verlag,
  1985.

\bibitem{hor:lpdo4}
\bysame, \emph{The analysis of linear partial differential operators {IV:}
  {Fourier} integral operators}, Springer-Verlag, 1985.

\bibitem{hor:uc_interp}
L.~{H\"ormander}, \emph{On the uniqueness of the {Cauchy} problem under partial
  analyticity assumptions}, Geometric optics and related topics ({Cortona},
  1996), Progr. Nonlinear Differential Equations Appl., vol.~32, {Birkh\"auser
  Boston, Boston, MA}, 1997, pp.~179--219.

\bibitem{ken_ruiz_sog:sobolev_unique}
C.~E. Kenig, A.~Ruiz, and C.~D. Sogge, \emph{Uniform {Sobolev} inequalities and
  unique continuation for second order constant coefficient differential
  operators}, Duke Math. J. \textbf{55} (1987), no.~2, 329--347.

\bibitem{Malda}
Juan~Martin Maldacena, \emph{{The Large N limit of superconformal field
  theories and supergravity}}, Int. J. Theor. Phys. \textbf{38} (1999),
  1113--1133, [Adv. Theor. Math. Phys.2,231(1998)].

\bibitem{meti:counter_holmg}
G.~M\'etivier, \emph{Counterexamples to {H\"olmgren's} uniqueness for analytic
  non linear {Cauchy} problems}, Invent. Math. \textbf{112} (1993), no.~1,
  217--222.

\bibitem{Ralston}
James~V. Ralston, \emph{Solutions of the wave equation with localized energy},
  Comm. Pure Appl. Math. \textbf{22} (1969), 807--823. \MR{0254433 (40 \#7642)}

\bibitem{robb_zuil:uc_interp}
L.~Robbiani and C.~Zuily, \emph{Uniqueness in the {Cauchy} problem for
  operators with partially holomorphic coefficients}, Invent. Math.
  \textbf{131} (1998), no.~3, 493--539.

\bibitem{Sbierski:2013mva}
Jan Sbierski, \emph{{Characterisation of the Energy of Gaussian Beams on
  Lorentzian Manifolds - with Applications to Black Hole Spacetimes}},
  arXiv:1412.4376, 2013.

\bibitem{shao:bdc_nv}
A.~Shao, \emph{Breakdown criteria for nonvacuum {Einstein} equations}, Ph.D.
  thesis, Princeton University, 2010.

\bibitem{shao:ksp}
\bysame, \emph{A generalized representation formula for systems of tensor wave
  equations}, Commun. Math. Phys. \textbf{306} (2011), no.~1, 51--82.

\bibitem{shao:bdc_nvp}
\bysame, \emph{On breakdown criteria for nonvacuum {Einstein} equations},
  Annales Henri Poincar\'e \textbf{12} (2011), no.~2, 205--277.

\bibitem{tata:uc_interp}
D.~Tataru, \emph{Unique continuation for solutions to {PDE's}: between
  {H\"ormander's} theorem and {Holmgren's} theorem}, Comm. Partial Differential
  Equations \textbf{20} (1995), no.~5--6, 855--884.

\bibitem{Vasy}
Andr\'as Vasy, \emph{{The wave equation on asymptotically anti-de Sitter
  spaces}}, Anal. PDE \textbf{5} (2012), no.~1, 81--144.

\bibitem{Warnick:2012fi}
C.M. Warnick, \emph{{The Massive wave equation in asymptotically AdS
  spacetimes}}, Commun.Math.Phys. \textbf{321} (2013), 85--111.

\end{thebibliography}

\end{document}